\documentclass[11pt,letterpaper]{article}
\usepackage[letterpaper,margin=1in]{geometry}
\usepackage[final]{pdfpages}

\usepackage{bm} 
\usepackage{url}
\usepackage{cite}
\usepackage{amssymb}
\usepackage{amsmath}
\usepackage{amsfonts}
\usepackage{amsthm}
\usepackage{algorithmic}
\usepackage{graphicx}
\usepackage{textcomp}
\usepackage{xcolor}
\usepackage{enumerate}   
\def\BibTeX{{\rm B\kern-.05em{\sc i\kern-.025em b}\kern-.08em
    T\kern-.1667em\lower.7ex\hbox{E}\kern-.125emX}}

\usepackage{hyperref}
\usepackage[colorinlistoftodos]{todonotes}

\usepackage[ruled, noend]{algorithm2e}
\usepackage[shortlabels]{enumitem}

\usepackage{tikz}

\usepackage{microtype}

\usepackage{relsize}

\newcommand{\set}[1]{\{#1\}}                    
\newcommand{\setof}[2]{\{{#1}\mid{#2}\}}        
\newcommand{\bigsetof}[2]{\left\{{#1}\mid{#2}\right\}}        
\newcommand{\E}{\mathop{\mathbb E}}    
\newcommand{\dom}{\textsf{Dom}}

\newcommand{\one}{\bm 1}

\newcommand{\degree}{\texttt{deg}}

\newcommand{\calH}{\mathcal H}

\newtheorem{thm}{Theorem}[section]
\newtheorem{lmm}[thm]{Lemma}
\newtheorem{prop}[thm]{Proposition}
\newtheorem{cor}[thm]{Corollary}

\newtheorem{ex}[thm]{Example}

\newtheorem{defn}[thm]{Definition}

\newtheorem{claim}{Claim}

\newcommand{\defeq}{\stackrel{\text{def}}{=}}

\newcommand{\N}{\mathbb N} 
\newcommand{\R}{\mathbb R} 
\newcommand{\Rp}{{\mathbb R}_{\tiny +}} 

\newcommand{\floor}[1]{\lfloor{#1}\rfloor}
\newcommand{\ceil}[1]{\lceil{#1}\rceil}

\newcommand{\mvd}{\twoheadrightarrow}
\newcommand{\fd}{\rightarrow}

\newcommand{\conehull}{\mathbf{conhull}}
\newcommand{\bigjoin}{\mathlarger{\mathlarger{\mathlarger{\Join}}}}
\newcommand{\bigtimes}{\mathlarger{\mathlarger{\mathlarger{\times}}}}

\begin{document}

\title{Applications of Information Inequalities \\ to Database Theory Problems \thanks{Research supported in part by NSF IIS 1907997, NSF
    IIS 1954222, and NSF-BSF 2109922.}}

\author{Dan Suciu \\
  \textit{University of Washington}\\
  Seattle, USA \\
  suciu@cs.washington.edu}

\maketitle

\begin{abstract}
  The paper describes several applications of information inequalities to problems in database theory.  The problems discussed include: upper bounds of a query's output, worst-case optimal join algorithms, the query domination problem, and the implication problem for approximate integrity constraints.  The paper is self-contained: all required concepts and results from information inequalities are introduced here, gradually, and motivated by database problems.
\end{abstract}

\section{Introduction}

\label{sec:intro}

Notions and techniques from information theory have found a number of
uses in various areas of database theory.  For example, entropy and
mutual information have been used to characterize database
dependencies \cite{DBLP:journals/tse/Lee87,DBLP:journals/tse/Lee87a}
and normal forms in relational and XML databases
\cite{DBLP:conf/pods/ArenasL02,DBLP:journals/jacm/ArenasL05}. More
recently, information inequalities were used with much success to
obtain tight bounds on the size of the output of a query on a given
database
\cite{DBLP:journals/siamcomp/AtseriasGM13,DBLP:journals/jacm/GottlobLVV12,
  DBLP:journals/talg/GroheM14,DBLP:conf/pods/KhamisNS16,DBLP:conf/pods/Khamis0S17},
{\em and} to devise query plans for worst-case optimal join
algorithms~\cite{DBLP:conf/pods/KhamisNS16,DBLP:conf/pods/Khamis0S17}.
Information theory was also used to compare the sizes of the outputs
of two queries, or, equivalently, to check query containment under bag
semantics~\cite{DBLP:journals/ejc/KoppartyR11,DBLP:journals/tods/KhamisKNS21}.
Finally, information theory has been used to reason about approximate
integrity constraints in the
data~\cite{DBLP:journals/lmcs/KenigS22,DBLP:conf/sigmod/KenigMPSS20}.

This paper presents some of these recent applications of information
theory to databases, in a unified framework.  All applications
discussed here make use of information inequalities, which have been
intensively studied in the information theory
community~\cite{Yeung:2008:ITN:1457455,DBLP:journals/tit/ZhangY97,zhang1998characterization,DBLP:conf/isit/Matus07,DBLP:journals/tit/KacedR13}.
We will introduce gradually the concepts and results on information
inequalities, motivating them with database applications.

We start by presenting in Sec.~\ref{sec:agm} a celebrated result in
database theory: the AGM upper bound, which gives a tight upper bound
on the query output size, given the cardinalities of the input
relations.  The AGM bound was first introduced by Grohe and
Marx~\cite{DBLP:conf/soda/GroheM06}, and refined in its current form
by Atserias, Grohe, and
Marx~\cite{DBLP:journals/siamcomp/AtseriasGM13}, hence the name AGM.
(A related result appeared earlier in~\cite{friedgut-kahn-1998}.)
While the original papers already used information inequalities to
prove these bounds, in this paper we provide an alternative,
elementary proof, which is based on a family of inequalities due to
Friedgut~\cite{DBLP:journals/tamm/Friedgut04}, and which are of
independent interest.

Next, we turn our attention in Sec.~\ref{sec:bound} to an extension of
the AGM bound, by providing an upper bound on the size of the query's
output using functional dependencies and statistics on degrees, in
addition to cardinality statistics.  The extension to functional
dependencies was first studied by Gottlob et
al.~\cite{DBLP:journals/jacm/GottlobLVV12} and then by Khamis et
al.~\cite{DBLP:conf/pods/KhamisNS16}, while the general framework was
introduced by Khamis et al.~\cite{DBLP:conf/pods/Khamis0S17}.  Here,
information inequalities are a necessity, and we use this opportunity
to introduce entropic vectors and polymatroids, and to define
information inequalities.  We show simple examples of how to compute
upper bounds on the query's output size using Shannon inequalities
(monotonicity and submodularity, reviewed in
Sec.~\ref{sec:background:ii}).

A natural question is whether the upper bound on the query's output
size provided by information inequalities is tight: we discuss this in
Sec.~\ref{sec:lower:bound}.  This question is surprisingly subtle, and
it requires us to dig even deeper into information theory, and discuss
non-Shannon inequalities.  More than 30 years ago,
Pippenger~\cite{pippenger1986} asserted that constraints on entropies
are the ``{\em laws of information theory}'' and asked whether the
basic Shannon inequalities form the complete laws of information
theory, i.e., whether every constraint on entropies can be derived
from the Shannon's basic inequalities.  In a celebrated result
published in 1998, Zhang and Yeung~\cite{zhang1998characterization}
answered Pippenger's question negatively by finding a linear
inequality that is satisfied by all entropic functions with 4
variables, but cannot be derived from Shannon's inequalities.  Later,
Mat{\'{u}}s~\cite{DBLP:conf/isit/Matus07} proved that, for 4 variables
or more, there are infinitely many, independent non-Shannon
inequalities.  In fact, it is an open problem whether the validity of
an information inequality is decidable.  We provide here a short,
self-contained proof of Zhang and Yeung's result.  This result has a
direct consequence to our problem, computing an upper bound on the
query's output size: we prove that Shannon inequalities are
insufficient to compute a tight upper bound.  In contrast, we show
that the upper bound derived by using general information inequalities
is tight, a result related to one by Gogacz and
Torunczyk~\cite{DBLP:conf/icdt/GogaczT17} (for cardinality constraints
and functional dependencies only) and another one by Khamis et
al.~\cite{DBLP:conf/pods/Khamis0S17} (for general degree constraints).
The take-away of this section is that we have two upper bounds on the
query's output size: one that uses Shannon inequalities, which is
computable but not always tight, and another one that uses general
information inequalities, which is tight but whose computability is an
open problem.

This motivates us to look in Sec.~\ref{sec:special:cases} at a special
case, when the two bounds coincide and, thus, are both tight and
computable.  This special case is when the statistics are restricted
to cardinalities, and to degrees on a single variable.  We call the
corresponding class of information inequalities {\em simple
  inequalities}, and prove that they are valid for all entropic
vectors iff they are provable using Shannon inequalities.  Moreover,
in this special case, the worst-case database instances (where the
size of the query's output reaches the theoretical upper bound) have a
simple yet interesting structure, called {\em normal database
  instances}, which generalize the {\em product database instances}
that are the worst case instances for the AGM bound.

In Sec.~\ref{sec:query:evaluation} we turn to the most exciting
application of upper bounds to the query's output size: the design of
Worst Case Optimal Join, WCOJ, algorithms, which compute a query in a
time that does not exceed the upper bound on their output size.  Thus,
a WCOJ algorithm is {\em worst-case optimal}.  The vast majority of
database systems today compute a conjunctive query as a sequence of
binary joins, whose intermediate results may exceed the upper bound on
the final output size.  Therefore, database execution engines are not
WCOJ algorithms.  For that reason, the discovery of the first WCOJ
algorithm by Ngo, Porat, R{\'{e}}, and
Rudra~\cite{DBLP:conf/pods/NgoPRR12,DBLP:journals/jacm/NgoPRR18} was a
highly celebrated result.  While the original WCOJ algorithm was
complex, some of the same authors described a very simple WCOJ, called
Generic Join (GJ) in~\cite{DBLP:journals/sigmod/NgoRR13}, which,
together with its refinement Leapfrog Trie Join
(LFTJ)~\cite{DBLP:conf/icdt/Veldhuizen14} forms the basis of the few
implementations to
date~\cite{DBLP:conf/sigmod/SchleichOC16,DBLP:journals/pvldb/FreitagBSKN20,DBLP:journals/tods/MhedhbiKS21,DBLP:journals/corr/abs-2301-10841}.
Looking back at these results, we observe that any concrete WCOJ
algorithm also provides a proof of the upper bound of the query's
output size, since the size of the output cannot exceed the
runtime of the algorithm.  A WCOJ algorithm can be designed in
reverse: start from a {\em proof} of the upper bound, then convert
that proof into a WCOJ algorithm.  We call this paradigm {\em from
  proofs to algorithms}, and illustrate it on three different proof
systems for information inequalities: we derive GJ, an algorithm we
call Heavy/Light, and PANDA.

Next, in Sec.~\ref{sec:domination} we move beyond upper bounds, and
consider a related problem: given two queries, check whether the size
of the output of the second query is always greater than or equal to
that of the first query.  This problem, called the {\em query
  domination problem}, is equivalent to the query containment problem
under bag semantics.  The latter was introduced by Chaudhuri and
Vardi~\cite{DBLP:conf/pods/ChaudhuriV93}, is motivated by the
semantics of SQL, where queries return duplicates, hence the answer to
a query is a bag rather than a set.  The query containment problem is:
given two queries, interpreted under bag semantics, check whether the
output of the first query is always contained in that of the second
query.  It has been shown that the containment problem is undecidable
for {\em unions} of conjunctive
queries~\cite{DBLP:journals/tods/IoannidisR95} and for conjunctive
queries with {\em inequalities}~\cite{DBLP:conf/pods/JayramKV06}, by
reduction from Hilbert's 10th problem.  However, it remains an open
problem to date whether the containment of two conjunctive queries is
decidable.  We describe in this section a surprising finding by
Kopparty and Rossman~\cite{DBLP:journals/ejc/KoppartyR11}, who have
reduced the containment problem to information inequalities.  This
result was further extended in~\cite{DBLP:journals/tods/KhamisKNS21},
and it was shown that the containment problem under bag semantics is
computationally equivalent to information inequalities with $\max$,
which are inequalities that assert that the maximum of a finite number
of linear expressions is $\geq 0$.  The decidability of either of
these problems remains open to date.

Finally, we present in Sec.~\ref{sec:conditional:inequalities}
another, quite distinct application of information inequalities:
reasoning about approximate integrity constraints.  The implication
problem for integrity constraints asks whether a set of integrity
constraints logically implies some other constraint: this is a problem
in Logic, and consists of checking the validity of a sentence
$\bigwedge_i \sigma_i \Rightarrow \sigma$.  When the integrity
constraints can be captured by some information measures, such as is
the case for Functional Dependencies and Multivalued Dependencies,
then an implication can be described as a {\em conditional information
  inequality}.  The problem we study is whether the exact implication
problem can be {\em relaxed} to an inequality between these
information measures, $\sum_i h(\sigma_i) \geq h(\sigma)$.  We review
a result from~\cite{DBLP:journals/lmcs/KenigS22} stating that every
exact implication between FDs and MVDs relaxes to an inequality.
However, in a surprising result, Kaced and
Romashchenko~\cite{DBLP:journals/tit/KacedR13} have given examples of
conditional information inequalities that do {\em not} relax.  In
other words, the exact implication holds, but the tiniest violation of
an integrity constraint in the premise may cause arbitrarily large
violation of the integrity constraint in the consequence.  Yet in
another turn, ~\cite{DBLP:journals/lmcs/KenigS22} show that {\em
  every} conditional information inequality relaxes with some error
term, which can be made arbitrarily small, at the cost of increasing
the coefficients of the terms representing the premise.  In
particular, every conditional inequality could be derived from an
unconditioned inequality, by having the error term tend to zero, since
in the conditional inequality the premise is assumed to be zero, hence
the magnitudes of their coefficients do not matter.  This section leads
us to our deepest dive into the space of entropic vectors and almost
entropic vectors: we show that the set of entropic vectors is neither
convex nor a cone, that its topological closure is a convex cone,
called the set of {\em almost entropic functions}, and use the theory
of closed convex cones to prove the relaxation-with-error theorem.

{\bf Acknowledgments} I am deeply indebted to my collaborators,
especially Hung Q. Ngo who introduced me to applications of
information inequalities to databases and with whom I had wonderful
collaborations, and also to Mahmoud Abo Khamis, Batya Kenig, and
Phokion G. Kolaitis.  I also thank Dan Olteanu and Andrei Romashchenko
for commenting on an early version of this paper.

\section{Basic Notations}

\label{sec:problem}

For two natural numbers $M, N$ we denote by
$[M:N] \defeq \set{M, M+1, \ldots, N}$; when $M=1$ we abbreviate
$[1:N]$ by $N$.  We will use upper case $X, Y, Z$ for variable names,
and lower case $x, y, z$ for values of these variables.  We use
boldface for tuples of variables, e.g. $\bm X, \bm Y$, or tuples of
values, e.g. $\bm x, \bm y$.

A {\bf conjunctive query}, CQ, is an expression of the form:
\begin{align}
  Q(\bm Y_0) = & \exists \bm Z (R_1(\bm Y_1) \wedge \cdots \wedge R_m(\bm Y_m)) \label{eq:cq}
\end{align}
Each $R_j(\bm Y_j)$ is called an {\em atom}: $R_j$ is a relation name,
and $\bm Y_j$ are variables.  We refer to $\bm Y_j$ interchangeably as
the {\em variables} of $R_j$, or the {\em attributes} of $R_j$.  The
variables $\bm Z$ are called {\em existential variables}, while
$\bm Y_0$ are called {\em head variables}.  We denote by $n$ the total
number of variables in the query, and by
$\bm X = \set{X_1, \ldots, X_n}$ the set of these variables.  Thus
$\bm X= \bm Y_0 \cup \bm Z$, and $\bm Y_j \subseteq \bm X$,
$\forall j$.

Fix some infinite domain $\dom$.  If $\bm X$ is a set of variables,
then we write $\dom^{\bm X}$ for the set of $\bm X$-tuples.  A {\em
  database instance} is $\bm D = (R_1^D, \ldots, R_m^D)$, where, for
each $j=1,m$, $R_j^D\subseteq \dom^{\bm Y_j}$, where $\bm Y_j$ are the
attributes of $R_j$.  Unless otherwise stated, relations are assumed
to be finite.  When $\bm D$ is clear from the context, then we will
drop the superscript and write simply $R_j$ for the instance $R_j^D$,
for $j=1,m$.

We denote by $Q(\bm D) \subseteq \dom^{\bm Y_0}$ the output, or answer
to the query $Q$ on the database $D$.  The {\em query evaluation
  problem} is: given a database instance $\bm D$, compute the output
$Q(\bm D)$.  The design and analysis of efficient query evaluation
algorithms is a fundamental problem in database systems and database
theory.  For the complexity of the query evaluation problem, we
consider only the {\em data complexity}, where $Q$ is fixed, and the
complexity is a function of the input database $\bm D$.

For a simple illustration, consider:
\begin{align}
  Q(X) = & \exists Y \exists Z(R(X,Y)\wedge S(Y,Z) \wedge T(Z,X)) \label{eq:cq:ex}
\end{align}
$Q$ returns all nodes $x$ that belong to an $RST$ triangle.

A {\em Boolean conjunctive query} is a conjunctive query with no head
variables.
At the other extreme, a {\em full conjunctive query} is a query with
no existential variables. For example, the query:
\begin{align}
  Q(X,Y,Z) = & R(X,Y) \wedge S(Y,Z) \wedge T(Z,X) \label{eq:cq:triangle}
\end{align}
is a full CQ computing all triangles formed by the relations
$R, S, T$.  Full conjunctive queries are of special importance because
they often occur as intermediate expressions during query evaluation.
Unless otherwise stated, we will assume in this paper that the query
is a full conjunctive query {\em without self-joins}, meaning that the
relation names of the atoms $R_1, R_2, \ldots$ are distinct.  Such a
query is also called a natural join of the relations
$R_1, \ldots, R_m$.

Fix a relation $R(\bm X)$, with $n$ attributes.  A {\bf functional
  dependency}, or FD, is an expression $\bm U \rightarrow \bm V$,
where $\bm U, \bm V \subseteq \bm X$.  An instance $R^D$ {\em
  satisfies} the FD, and we write
$R^D \models \bm U \rightarrow \bm V$, if for any two tuples
$\bm x_1, \bm x_2 \in R^D$, $\bm x_1.\bm U = \bm x_2.\bm U$ implies
$\bm x_1.\bm V = \bm x_2.\bm V$.  A set of functional dependencies
$\Sigma$ {\em implies} a functional dependency $\bm U \fd \bm V$, in
notation $\Sigma \models \bm U \fd \bm V$, if, for every instance
$R^D$, if $R^D \models \Sigma$ then
$R^D \models \bm U \rightarrow \bm V$.  Armstrong's
axioms~\cite{DBLP:journals/tods/ArmstrongD80} form a complete
axiomatization of the implication problem for FDs.  The {\em closure}
of $\bm U \subseteq \bm X$, denoted $\bm U^+$, is the set of all
attributes $X$ s.t.  $\Sigma \models \bm U \fd X$.  The closure can be
computed in polynomial time in the size of $\bm U$ and $\Sigma$.  A
set $\bm U$ is {\em closed} if $\bm U^+ = \bm U$.  A {\em super-key}
for $R(\bm X)$ is a set $\bm U$ with the property that
$\bm U^+= \bm X$, and a {\em key} is a minimal set of attributes that
is a superkey.

A finite {\bf lattice} is a partially ordered set $(L,\preceq)$ where
every two elements $x,y \in L$ have a least upper bound $x \vee y$,
and a greatest lower bound $x \wedge y$.  In particular the lattice
has a smallest and a largest element, usually denoted by
$\hat 0, \hat 1$.  Consider now a set of variables $\bm X$, and a set
of functional dependencies, $\Sigma$, over $\bm X$.  We denote by
$(L_{\Sigma}, \subseteq)$ the lattice consisting of the closed sets,
$L_\Sigma = \setof{\bm U}{\bm U^+=\bm U}$.  One can verify that the
operations in this lattice are
$\bm U \wedge \bm V \defeq \bm U \cap \bm V$ and
$\bm U \vee \bm V \defeq (\bm U \cup \bm V)^+$.

The {\bf cartesian product} of two relations $R(\bm X), S(\bm Y)$ with
disjoint sets of attributes is the set
$R \times S \defeq \setof{(\bm x,\bm y)}{\bm x \in R, \bm y \in S}$
with attributes $\bm X \cup \bm Y$; its size is
$|R \times S| = |R|\cdot |S|$.  Fix a set of attributes $\bm X$, and
two $\bm X$-tuples $\bm x = (x_1, \ldots, x_n)$ and
$\bm x' = (x_1',\ldots, x_n')$.  Their {\bf domain product} is the
$\bm X$-tuple
$\bm x \otimes \bm x' \defeq ((x_1,x'_1), \ldots, (x_n,x'_n))$; thus,
the value of each attribute is a pair.
\begin{defn} \label{def:domain:product} The {\em domain product} of
  two relation instances $R$ and $S$, with the same set of attributes
  $\bm X$, is
  $R \otimes S \defeq \setof{\bm x \otimes \bm x'}{\bm x\in R, \bm
    x'\in S}$.
\end{defn}
We have $|R\otimes S| = |R| \cdot |S|$.  If
$\bm D_i = (R_1^{D_i}, \ldots, R_m^{D_i})$, $i=1,2$, are two database
instances over the same schema, then we define their domain product
$\bm D_1 \otimes \bm D_2$ as
$(R_1^{D_1}\otimes R_1^{D_2}, \ldots, R_m^{D_1} \otimes R_m^{D_2})$.
One can check that
$Q(\bm D_1 \otimes \bm D_2) = Q(\bm D_1)\otimes Q(\bm D_2)$ for any
conjunctive query $Q$.  The domain product should not be confused with
the cartesian product.  It was first introduced by
Fagin~\cite{DBLP:journals/jacm/Fagin82} (under the name {\em direct
  product}) to prove the existence of an Armstrong relation for
constraints defined by Horn clauses, and later used by Geiger and
Pearl~\cite{GeigerPearl1993} to prove that Conditional Independence
constraints on probability distributions also admit an Armstrong
relation.  The same construction appears under the name ``fibered
product'' in~\cite{DBLP:journals/ejc/KoppartyR11}.

\section{Warmup: the AGM Bound}

\label{sec:agm}

Consider a full conjunctive query:
\begin{align}
  Q(\bm X) = \bigwedge_{j=1,m} R_j(\bm Y_j) \label{eq:cq:full}
\end{align}
where $\bm X = \set{X_1, \ldots, X_n}$.  Assume we have a database
$\bm D$, and we know the cardinality of each relation $R_j^D$.  How
large could the query output be?  The answer is given by an elegant
result, initially formulated by Grohe and
Marx~\cite{DBLP:conf/soda/GroheM06} and later refined by Atserias,
Grohe, and Marx~\cite{DBLP:journals/siamcomp/AtseriasGM13}, and is
called today the AGM bound of the query $Q$.  To state this bound, we
first need to review the connection between conjunctive queries and
hypergraphs.

We associate $Q$ in~\eqref{eq:cq:full} with the hypergraph
$\calH = (\bm X, E)$, where $E=\set{\bm Y_1, \ldots, \bm Y_m}$.  In
other words, the nodes of the hypergraph are the variables, and its
hyperedges are the atoms of the query.  A {\em fractional edge cover}
of the hypergraph $\calH$ is a tuple of non-negative weights
$\bm w = (w_j)_{j=1,m}$, such that every variable $X_i$ is {\em
  covered}, meaning:
\begin{align}
  \forall i=1,n:&& \sum_{j: X_i \in \bm Y_j} w_j \geq  &1 \label{eq:edge:cover}
\end{align}
A fractional edge cover of the query $Q$ is a fractional edge cover of
its associated hypergraph.  The AGM bound is the following:

\begin{thm}[AGM Bound] \label{thm:agm} For any fractional edge cover
  $\bm w$ of the query~\eqref{eq:cq:full}, and every instance $\bm D$:
  \begin{align}
    |Q(\bm D)| \leq & \prod_{j=1,m} |R_j^D|^{w_j} \label{eq:agm}
  \end{align}
\end{thm}
To reduce clutter, we will often drop $\bm D$ from both $Q(\bm D)$ and
$R_j^D$, and write the bound simply as $|Q| \leq \prod_j
|R_j|^{w_j}$. 

Let $\bm B = (B_j)_{j=1,m}$ be a non-negative vector, representing the
cardinalities of the relations in the database.  We define:
\begin{align}
  AGM(Q,\bm B) \defeq & \min_{\bm w} \prod_{j=1,m} B_j^{w_j} \label{eq:agm:formula}
\end{align}
where $\bm w$ ranges over all fractional edge covers of the query's
hypergraph.  Then Theorem~\ref{thm:agm} can be restated as follows:
for every instance $\bm D$, if $|R_j^D| \leq B_j$ for $j=1,m$, then
$|Q(\bm D)| \leq AGM(Q,\bm B)$.  When $\bm B$ is clear from the
context, then we write the bound simply as $AGM(Q)$.

Before we prove the bound, we illustrate it with a classic example.


\begin{ex} \label{ex:triangle} Consider the triangle
  query~\eqref{eq:cq:triangle}, which we repeat here:
  $Q(X,Y,Z) = R(X,Y) \wedge S(Y,Z) \wedge T(Z,X)$.  Its associated
  hypergraph is a graph with three nodes $X,Y,Z$ and three edges
  forming a triangle.  A fractional edge cover is any non-negative
  tuple $(w_R,w_S,w_T)$ satisfying:
  \begin{align*}
&
                   \begin{array}{lllll}
                         \mbox{Cover $X$: } & w_R+&&w_T&\geq 1\\
                         \mbox{Cover $Y$: } & w_R+&w_S&&\geq 1\\
                         \mbox{Cover $Z$: } &&w_S+&w_T&\geq 1
                   \end{array}
  \end{align*}
  The inequality $|Q| \leq |R|^{w_R}\cdot |S|^{w_S} \cdot |T|^{w_T}$
  holds for every fractional edge cover.  Consider the following four
  fractional edge covers: $(0,1,1),(1,0,1),(1,1,0),(1/2,1/2,1/2)$:
  these are the extreme vertices of the edge-covering polytope.  It
  follows that the AGM bound in~\eqref{eq:agm:formula} is achieved at
  one of the four extreme vertices:
  \begin{align*}
    A&GM(Q) = \\
    & \min\left(|S|\cdot |T|,|R|\cdot |T|,|R|\cdot |S|,|R|^{1/2}\cdot|S|^{1/2}\cdot|T|^{1/2}\right)
  \end{align*}
  When $|R|=|S|=|T|=N$ then $AGM(Q) = N^{3/2}$.
\end{ex}

In the rest of this section we will prove the AGM
bound~\eqref{eq:agm}, then show that the bound is tight.

{\bf Friedgut's Inequalities} While the original proof of the AGM
bound used information inequalities, we postpone the discussion of
information inequalities until Sec.~\ref{sec:bound}, where we consider
more general statistics.  Instead, we give here a simple, elementary
proof, based on an elegant family of inequalities introduce by
Friedgut~\cite{DBLP:journals/tamm/Friedgut04}.

Fix a hypergraph $\calH = (\bm X, E)$.  Let $N > 0$ be a natural
number, and for each hyperedge $\bm Y_j \in E$, let
$r_j \in \R_+^{N^{|\bm Y_j|}}$ be a non-negative, multi-dimensional
vector with $|\bm Y_j|$ dimensions; we will refer to $r_j$ as a {\em
  tensor}.  In what follows, we denote by $\bm i$ a tuple
$\bm i = (i_1, \ldots, i_n) \in [N]^{\bm X}$, and by $\bm i_j$ its
projection on $\bm Y_j$.

\begin{thm}[Friedgut's
  Inequality]~\cite{DBLP:journals/tamm/Friedgut04} \label{th:friedgut}
  For every fractional edge cover $\bm w$ of the hypergraph $\calH$,
  the following holds:
  \begin{align}
    \sum_{\bm i} \prod_{j=1,m} r_j[\bm i_j] \leq & \prod_{j=1,m} \left(\sum_{\bm i_j}r_j[\bm i_j]^{\frac{1}{w_j}}\right)^{w_j} \label{eq:friedgut}
  \end{align}
\end{thm}

\begin{figure}
  \centering
  {\footnotesize
  \begin{align*}
    & \mbox{Cauchy-Schwartz:} & \sum_i a[i]\cdot b[i] \leq & \left(\sum_i a[i]^2\right)^{1/2}\cdot\left(\sum_i b[i]^2\right)^{1/2}&&\includegraphics[width=0.1\linewidth]{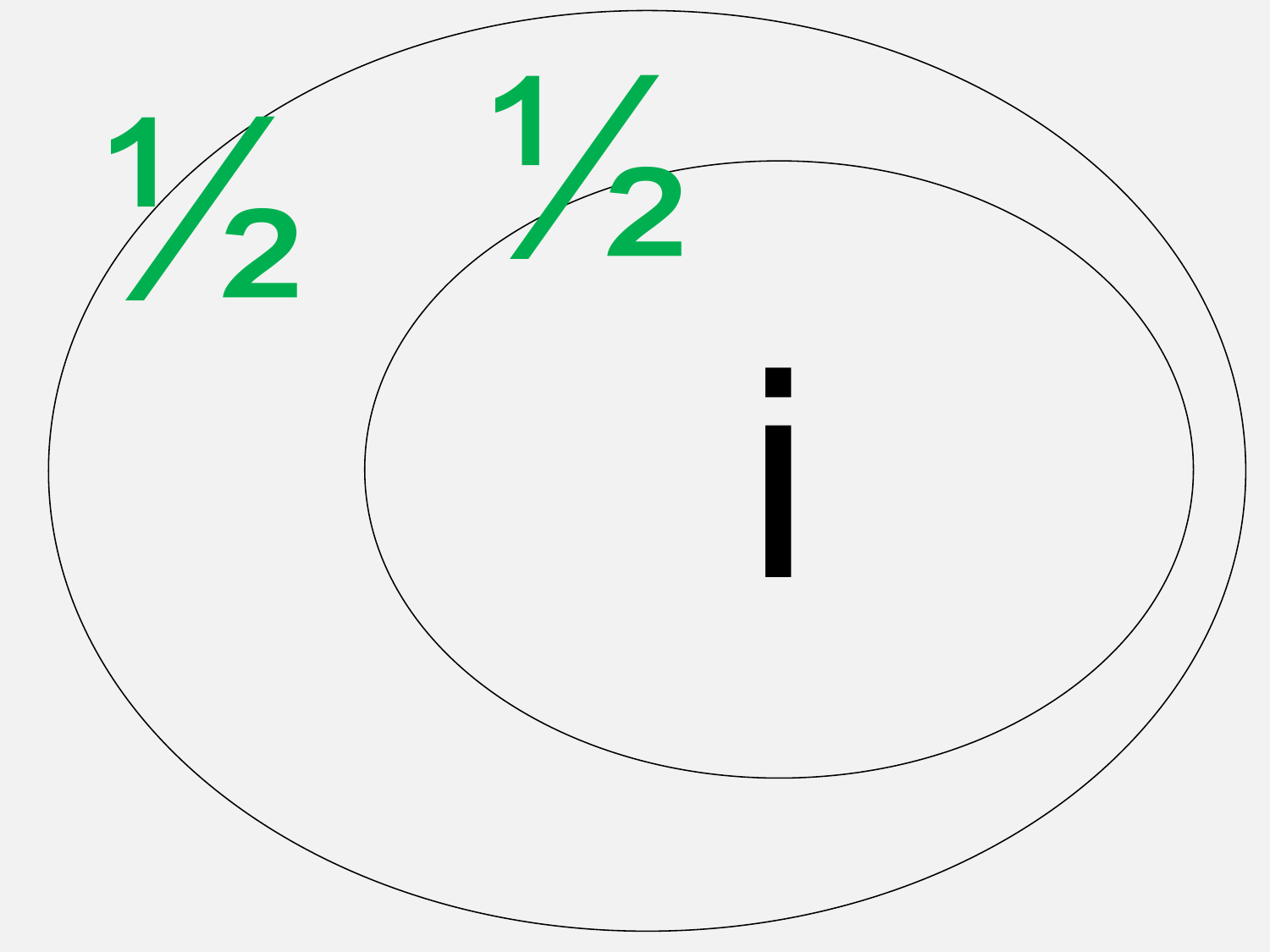}\\
    & \mbox{H\"older:} & \sum_i \prod_j a_j[i] \leq & \prod_j\left(\sum_i a_j[i]^{\frac{1}{w_j}}\right)^{w_j}\mbox{ when $\sum_j w_j \geq 1$}&&\includegraphics[width=0.1\linewidth]{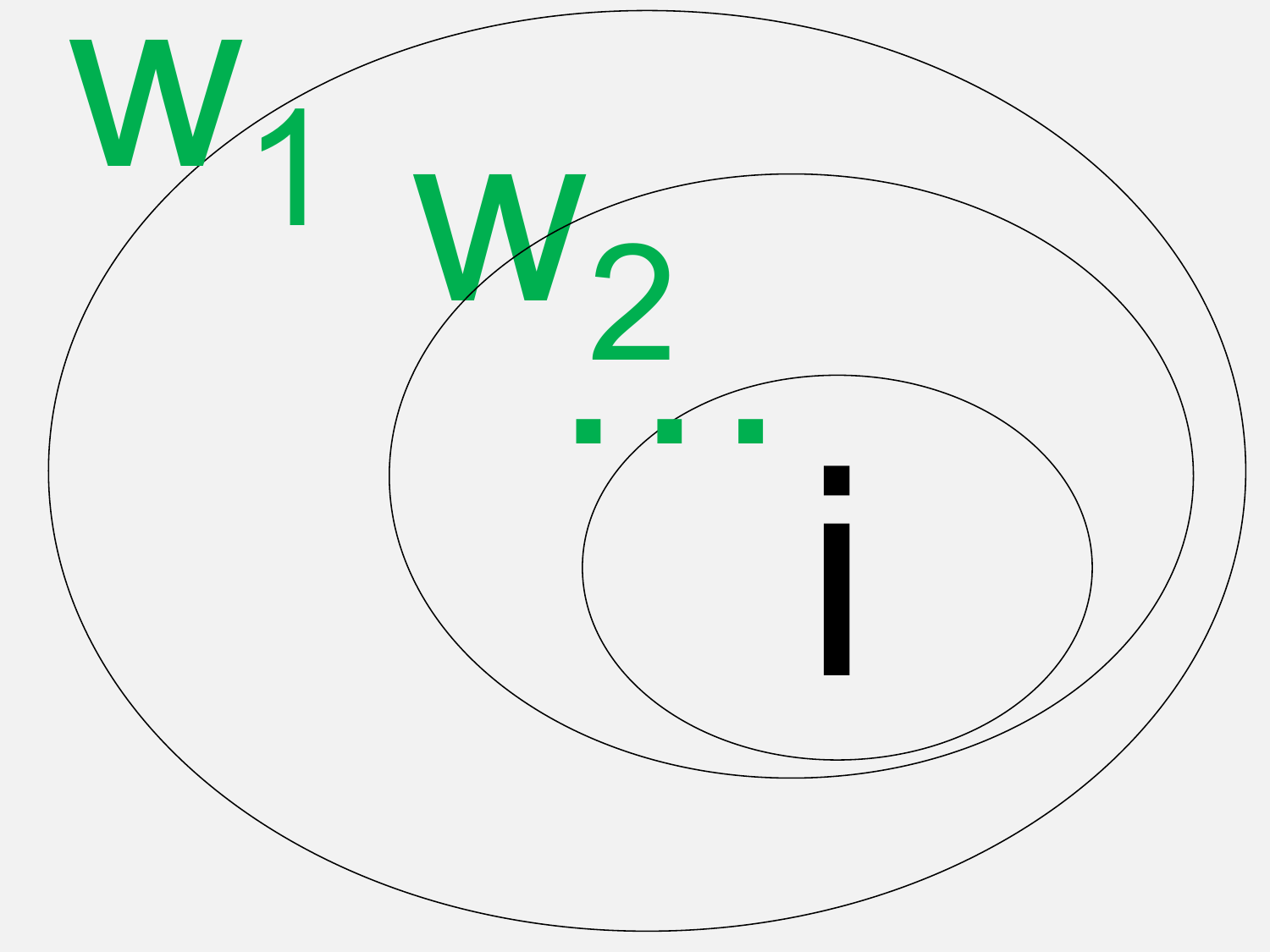}\\
    & \mbox{Friedgut:} & \sum_{i,j,k} a[i,j]\cdot b[j,k] \cdot c[k,i] \leq & \left(\sum_{i,j}a[i,j]^2\right)^{1/2}\cdot\left(\sum_{j,k}b[j,k]^2\right)^{1/2}\cdot\left(\sum_{k,i}c[k,i]^2\right)^{1/2}&&\includegraphics[width=0.1\linewidth]{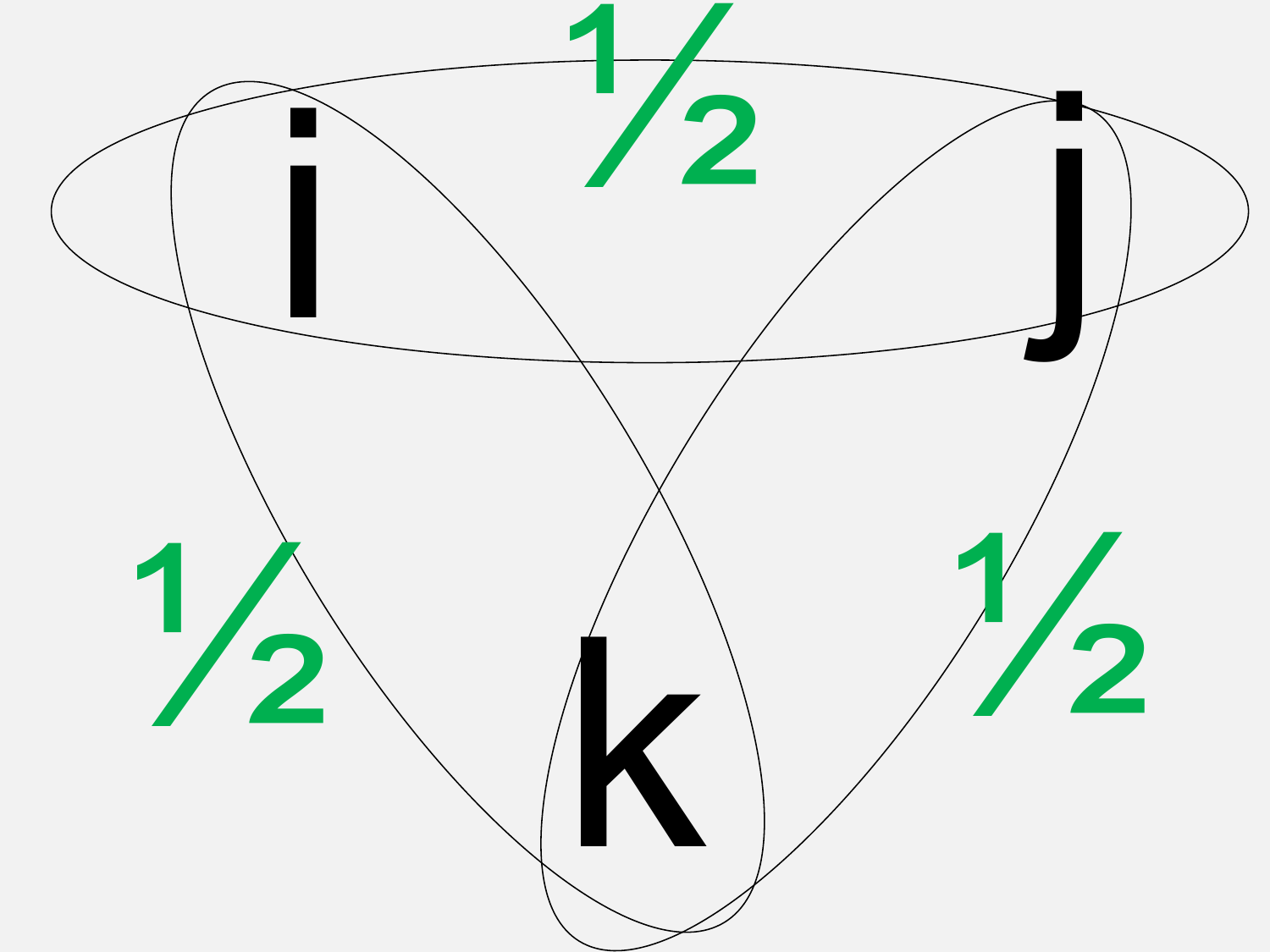}
  \end{align*}
  }
  \caption{Examples of Friedgut's inequalities~\eqref{eq:friedgut}.
    In each case we show the associated hypergraph on the right.}
  \label{fig:friedgut}
\end{figure}

Fig.~\ref{fig:friedgut} illustrates several instances
of~\eqref{eq:friedgut}.  We invite the reader to check that
Loomis–Whitney's inequality~\cite{zbMATH03054426} is also an instance
such an inequality.  Using Theorem~\ref{th:friedgut} we can prove the
AGM bound as follows.  Given a relational instance
$\bm D = (R_1^D, \ldots, R_m^D)$ define the following tensors:
\begin{align*}
r_j[x_1, \ldots, x_{a_j}] \defeq
&
  \begin{cases}
    1 & \mbox{if $(x_1, \ldots, x_{a_j}) \in R_j^D$} \\
    0 & \mbox{otherwise}
  \end{cases}
\end{align*}
Then the LHS of~\eqref{eq:friedgut} is $|Q(\bm D)|$ and the RHS is
$\prod_j |R_j|^{w_j}$.

\begin{proof}
  (of Theorem~\ref{th:friedgut}) While the original proof also used
  information inequalities, we give here a direct proof, by induction
  on the number $n$ of vertices of the hypergraph $\calH$.  (This
  proof generalizes Loomis–Whitney's proof in~\cite{zbMATH03054426}.)

  We replace each tensor expression $r_j[\bm i_j]$ with
  $(r_j[\bm i_j])^{w_j}$, then in order to prove~\eqref{eq:friedgut}
  it suffices to prove:
  \begin{align}
    \sum_{\bm i} \prod_j r_j[\bm i_j]^{w_j} \leq & \prod_j\left(\sum_{\bm i_j} r_j[\bm i_j]\right)^{w_j} \label{eq:friedgut:1}
  \end{align}
  We notice that the index variables $\bm i = (i_1, \ldots, i_n)$ used
  in the summation correspond one-to-one to the nodes of the
  hypergraph $\bm X=\set{X_1, \ldots, X_n}$, and the subset
  $\bm i_j$ contains the index variables corresponding to nodes in
  $\bm Y_j$.  We now prove~\eqref{eq:friedgut:1} by induction on
  $n$.
  
  When $n=1$ then this is H\"older's inequality (see
  Fig.~\ref{fig:friedgut}), whose proof can be found in textbooks.
  Assume $n > 1$ and consider the hypergraph $\calH'$ obtained by
  removing the last variable $X_n$: its nodes are
  $\set{X_1, \ldots, X_{n-1}}$ and its hyperedges are
  $\setof{\bm Y_j - \set{X_n}}{j=1,m}$.  The weights
  $w_1, \ldots, w_m$ continue to be a fractional edge cover for
  $\calH'$.  Group the LHS of Eq.~\eqref{eq:friedgut:1} by factoring
  out the sum over the variable $i_n$, and apply induction hypothesis
  to the summation over the other variables $i_1, \ldots, i_{n-1}$,
  which form the hypergraph $\calH'$:
  \begin{align*}
\sum_{i_n} & \left(\sum_{i_1,\ldots,i_{n-1}}\prod_j r_j[\bm i_j]^{w_j}\right)
\leq \sum_{i_n} \prod_j \left(\sum_{\bm i_j-\set{i_n}}r_j[\bm i_j]\right)^{w_j}
  \end{align*}
  We factor out the products that do not depend on the variable $i_n$,
  then use the fact that $\sum_{j: X_n \in \bm Y_j} w_j \geq 1$
  because $X_n$ is covered, and apply H\"older's inequality
  (Fig.~\ref{fig:friedgut}) with $a_j[i_n] \defeq r_j[\bm i_j]$.  The
  RHS of the expression above becomes:
  \begin{align*}
\prod_{j:X_n \not\in \bm Y_j}&\left(\sum_{\bm i_j}r_j[\bm i_j]\right)^{w_j}
\cdot \sum_{i_n} \prod_{j:X_n \in \bm Y_j}\left(\sum_{\bm i_j-\set{i_n}}r_j[\bm i_j]\right)^{w_j}\\
\leq & \prod_{j:X_n \not\in \bm Y_j}\left(\sum_{\bm i_j}r_j[\bm i_j]\right)^{w_j}
\cdot \prod_{j:X_n \in \bm Y_j}\left(\sum_{i_n}\sum_{\bm i_j-\set{i_n}}r_j[\bm i_j]\right)^{w_j}
  \end{align*}
  This is the RHS of~\eqref{eq:friedgut:1}, which completes the proof.
\end{proof}

{\bf The lower bound} How tight is the AGM bound?  One key insight
in~\cite{DBLP:journals/siamcomp/AtseriasGM13} is that, while the upper
bound is described by a linear program, a lower bound can be described
using the dual linear program: tightness follows from the strong
duality theorem for linear programs.  They proved:

\begin{thm}\label{th:agm:lower:bound} For any query $Q$ with $n$
  variables, and vector $\bm B$ there exists a database $\bm D$
  s.t. $|Q(\bm D)|\geq \frac{1}{2^n}AGM(Q,\bm B)$. We call such a
  database $\bm D$ a {\em worst-case instance}.
\end{thm}

\begin{proof}
  The logarithm of the AGM bound~\eqref{eq:agm:formula} is the optimal
  value of the following {\em primal linear program}:
  \begin{align*}
    \mbox{minimize }&& \sum_j w_j \log B_j\\
    \mbox{where }&& \forall i:\ \sum_{j: X_i \in \bm Y_j}w_j \geq 1\\
                    && \forall j: w_j \geq 0
  \end{align*}
  The {\em dual linear program} is:
  \begin{align*}
    \mbox{maximize }&&\sum_i v_i\\
    \mbox{where }&& \forall j:\ \sum_{i: X_i \in \bm Y_j}v_i \leq \log B_j\\
                    && \forall i: v_i \geq 0
  \end{align*}
  For any two feasible solutions $\bm w, \bm v$ of the primal and
  dual, {\em weak duality} holds:
  $\sum_j w_j \log B_j \geq \sum_i v_i$.  If $\bm w^*, \bm v^*$ are
  the optimal solutions, then the {\em strong duality} theorem states
  that these two expressions are equal, therefore:
  \begin{align}
    AGM(Q,\bm B) = & 2^{\sum_j w_j^*\log B_j} = 2^{\sum_i v_i^*}=\prod_i 2^{v_i^*} \label{eq:agm:primal}
  \end{align}
  If $\bm v$ is any dual solution, we construct the following database
  instance $\bm D$: for each variable $X_i$, define the domain
  $V_i \defeq [\floor{2^{v_i}}] = \set{1,2,\ldots,\floor{2^{v_i}}}$,
  and set $R_j^D \defeq \bigtimes_{i: X_i \in \bm Y_j} V_i$, for
  $j=1,m$.  We call $\bm D$ a {\em product database instance}, because
  each relation is a cartesian product.  $\bm D$ satisfies the
  cardinality constraints $|R_j^D| \leq B_j$ because
  \begin{align*}
    |R_j^D| =& \prod_{i: X_i \in \bm Y_j} \floor{2^{v_i}} \leq  2^{\sum_{i: X_i \in \bm Y_j}v_i} \leq 2^{\log B_j}= B_j
  \end{align*}
  Similarly, the output to the query is the product
  $Q(\bm D) = \bigtimes_i V_i$, and its size is
  $\prod \floor{2^{v_i}}$.  At optimality, when $\bm v = \bm v^*$,
  \begin{align}
    |Q(\bm D)| = \prod_i \floor{2^{v_i^*}} \label{eq:agm:dual}
  \end{align}
  Theorem~\ref{th:agm:lower:bound} follows from~\eqref{eq:agm:primal}
  and~\eqref{eq:agm:dual}, and observing that
  $\floor{2^{v_i^*}} \geq \frac{1}{2} 2^{v_i^*}$.  
\end{proof}

Thus, one could say that the AGM bound is tight up to a ``rounding
error''.  The original
paper~\cite{DBLP:journals/siamcomp/AtseriasGM13} provides an extensive
discussion on tightness and proves two facts.  First, they construct
arbitrarily large databases $\bm D$ where the AGM bound is tight
exactly.
Second, they
describe an example where the ratio between the lower and upper bound
can be arbitrarily close to $1/2^n$, as $n$ grow arbitrarily large;
despite this example, the AGM bound is considered to be tight for
practical purposes.


{\bf Discussion} The AGM bound is elegant in that it solves completely
the problem it set out to solve: find the tight upper bound when the
cardinalities of all relations are known, and nothing else is known.
However, the bound is limited, in that it cannot take advantage of
other statistics or constraints on the input data, which are often
available in practice.  For example, consider the join of two
relations, $Q(X,Y,Z) = R(X,Y) \wedge S(Y,Z)$, and assume that both
$|R|, |S| \leq B$.  The AGM bound is $|Q| \leq B^2$ (because the only
fractional edge cover is $(1,1)$), and the reader can check that this
is tight, i.e. there exists relations where $|R|=|S|=B$ and
$|Q| = B^2$.  But, in practice, joins are often key/foreign-key joins,
for example, $S.Y$ may be a key in $S$, and in that case $|Q| \leq B$,
because every tuple in $R$ joins with at most one tuple in $S$.
%
In order to account for additional information about the data, like
keys or constraints on degrees, we need to use a more powerful tool
than Friedgut's inequalities: information inequalities.

\section{Max-Degree Query Bounds}

\label{sec:bound}

We describe now the general framework for computing an upper bound for
the query output size, using information inequalities.  We will use
the cardinalities of the input relations (like in the AGM bound), keys
or, more generally, functional dependencies for individual relations,
and bounds on degrees, also called maximum frequencies, which
generalize keys and functional dependencies.  This section is based
largely
on~\cite{DBLP:journals/jacm/GottlobLVV12,DBLP:conf/pods/KhamisNS16,DBLP:conf/pods/Khamis0S17}.
We start with a short review of information inequalities.

\subsection{Background on Information Inequalities}

\label{sec:background:ii}

Consider a finite probability distribution $(D,p)$, where
$p : D \rightarrow [0,1]$, $\sum_{x \in D} p(x) = 1$.  We denote by
$X$ the random variable with outcomes in $D$, and define its {\em
  entropy} as:
\begin{align}
  H(X) \defeq & - \sum_{x \in D} p(x) \log p(x) \label{eq:h}
\end{align}
If $N \defeq |D|$, then $0 \leq H(X) \leq \log N$, the equality
$H(X)=0$ holds iff $X$ is deterministic (i.e.  $\exists x\in D$,
$p(x)=1$), and the equality $H(X) = \log N$ holds iff $X$ is uniformly
distributed (i.e. $p(x) = 1/N$ for all $x \in N$).


Consider now a finite probability distribution $(R, p)$, where
$R \subseteq \dom^{\bm X}$ is a non-empty, finite relation with $n$
attributes $\bm X = \set{X_1, \ldots, X_n}$.  We will always assume
w.l.o.g. that $R$ is the support of $p$, otherwise we just remove from
$R$ the tuples $\bm x$ where $p(\bm x)=0$.  For each
$\alpha \subseteq [n]$, define $X_\alpha \defeq (X_i)_{i \in \alpha}$
the joint random variable obtained as follows: draw randomly a tuple
$\bm x \in R$ with probability $p(\bm x)$, then return $\bm x_\alpha$.
We associate the probability space $(R, p)$ with a vector
$\bm h \in \Rp^{2^{[n]}}$, by defining
$\bm h_\alpha \defeq H(\bm X_\alpha)$, and call $\bm h$ an {\em
  entropic vector}.  For any vector $\bm h$ (entropic or not) we will
write $h(\bm X_\alpha)$ for $\bm h_\alpha$.  In other words, we will
blur the distinction between a vector in $\R_+^{2^{[n]}}$, a vector in
$\R_+^{2^{\bm X}}$, and a function $2^{\bm X} \rightarrow \R_+$.

In analyzing properties of queries, we often examine the entropic
vector derived from a uniform distribution.

\begin{defn} \label{def:uniform} The {\em uniform} probability space
  associated to a non-empty, finite relation
  $R \subseteq \dom^{\bm X}$ is $(R,p)$ where $p(\bm x) = 1/|R|$ for
  every tuple $\bm x \in R$.  We will call its entropic vector $\bm h$
  {\em uniform} and say that it is {\em associated} to $R$.
\end{defn}

If $\bm h$ is associated to $R$, then $h(\bm X) = \log |R|$, and
$h(\bm U) \leq \log |\Pi_{\bm U}(R)|$ for every subset
$\bm U \subseteq \bm X$.  Fig.~\ref{fig:parity} illustrates the
entropic vector $h$ associated to a relational instance $R$ with
attributes $X,Y,Z$; we call $h$ the {\em parity function}, because the
relation $R$ contains all triples $(x,y,z)$ that have an even number
of $1$.

\begin{figure}\centering

  \begin{minipage}[b]{0.45\linewidth}
    {
    \begin{align*}
      &
        \begin{array}{|c|c|c|c}\cline{1-3}
          X&Y&Z&\\ \cline{1-3}
          0&0&0&p=1/4 \\
          0&1&1&p=1/4 \\
          1&0&1&p=1/4 \\
          1&1&0&p=1/4 \\ \cline{1-3}
        \end{array}
    \end{align*}
    }
  \end{minipage}
  \hfill
  \includegraphics[width=0.45\linewidth]{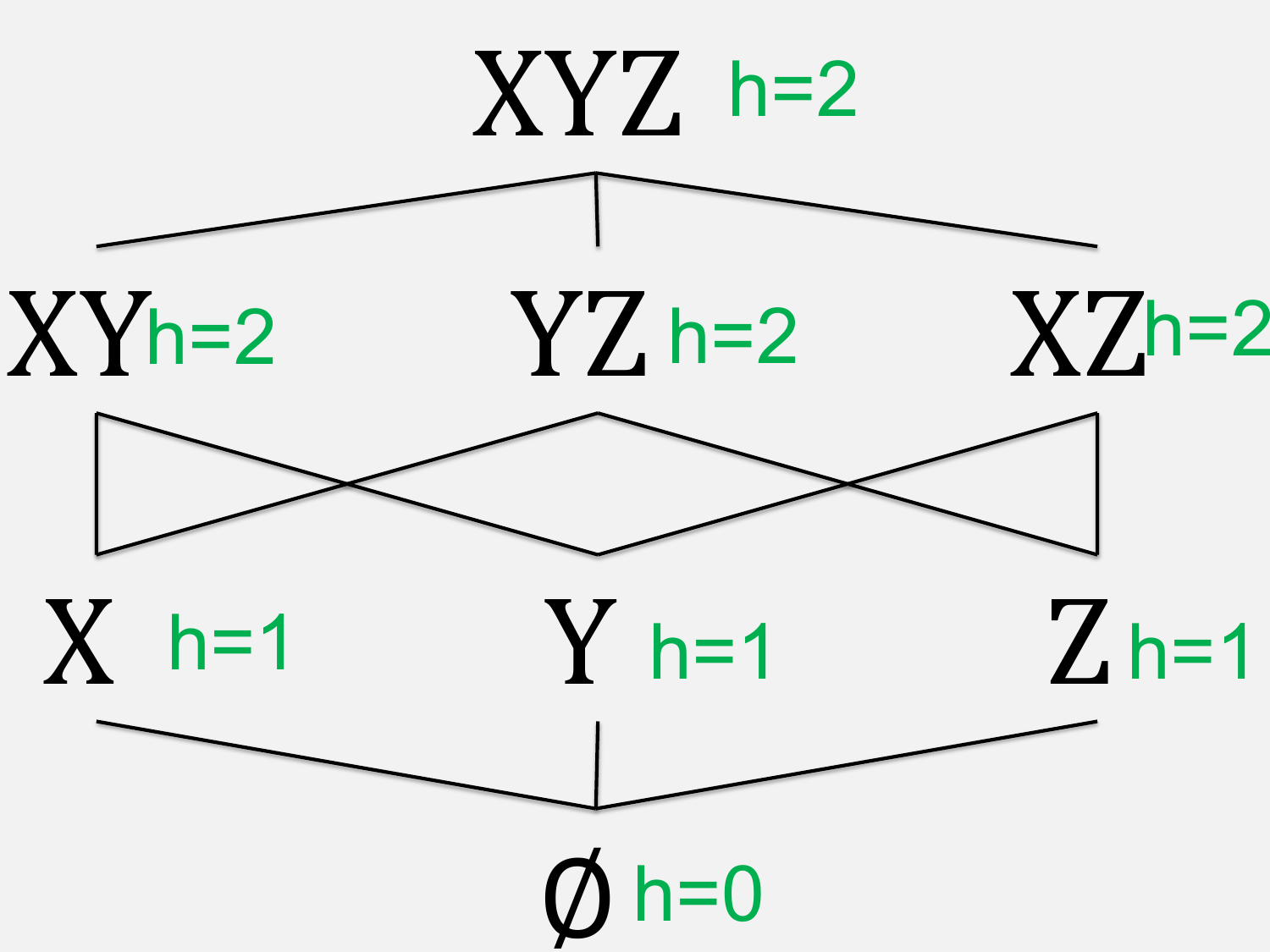}
  \caption{A relation defining the parity entropy $\bm h$.  The
    marginal distribution of $X$ is $p(X=0)=p(X=1)=1/2$, hence its
    entropy is $h(X)=1$, and similarly for the others values.}
  \label{fig:parity}
\end{figure}
Any entropic vector $\bm h \in \R_+^{2^{\bm X}}$ satisfies the
following {\em basic Shannon inequalities}:
\begin{align}
  h(\emptyset) = & 0 \label{eq:emptyset:zero}\\
  h(\bm U\cup \bm V) \geq & h(\bm U)\label{eq:monotonicity}\\
  h(\bm U) + h(\bm V) \geq & h(\bm U \cup \bm V) + h(\bm U \cap \bm V)\label{eq:submodularity}
\end{align}
The last two inequalities are called called {\em monotonicity} and
{\em submodularity} respectively.  A {\em Shannon inequality} is a
positive linear combination of basic Shannon inequalities.

Any vector $h : 2^{\bm X} \rightarrow \R_+$ that satisfies the basic
Shannon inequalities is called a {\em polymatroid}.  The set of
entropic vectors is denoted by $\Gamma_n^*$ and the set of
polymatroids is denoted by $\Gamma_n$, where $n=|\bm X|$ is the number
of variables.  The following holds:
$\Gamma_n^* \subsetneq \Gamma_n \subsetneq \Rp^{2^{[n]}}$.  In
particular, not every polymatroid is entropic, as we will see shortly
(in Fig.~\ref{fig:zhang:yeung:h}).  Fig.~\ref{fig:diagram} represents
these two sets, as well as other sets, defined later in this paper.
In some of the literature the entropic vectors and the polymatroids
are defined as $(2^n-1)$-dimensional vectors, by dropping the
$\emptyset$-dimension, because, in that case, both sets $\Gamma_n^*$
and $\Gamma_n$ have a non-empty interior.  We prefer to use $2^n$
dimensions since this simplifies most of the discussion, and postpone
dealing with the non-empty interior until
Section~\ref{sec:proof:of:duality}.

\begin{figure}
  \begin{minipage}[b]{0.45\linewidth}
    \begin{tabular}{ll}
    $\Gamma_n$: & polymatroids \\
    $\bar\Gamma_n^*$: & almost-entropic \\
      $\Gamma_n^*$: & entropic \\
      $\Upsilon_n$: & group realizable \\
    $N_n$: & normal polymatroids \\
    $M_n$: & modular polymatroids
    \end{tabular}
  \end{minipage}
  \hfill
  \begin{minipage}[b]{0.5\linewidth}
    \includegraphics[width=1.0\linewidth]{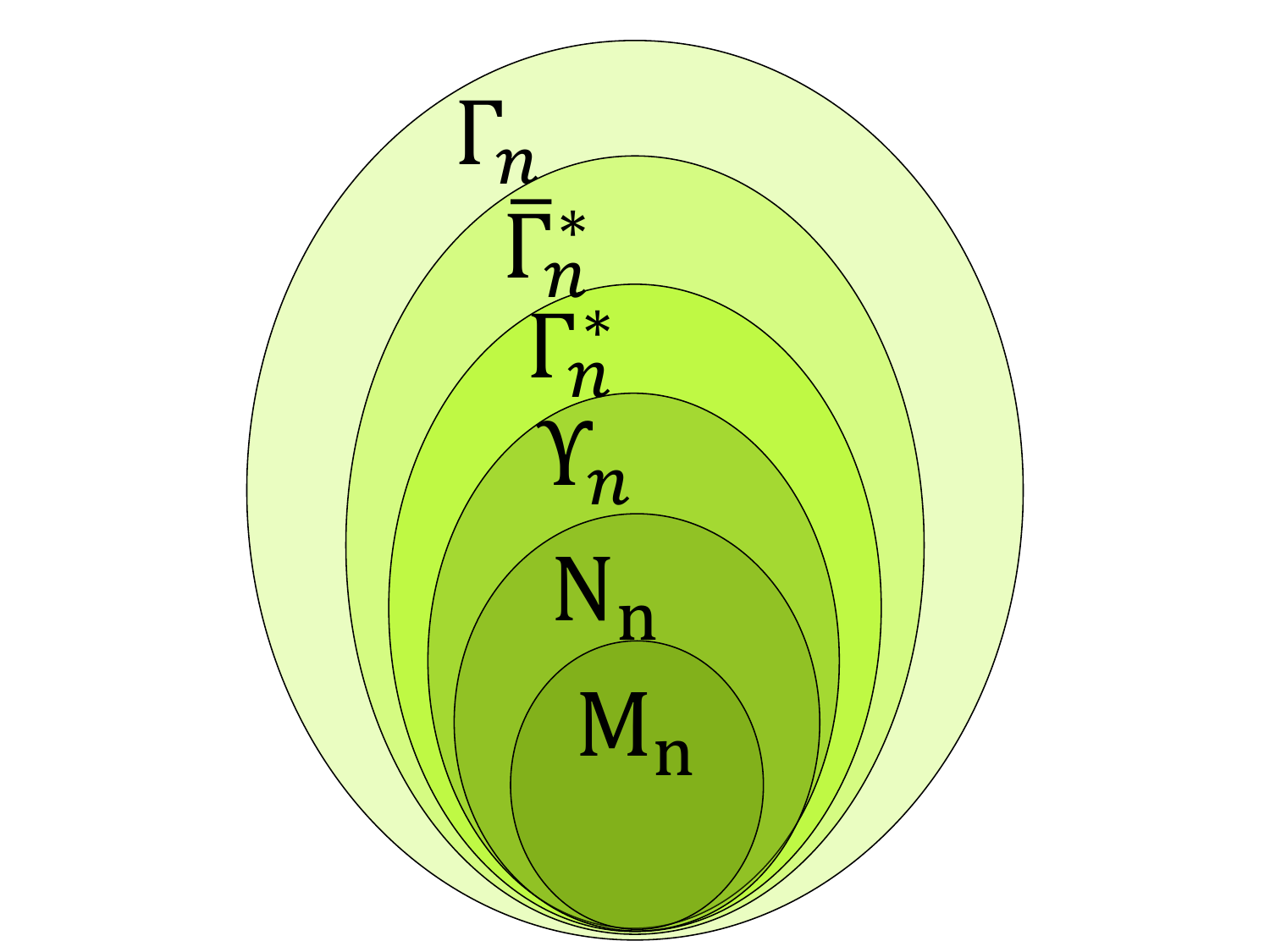}\vspace{-1cm}
  \end{minipage}
  \caption{Landscape of polymatroids}
  \label{fig:diagram}
\end{figure}

An information inequality is an assertion stating that a linear
expression of entropic terms is $\geq 0$.

\begin{defn} \label{def:ii}
  We associate to any vector $\bm c \in \R^{2^{[n]}}$ the following
  {\em information inequality}:
  \begin{align}
    \sum_{\alpha \subseteq [n]} c_\alpha h(\bm X_\alpha) \geq & 0 \label{eq:information:inequality}
  \end{align}
  By using the dot-product notation, we can write the inequality as
  $\bm c \cdot \bm h \geq 0$.  If the inequality holds for all
  $h \in K$, where $K \subseteq \Rp^{2^{[n]}}$ is some set, then we
  say that it is valid for $K$, and write
  $K \models \bm c \cdot \bm h \geq 0$.
\end{defn}

Thus, we will talk about inequalities valid for entropic vectors, or
valid for polymatroids, and the latter are precisely the Shannon
inequalities (this is implicit in the proof of
Th.~\ref{th:primal:dual:bound:polyhedral} below).  Any Shannon
inequality is also valid for entropic vectors; however, we will see in
Th.~\ref{th:zhang:yeung} below a non-Shannon inequality, which is
valid for entropic vectors, but not for polymatroids.  In analogy with
mathematical logic, one should think the vectors $\bm h$ as {\em
  models}, inequalities $\bm c \cdot \bm h \geq 0$ as {\em formulas},
and sets $K \subseteq \Rp^{2^{[n]}}$ as {\em classes of models}.

\begin{ex}
  The following is a Shannon inequality, called {\em Shearer's
    inequality}:
  \begin{align}
    h(XY)+h(YZ)+h(ZX)-2h(XYZ) \geq & 0 \label{eq:shearer}
  \end{align}
  To prove it, we apply submodularity twice, underlining the affected terms:
  \begin{align*}
    &\underline{h(XY)+h(YZ)}+h(ZX) \\
    &\ \ \ \ \geq  h(XYZ)+\underline{h(Y) + h(ZX)}\\
    &\ \ \ \ \geq h(XYZ)+h(XYZ)+h(\emptyset) = 2h(XYZ)
  \end{align*}
  Equivalently, we observe that~\eqref{eq:shearer} is the sum of the
  following basic Shannon inequalities:
  \begin{align*}
    h(XY) + h(YZ) - h(Y) - h(XYZ) \geq & 0 \\
    h(Y) + h(ZX) - h(\emptyset) - h(XYZ) \geq & 0 \\
    h(\emptyset) = & 0
  \end{align*}
\end{ex}

We will prove shortly (Theorem~\ref{th:primal:dual:bound:polyhedral}
below) that one can decide in time exponential in $n$ whether an
inequality is valid for all polymatroids.
%
%
In contrast, it
is an open problem whether entropic validity is decidable.

\subsection{The Entropic Bound and the Polymatroid Bound}

\label{sec:entropic:polymatroid:bound}

The general framework for computing a bound on a query's output uses
degree constraints, which, in turn, correspond to conditional
entropies.  We define these two notions first.

We write $\bm U \bm V$ for set $\bm U \cup \bm V$.  Given
$\bm h \in \R_+^{2^{[n]}}$, define:
\begin{align}
  h(\bm V | \bm U) \defeq & h(\bm U \bm V) - h(\bm U) \label{eq:conditional}
\end{align}
$\bm U, \bm V$ need not be disjoint, and
$h(\bm V|\bm U) = h(\bm V-\bm U|\bm U)$; for example,
$h(XY|X)=h(Y|X)$.  If $h(\bm V|\bm U) = 0$ then we say that $\bm h$
satisfies the functional dependency $\bm U \rightarrow \bm V$, and we
write $\bm h \models \bm U \rightarrow \bm V$.
Lee~\cite{DBLP:journals/tse/Lee87} proved that, if $R$ is a relation
instance with attributes $\bm X$, $p:R \rightarrow [0,1]$ is a
probability distribution, and $\bm h$ is its entropic vector, then
$R \models \bm U \fd \bm V$ iff $\bm h \models \bm U \fd \bm V$.  For
a simple illustration, referring to Fig.~\ref{fig:parity}, both $R$
and its entropy $\bm h$ satisfy the FDs $XY \rightarrow Z$,
$XZ\rightarrow Y$, and $YZ\rightarrow X$: for example $XY$ is a key
(all 4 tuples have distinct values $XY$) and
$h(Z|XY)=h(XYZ)-h(XY)=2-2=0$.

Fix $\bm U$, and denote by
$h(-|\bm U): 2^{\bm X-\bm U}\rightarrow \R_+$ the function
$\bm V \mapsto h(\bm V|\bm U)$.  If $\bm h$ is a polymatroid, then
$h(-|\bm U)$ is also a polymatroid, called the conditional
polymatroid.  If $\bm h$ is an entropic vector, then, surprisingly,
$\bm h(-|\bm U)$ is not necessarily entropic (as we will see later in
Sec.~\ref{subsec:almost:entropic}), yet the name {\em conditional
  entropy} is justified by the following.  Suppose $\bm h$ is
associated to $(R, p)$.  Fix an outcome $\bm u\in \dom^{\bm U}$,
consider the random variable $\bm V$ conditioned on $\bm U = \bm u$,
and denote its entropy by $h(\bm V|\bm U=\bm u)$.  Then:
\begin{align}
  h(\bm V | \bm U) = & \E_{\bm u}\left[h(\bm V | \bm U = \bm u)\right]\label{eq:def:conditional:expectation}
\end{align}
In other words, $h(\bm V | \bm U)$ equals the {\em expectation} over
the outcomes $\bm u$ of the (standard) entropy of the random variable
$\bm V$ conditioned on $\bm U = \bm u$.  The proof of
identity~\eqref{eq:def:conditional:expectation} consists of applying
directly the definition of the entropy given in Eq.~\eqref{eq:h}.

When proving Shannon inequalities it is sometimes convenient to write
the submodularity inequality as
$h(\bm V|\bm U) \geq h(\bm V | \bm U \bm W)$.\footnote{This is
  equivalent to
  $h(\bm U \bm V) - h(\bm U) \geq h(\bm U \bm V \bm W) - h(\bm U \bm
  W)$; when $\bm V \cap \bm W=\emptyset$ then this is a submodularity
  inequality.} In other words, conditioning on more variables can only
decrease the entropy.

\begin{ex} \label{ex:ii:conditionals} We illustrate a simple Shannon
  inequality with conditionals:
    \begin{align*}
    \underline{h(XY)}&\underline{+h(YZ)}+h(ZU) + h(U|XZ) + h(X|YU) \geq\\
      \geq & h(XYZ)+\underline{h(Y)+h(ZU)}+h(U|XZ) + h(X|YU)\\
    \geq & h(XYZ) + h(YZU) + \underline{h(U|XZ)} + \underline{h(X|YU)}\\
    \geq & h(XYZ) + h(YZU) + h(U|XYZ) + h(X|YZU)\\
    = & 2h(XYZU) 
  \end{align*}
\end{ex}

Next, we define degrees of a relation instance
$R \subseteq \dom^{\bm X}$.  Given subsets
$\bm U, \bm V \subseteq \bm X$, and $\bm u \in \dom^{\bm U}$, the {\em
  $\bm V$-degree of $\bm U =\bm u$ in $R$} is the number of distinct
values $\bm v$ that occur in $R$ together with $\bm u$; the {\em
  max-$\bm V$-degree of $\bm U$} is the maximum degree over all values
$\bm u$.  Formally:
\begin{align*}
\degree_R(\bm V | \bm U= \bm u) \defeq & \left|\setof{\bm v}{(\bm u, \bm v) \in \Pi_{\bm U \bm V}(R)}\right|\\
\degree_R(\bm V | \bm U) \defeq & \max_{\bm u}\left(\degree_R(\bm V | \bm U=\bm u)\right)
\end{align*}
We note that $\degree_R(\bm V|\bm U) \geq 1$ (since we assumed
$R\neq \emptyset$), and equality holds iff $R$ satisfies the
functional dependency $\bm U \rightarrow \bm V$.

\begin{defn} \label{def:statistics} Fix a relation $R(\bm X)$. A {\em
    degree statistics}, or a {\em statistics} in short, $\sigma$, is
  an expression of the form $\sigma = (\bm V|\bm U)$ where
  $\bm U, \bm V\subseteq \bm X$; when $\bm U=\emptyset$ then we call
  $\sigma$ a {\em cardinality statistics}, and write it as $(\bm V)$.
  If $\Sigma$ is a set of statistics, then we call
  $\bm B = (B_{\sigma})_{\sigma \in \Sigma}$, where $B_\sigma \geq 1$,
  {\em statistics values} associated to $\Sigma$.  The {\em
    log-statistics values} are
  $\bm b = \log \bm B = (b_\sigma := \log B_\sigma)_{\sigma \in
    \Sigma}$.
\end{defn}

We abbreviate $h(\bm V|\bm U)$ and $\degree_R(\bm V|\bm U)$ with
$h(\sigma)$ and $\degree_R(\sigma)$ respectively.  If $\Sigma$ is a
set of statistics, then a {\em $\Sigma$-inequality} is an inequality
of the following form:
\begin{align} \sum_{\sigma \in \Sigma} w_\sigma h(\sigma) \geq & h(\bm
X) \label{eq:conditional:ii}
\end{align}
where $\bm w = (w_\sigma)_{\sigma \in \Sigma}$ are nonnegative
weights.

Fix a hypergraph $\calH = (\bm X, E)$.  We say that $\Sigma$ is {\em
  guarded by $\calH$} if, for every $\sigma = (\bm V|\bm U)$ in
$\Sigma$, there exists a hyperedge $\bm Y_{\sigma} \in E$ such that
$\bm U, \bm V \subseteq \bm Y_{\sigma}$; we call $\bm Y_{\sigma}$ the
{\em guard} of $\sigma$.  When $\calH$ is the hypergraph of a query
$Q$, then we say that $\Sigma$ is guarded by $Q$, and that
$R_{\sigma}$ is the guard of $\sigma$.  The following theorem
establishes the key connection between information inequalities and
query output size.  The proof relies on a method originally introduced
by Chung et al. for a combinatorial
problem~\cite{DBLP:journals/jct/ChungGFS86}, and adapted by Grohe and
Marx for constraint satisfaction~\cite{DBLP:journals/talg/GroheM14},
then by Atserias, Grohe, and Marx for their AGM
bound~\cite{DBLP:journals/siamcomp/AtseriasGM13}.

\begin{thm} \label{th:degree:bound} Assume $\Sigma$ is guarded by $Q$.
  If the $\Sigma$-inequality~\eqref{eq:conditional:ii} is valid for
  entropic vectors then:
  \begin{align}
    |Q| \leq & \prod_{\sigma \in \Sigma} \degree_{R_\sigma}^{w_\sigma}(\sigma) \label{eq:conditional:q}
  \end{align}
  where $R_\sigma$ is the guard of $\sigma \in \Sigma$.
\end{thm}

\begin{proof}
  Fix a database instance $\bm D$, and let $\bm h$ be the entropic
  vector associated to the relation $Q(\bm D)$; by uniformity,
  $h(\bm X) = \log |Q(\bm D)|$.  If
  $\sigma = (\bm V|\bm U) \in \Sigma$ has guard $R_\sigma$, then:
  \begin{align*}
   h(\sigma) =  & \E_{\bm u}\left[h(\bm V | \bm U = \bm u)\right] \leq  \max_{\bm u} h(\bm V | \bm U = \bm u)\\
    \leq & \max_{\bm u} \log \deg_{R_\sigma}(\bm V | \bm U=\bm u) \\
           =  & \log \deg_{R_\sigma}(\bm V | \bm U) = \log \deg_{R_\sigma}(\sigma)
  \end{align*}
  Using~\eqref{eq:conditional:ii} we derive:
  \begin{align*}
    \sum_\sigma w_\sigma \log \deg_{R_\sigma}(\sigma) \geq & \sum_\sigma w_\sigma h(\sigma) \geq  h(\bm X) = \log|Q|
  \end{align*}
\end{proof}

Inequality~\eqref{eq:conditional:ii} is similar to the AGM
inequality~\eqref{eq:agm}.  Next, we proceed as we did for the AGM
bound: fix numerical values for the statistics, then minimize the
bound over all valid $\bm w$'s.  We say that a database instance
$\bm D$ {\em satisfies} the statistics $\Sigma, \bm B$, in notation
$\bm D \models (\Sigma, \bm B)$, if
$\degree_{R_\sigma^{\bm D}}(\sigma)\leq B_\sigma$ for all
$\sigma \in \Sigma$.  Similarly, we say that a vector $\bm h$
satisfies the log-statistics $\Sigma, \bm b$ if
$h(\sigma) \leq b_\sigma$ for all $\sigma$.  We define:

\begin{defn}[Query Upper Bound] \label{def:ep:bound} Let
  $\Sigma, \bm B$ be statistics values, guarded by the query $Q$.  Fix
  some set $K\subseteq \R^{2^{[n]}}$.  The {\em Upper Bound}
  w.r.t. $K$ of the query $Q$ is:
  \begin{align*}
    \text{U-Bound}_K(Q,\Sigma,\bm B) \defeq & \inf_{\bm w: K \models \mbox{Eq.\eqref{eq:conditional:ii}}} \prod_{\sigma\in\Sigma} B_\sigma^{w_\sigma}
  \end{align*}
  The {\em entropic upper bound} is $\text{U-Bound}_{\Gamma_n^*}$ and
  the {\em polymatroid upper bound} is $\text{U-Bound}_{\Gamma_n}$.
\end{defn}

Sometimes it is more convenient to use the logarithm and define:
\begin{align}
  \text{Log-U-Bound}_K(Q,\Sigma,\bm b) \defeq & \inf_{\bm w: K \models \mbox{Eq.\eqref{eq:conditional:ii}}} \sum_{\sigma\in\Sigma}w_\sigma b_\sigma \label{eq:def:log-u-bound}
\end{align}
%
%
\begin{cor} \label{cor:entropic:bound} The following hold:
  \begin{itemize}
  \item
    $\text{U-Bound}_{\Gamma_n^*}(Q,\Sigma,\bm B)\leq
    \text{U-Bound}_{\Gamma_n}(Q,\Sigma,\bm B)$, and
  \item If $\bm D \models (\Sigma, \bm B)$ then
    $|Q(\bm D)| \leq \text{U-Bound}_{\Gamma_n^*}(Q,\Sigma,\bm B)$.
  \end{itemize}
\end{cor}
The first item is by $\Gamma_n^* \subseteq \Gamma_n$, the second is by
Th.~\ref{th:degree:bound}.

Let's compare these bound with the AGM bound~\eqref{eq:agm:formula}.
There, we had to minimize an expression where $\bm w$ ranged over the
fractional edge covers of the query's hypergraph.  In our new setting,
$\bm w$ ranges over valid $\Sigma$-inequalities, a much more difficult
task.  To compute the polymatroid bound, $\bm w$ ranges over
$\Sigma$-inequalities valid for polymatroids, and we will show in
Th.~\ref{th:primal:dual:bound:polyhedral} that this bound can be
computed in time exponential in $n$.  However, in order to compute the
entropic bound, $\bm w$ needs to define a valid entropic inequality,
and it is currently open whether this bound is computable.  On the
other hand, we will prove that the entropic bound is asymptotically
tight, while the polymatroid bound is not.  Thus, we are faced with a
difficult choice, between and exact but non-computable bound, or a
computable but inexact bound.  This justifies examining non-trivial
special cases of statistics $\Sigma$ when these two bounds agree.
We illustrate the entropic upper bound with an example.

\begin{ex} \label{ex:pods2016}
  Consider the following conjunctive query:
  \begin{align*}
    Q(X,Y,Z,U) = & R(X,Y) \wedge S(Y,Z) \wedge T(Z,U) \\
    \wedge & A(X,Z,U)\wedge B(X,Y,U)
  \end{align*}
  Suppose that we are given the following set of statistics
  $\Sigma = \set{(XY), (YZ), (ZU), (U|XZ), (X|YU)}$.  In other words,
  we have bounds on the cardinalities of $R, S, T$, but not of $A, B$,
  hence we can assume that $|A|=|B|=\infty$.  Instead, we have the
  statistics $\degree_A(U|XZ)$ and $\degree_B(X|YU)$.  The AGM
  bound~\eqref{eq:agm} is $|Q| \leq |R| \cdot |T|$, because the only
  fractional edge cover whose bound is $< \infty$ is $w_R = w_T = 1$
  and $w_S = w_A = w_B = 0$.
  
  The polymatroid bound follows from these $\Sigma$-inequalities:
  \begin{align*}
    & h(XY)+h(YZ)+h(ZU) + h(U|XZ) + h(X|YU) \geq \nonumber \\
    & \hspace{5cm} \geq  2h(XYZU) \\
    & h(XY) + h(ZU) \geq  h(XYZU) \\
    & h(XY) + h(YZ) + h(U|XZ) \geq h(XYZU) \\
    & h(YZ) + h(ZU) + h(X|YU) \geq h(XYZU)
  \end{align*}
  We proved the first inequality in Example~\ref{ex:ii:conditionals},
  while the other three are immediate.  They imply:
  \begin{align*}
    |Q| \leq & \big(|R|\cdot |S| \cdot |T| \cdot \degree_A(U|XZ) \cdot \degree_B(X|YU)\big)^{1/2}\\
    |Q| \leq & |R|\cdot |T| \\
    |Q| \leq & |R| \cdot |S| \cdot  \degree_B(X|YU)\\
    |Q| \leq & |S| \cdot |T| \cdot  \degree_A(U|XZ)
  \end{align*}
  The AGM bound is the second inequality. We show in
  Appendix~\ref{app:ex:pods2016} that the entropic bound is the
  minimum over all four expressions above. (This requires proving that
  there is no ``better'' inequality that gives us a smaller bound.)
  Since all four inequalities are Shannon inequalities, it follows
  that, in this case, the entropic bound is equal to the polymatroid
  bound.

  When $XZ$ is a key in $A$, and $YU$ is a key in $B$, then the
  polymatroid bound simplifies to:
  \begin{align*}
    |Q| \leq \min((|R|\cdot|S|\cdot|T|)^{1/2},|R|\cdot|T|, |R|\cdot |S|, |S| \cdot |T|)
  \end{align*}
\end{ex}

{\bf Special Case: Cardinality Constraints} We show next that the AGM
bound is a special case where the polymatroid and entropic bounds
coincide; we will see a more general setting when this happens in
Sec.~\ref{sec:special:cases}.  Assume $\Sigma$ is restricted to
cardinality constraints, and assume for simplicity that $\Sigma$ has
exactly one cardinality constraint $(\bm Y_j)$ for each relation
$R_j(\bm Y_j)$ in the query;
$\Sigma = \set{(\bm Y_1),\ldots, (\bm Y_m)}$.  The
$\Sigma$-inequality~\eqref{eq:conditional:ii} becomes:
\begin{align}
  \sum_{j=1,m} w_j h(\bm Y_j) \geq & h(\bm X) \label{eq:shearer:ii}
\end{align}
When $w_1=\cdots=w_m$ then~\eqref{eq:shearer:ii} is called Shearer's
inequality.
\begin{thm} \label{th:shearer}
  The following are equivalent:
  \begin{enumerate}[(i)]
  \item \label{item:th:shearer:1} Inequality~\eqref{eq:shearer:ii} is
    valid for polymatroids.
  \item \label{item:th:shearer:2} Inequality~\eqref{eq:shearer:ii} is
    valid for entropic functions.
  \item \label{item:th:shearer:3} The weights $\bm w$ form a
    fractional edge cover of the hypergraph
    $(\bm X, \setof{\bm Y_j}{j=1,m})$.
  \end{enumerate}
\end{thm}

\begin{figure}
    \begin{align*}
    R = &
      \begin{array}{|ccc|c} \cline{1-3}
        X_1 \ldots X_{i-1} & X_i & X_{i+1} \ldots X_n & p \\ \cline{1-3}
        0 \ \ \ \ldots \ \ \ 0 & 0 & 0 \ \ \ \ldots\ \ \  0 & 1/2 \\
        0 \ \ \ \ldots \ \ \ 0 & 1 & 0 \ \ \ \ldots\ \ \  0 & 1/2 \\ \cline{1-3}
      \end{array} 
        & h^{X_i}(\bm U) =
       \begin{cases}
        1 & \mbox{if $X_i \in \bm U$}\\
        0 & \mbox{if $X_i \not\in \bm U$}
      \end{cases}
  \end{align*}
\caption{A relation $R$ with two tuples that agree on all attributes,
  except $X_i$.  Its entropic vector is called the {\em basic modular
    function}, $\bm h^{X_i}$; it is used in Th.~\ref{th:shearer}, and
  discussed in more detail in Sec.~\ref{sec:special:cases}.}
  \label{fig:modular}
\end{figure}

\begin{proof}
  \ref{item:th:shearer:1} $\Rightarrow$ \ref{item:th:shearer:2} is
  immediate.  For~\ref{item:th:shearer:2} $\Rightarrow$
  \ref{item:th:shearer:3}, assume that~\eqref{eq:shearer:ii} holds for
  all entropic vectors, and consider any variable $X_i \in \bm X$, for
  $i=1,n$.  Consider the {\em basic modular entropic function}
  $\bm h^{X_i}$ shown in Fig.~\ref{fig:modular}.  Since $\bm h^{X_i}$
  satisfies inequality~\eqref{eq:shearer:ii}, it follows that
  $\sum_{j=1,m: X_i \in \bm Y_j} w_j \geq 1$ (because
  $h^{X_i}(\bm Y_j)=1$ iff $X_i \in \bm Y_j$), proving that $\bm w$ is
  a fractional edge cover.

  It remains to prove the implication~\ref{item:th:shearer:3}
  $\Rightarrow$ \ref{item:th:shearer:1}.  This is a well known result,
  and it admits multiple proofs (we give a second proof in
  Sec.~\ref{sec:heavy:light}).
  Here, we will prove the inequality by induction on $n$:
  \begin{itemize}
  \item Partition the set of indices $j$ into $J_0$ and $J_1$:
    \\
    $J_0 \defeq \setof{j}{X_n \not\in \bm Y_j}$,
    $J_1 \defeq \setof{j}{X_n \in \bm Y_j}$.
  \item If $j\in J_1$ then write
    $h(\bm Y_j) = h(X_n) + h(\bm Y_j-X_n|X_n)$.  Note that
    $\sum_{j \in J_1} w_j \geq 1$ because $X_n$ is covered.
  \item If $j\in J_0$ then write $h(\bm Y_j) \geq h(\bm Y_j-X_n|X_n)$.
  \end{itemize}
  Using the steps above we obtain:
  \begin{align*}
    \sum_j  w_jh(\bm Y_j) \geq  & h(X_n) + \sum_j w_jh(\bm Y_j-X_n|X_n) \\
    \geq &   h(X_n) + h(\bm X-X_n|X_n) =  h(\bm X)
  \end{align*}
  The last line used induction on the polymatroid $h(- | X_n)$.
\end{proof}

It follows immediately that the AGM bound, the entropic bound, and the
polymatroid bound coincide in the simple case when the statistics are
restricted to the cardinalities of the input relations.

{\bf Discussion} Degree constraints occur often and naturally in
database applications.  For example, if a relation $R_j(\bm Y_j)$ has
a key $\bm U$, then $\degree_{R_j}(\bm Y_j|\bm U) = 1$.  In practice
almost every relation has a key, so this case is very common.  In
other cases some cardinality constraints can be obtained directly from
the application. For example, suppose that in a database of customers
we require that no customer may have more than 10 credit cards, which
naturally leads to a max-degree constraint.  Such constraints are used
in some modern systems, for example in {\em scale-independent query
  processing}~\cite{DBLP:conf/cidr/ArmbrustFPLTTO09,
  DBLP:journals/pvldb/ArmbrustCKFFP11,
  DBLP:conf/sigmod/ArmbrustLKFFP13}.


\section{The Worst-Case Instance}

\label{sec:lower:bound}


Informally, we call a database instance $\bm D$ a {\em worst-case
  instance} if it satisfies the given statistics, and the query's
output is as large as, or approaches asymptotically (in a sense to be
made precise), the entropic upper bound.  We will show that such a
worst-case instance exists, proving that the entropic bound is {\em
  asymptotically tight}, which is a weaker notion of tightness than
for the AGM bound in Th.~\ref{th:agm:lower:bound}.  We will also show
that, in general, the polymatroid bound is not tight, even for this
weaker notion of tightness.

To construct the worst-case instance we need a {\em dual} definition
of the entropic and polymatroid bounds.  We define them directly using
log-version:

\begin{defn}[Query Lower Bound] \label{def:ep:bound:dual} Fix
  log-statistics $\Sigma, \bm b$.  For any set
  $K\subseteq \R^{2^{[n]}}$, the {\em Log Lower Bound} w.r.t. $K$ is:
  \begin{align}
    \text{Log-L-Bound}_K(Q,\Sigma,\bm b) \defeq & \sup_{\bm h \in K: \bm h \models (\Sigma, \bm b)}h(\bm X)\label{eq:k:bound}
  \end{align}
  As before, the {\em entropic log-lower bound} is
  $\text{Log-L-Bound}_{\Gamma_n^*}$, and the {\em polymatroid
    log-lower bound} is $\text{Log-L-Bound}_{\Gamma_n}$.
\end{defn}

The log-lower bound asks us to find a vector $\bm h$ that satisfies
all log-statistics $\Sigma, \bm b$, and where $\bm h(\bm X)$ is as
large as possible.  We call $\bm h$ a {\em worst case} entropic
vector, or the worst case polymatroid respectively.  Using $\bm h$, we
would like to construct a worst-case database instance $\bm D$, that
satisfies $\Sigma, \bm B$, and $\log |Q(\bm D)| = h(\bm X)$.  The
difficulty lies in the fact that, when $\bm h$ is a polymatroid then
such a database may not exists general, and when $\bm h$ is an
entropic vector, then it may be realized by a probability space that
is non-uniform, hence we cannot use it to construct $\bm D$.  We start
by observing the following, which is easy to check:
\begin{align}
  \text{Log-L-Bound}_K(Q,\Sigma,\bm b)\leq & \text{Log-U-Bound}_K(Q,\Sigma,\bm b) \label{eq:l:leq:u:for:k}
\end{align}
Indeed, if $\bm h\in K$ satisfies $\bm h \models (\Sigma, \bm b)$, and
$\bm w$ satisfies
$\forall \bm h \in K, \sum_\sigma w_\sigma h(\sigma) \geq h(\bm X)$,
then
$h(\bm X) \leq \sum_\sigma w_\sigma h(\sigma) \leq \sum_\sigma
w_\sigma b_\sigma$, and the claim follows from
$\text{Log-L-Bound}_K=\sup_{\bm h} h(\bm X) \leq
\text{Log-U-Bound}_K=\inf_{\bm w} \sum_\sigma w_\sigma b_\sigma$.

When $K=\Gamma_n$, then~\cite{DBLP:conf/pods/Khamis0S17} showed
that the two bounds are equal.  We prove a slightly more general
statement:

\begin{thm} \label{th:primal:dual:bound:polyhedral} Suppose the set
  $K$ is defined by linear constraints:
  $K = \setof{\bm h \in \Rp^{2^{[n]}}}{\bm M \cdot \bm h \geq 0}$,
  where $\bm M$ is some matrix.\footnote{Equivalently: $K$ is a
    polyhedral cone, reviewed in Sec.~\ref{subsec:almost:entropic}.}
  Then, $\text{Log-L-Bound}_K$ and $\text{Log-U-Bound}_K$ are defined
  by a pair of primal/dual linear programs, with a number of variables
  exponential in $n$; the expressions $\inf, \sup$
  in~\eqref{eq:def:log-u-bound}, ~\eqref{eq:k:bound} can be replaced
  by $\min, \max$; and Eq.~\eqref{eq:l:leq:u:for:k} becomes an
  equality, $h^*(\bm X) = \sum_\sigma w^*_\sigma b_\sigma$, where
  $\bm h^*$, $\bm w^*$ are the optimal solutions of the primal and
  dual program respectively.
\end{thm}

\begin{proof} Denote $s = |\Sigma|$, and let $\bm A$ be the
  $s \times 2^n$ matrix that maps $\bm h$ to the vector
  $\bm A \cdot \bm h = (h(\sigma))_{\sigma \in \Sigma}\in \R^s$.  Let
  $\bm c \in \R^{2^n}$ be the vector $c_{\bm X} = 1$, $c_{\bm U}=0$
  for $\bm U \neq \bm X$.  The two bounds are the optimal solutions to
  the following pair of primal/dual linear programs:
  \begin{align*}
    &
      \begin{array}{ll|ll}
        \multicolumn{2}{l|}{\text{Log-L-Bound}_K}& \multicolumn{2}{|l}{\text{Log-U-Bound}_K} \\ \hline
        \multicolumn{2}{l|}{\mbox{Maximize } \bm c^T \cdot \bm h} & \multicolumn{2}{|l}{\mbox{Minimize }  \bm w^T\cdot\bm b}\\
        \mbox{where} & \bm A \cdot \bm h \leq \bm b &  \mbox{where} & \bm w^T \cdot \bm A - \bm c^T \geq \bm u^T \cdot\bm M\\
                              & -\bm M \cdot \bm h \leq 0 & &
      \end{array}
  \end{align*}
  where the primal variables are $\bm h \geq 0$, and the dual
  variables are $\bm w, \bm u \geq 0$; the reader may check that the
  two programs above form indeed a primal/dual pair.
  $\text{Log-L-Bound}_K$ is by definition the optimal value of the
  program above.  We prove that $\text{Log-U-Bound}_K$ is the value of
  the dual.  First, observe that the
  $\Sigma$-inequality~\eqref{eq:conditional:ii} is equivalent to
  $(\bm w^T \cdot \bm A - \bm c^T)\cdot h \geq 0$.  We claim that this
  inequality holds $\forall \bm h \in K$ iff there exists $\bm u$
  s.t. $(\bm w,\bm u)$ is a feasible solution to the dual.  For that
  consider the following primal/dual programs with variables
  $\bm h\geq 0$ and $\bm u\geq 0$ respectively:
  \begin{align*}
    &
      \begin{array}{ll|ll}
        \multicolumn{2}{l|}{\mbox{Minimize } (\bm w^T \cdot\bm A - \bm c^T) \cdot \bm h} & \multicolumn{2}{|l}{\mbox{Maximize }  0}\\
        \mbox{where} & \bm M  \cdot \bm h \geq 0 &  \mbox{where} & \bm u^T \cdot \bm M\leq\bm w^T \bm A - \bm c^T 
      \end{array}
  \end{align*}
  The primal (left) has optimal value 0 iff the inequality
  $(\bm w^T \cdot \bm A - \bm c^T) \cdot \bm h\geq 0$ holds forall
  $\bm h \in K$; otherwise its optimal is $-\infty$.  The dual (right)
  has optimal value 0 iff there exists a feasible solution $\bm u$;
  otherwise its optimal is $-\infty$.  Strong duality proves our
  claim.
\end{proof}


When $K=\Gamma_n^*$, then the two terms in~\eqref{eq:l:leq:u:for:k}
may no longer be equal in general, but we prove that they are equal
asymptotically.  Call a {\em $k$-amplification} of a set of
log-statistics $\Sigma, \bm b$ the log-statistics $\Sigma, k\bm b$,
where $k$ is a natural number.  Observe that the entropic log-upper
bound increases linearly with the $k$-amplification:
\begin{align}
\text{Log-U-Bound}_{\Gamma_n^*}(Q,\Sigma,k\bm b) = &
k\text{Log-U-Bound}_{\Gamma_n^*}(Q,\Sigma,\bm b) \label{eq:linearity:log:u:bound}
\end{align}
The lower bound increases at least linearly,
$\text{Log-L-Bound}_{\Gamma_n^*}(Q,\Sigma,k\bm b) \geq
k\text{Log-L-Bound}_{\Gamma_n^*}(Q,\Sigma,\bm b)$, because of the
following proposition:

\begin{prop} \label{prop:sum} If $\bm h_1, \bm h_2$ are two entropic
  vectors, then so is $\bm h_1 + \bm h_2$. 
\end{prop}
\begin{proof} Suppose $\bm h_1, \bm h_2$ are realized by two finite
  probability spaces $(R_1, p_1), (R_2,p_2)$.  Then their sum is
  realized by $(R_1 \otimes R_2, p)$ (see
  Def.~\ref{def:domain:product}), where
  $p(\bm x_1 \otimes \bm x_2)\defeq p(\bm x_1) \cdot p(\bm x_2)$.
\end{proof}

We prove:

\begin{thm} \label{th:primal:dual:bound} Fix any $Q, \Sigma, \bm b$.
  The entropic upper and lower bounds are {\em asymptotically equal},
  in the following sense:
  \begin{align}
  \sup_k \frac{\text{Log-L-Bound}_{\Gamma_n^*}(Q,\Sigma,k\bm b)}{\text{Log-U-Bound}_{\Gamma_n^*}(Q,\Sigma,k\bm b)}=&1 \label{eq:th:primal:dual:bound:entropic}
  \end{align}
\end{thm}

The proof of this result, which appears to be novel, requires a
discussion of the set of {\em almost-entropic functions}, and we defer
this to Sec.~\ref{sec:conditional:inequalities}.

Finally, we can answer the central question in this section: the
entropic bound is tight asymptotically, while the polymatroid bound is
not.

\begin{thm} \label{th:tight:non-tight} (1) For any query $Q$ and
  statistics $\Sigma, \bm B$, the entropic bound is asymptotically
  tight, in the following sense:
  \begin{align}
    \sup_{k}\frac{\sup_{\bm D: \bm D \models \bm B^k}\log |Q(\bm D)|}{\text{Log-L-Bound}_{\Gamma_n^*}(Q,\Sigma,k\bm b)}=&1 \label{eq:tight}
  \end{align}
  (2) The polymatroid bound is not asymptotically tight: there exists
  a query $Q$ and statistics $\Sigma, \bm B$ such that:
  \begin{align}
    \sup_{k}\frac{\sup_{\bm D: \bm D \models \bm B^k}\log |Q(\bm D)|}{\text{Log-L-Bound}_{\Gamma_n}(Q,\Sigma,k\bm b)}\leq&\frac{43}{44} \label{eq:non-tight}
  \end{align}
  Here $\bm B^k\defeq (B^k_\sigma)_{\sigma\in\Sigma}$, and
  $\bm b \defeq \log \bm B$.  Moreover, this property holds even if
  $\Sigma, \bm B$ consists only of cardinality constraints and
  functional dependencies, i.e. $\forall \sigma \in \Sigma$ either
  $\sigma = (\bm V)$ or $B_\sigma = 1$.
\end{thm}

Equivalently, Eq.~\eqref{eq:tight} says that, for any statistics
$\Sigma,\bm B$, if we allow a sufficiently large amplification
$\Sigma, \bm B^k$, then there exists a worst-case instance $\bm D$
satisfying the amplified statistics such that $\log |Q(\bm D)|$
approaches $\text{L-Bound}_{\Gamma_n^*}(Q,\Sigma,\bm B^k)$.  Notice
that this is a weaker notion of tightness than for the AGM bound in
Theorem~\ref{th:agm:lower:bound}.
There, tightness referred to the ratio betwen the lower and upper
bound, while here tightness refers to the ratio of their logarithms, a
weaker notion.  Eq.~\eqref{eq:non-tight} says that even this weaker
tightness fails for the polymatroid bound.

The first proof of asymptotic tightness was given by Gogacz and
Torunczyk~\cite{DBLP:conf/icdt/GogaczT17}, for the restricted case
when the statistics are either cardinalities, or functional
dependencies.  The general case was proven
in~\cite{DBLP:conf/pods/Khamis0S17}.  Both results were stated
slightly differently from ours, by using almost-entropic functions.
We prefer to state our result in terms of the entropic functions,
since it is more natural, and defer the discussion of almost entropic
functions to Sec.~\ref{sec:conditional:inequalities}, where they have
a very natural justification.

In the rest of this section we prove Theorem~\ref{th:tight:non-tight}.
We use this opportunity to continue our dive into the fascinating
world of entropic functions, and non-Shannon inequalities, which are
needed for the proof.  However, the rest of this section is rather
technical, and readers not interested in this background may safely
skip the rest of this section, since we do not need it, except for the
short introduction of mutual information.

\subsection{Background: Non-Shannon Inequalities, Lattices, Groups}

\label{sec:background:non-shannon}

{\bf Mutual Information}
Given a vector $\bm h \in \R^{2^{[n]}}$ and three disjoint sets of
variables $\bm U, \bm V, \bm W \subseteq \bm X$, we denote by:
\begin{align}
  I_h(\bm V; \bm W|\bm U) \defeq & h(\bm U\bm V) + h(\bm U\bm W)-h(\bm U) - h(\bm U\bm V \bm V)
\label{eq:mi}
\end{align}
When $h$ is clear from the context, then we will drop the index $h$
from $I_h$ and write simply $I$.  When $h$ is an entropic vector, then
$I_h(\bm V; \bm W|\bm U)$ is called the {\em mutual information of
  $\bm V, \bm W$ conditioned on $\bm U$}.  In that case,
$I_h(\bm V; \bm W|\bm U) = 0$ iff the probability space realizing $h$
satisfies $\bm V \perp \bm W | \bm U$, meaning that $\bm V, \bm W$ are
independent conditioned on $\bm U$.  With some abuse, we will call
$I_h$ a conditional mutual information even when $\bm h$ is a
polymatroid.  The following properties hold:

\begin{prop} \label{prop:I:properties}
  For any polymatroid $h$:
  \begin{itemize}
  \item $I_h(\bm V; \bm W|\bm U) \geq 0$ (this is a submodularity
    inequality).
  \item The {\em chain rule} holds: \\
    $I_h(\bm U \bm V; \bm W|\bm Z) = I_h(\bm U;\bm W|\bm Z) + I_h(\bm  V;\bm W|\bm U \bm Z)$.
  \item An {\em elemental mutual information term} is an expression of
    the form $I_h(X_i;X_j|\bm U)$, where $X_i\neq X_j \not\in \bm U$
    are single variables.  Every mutual information
    $I(\bm V;\bm W|\bm U)$ is the sum of elemental terms.
  \item For any subsets $\bm U_0\subseteq \bm U$ and
    $\bm V_0 \subseteq \bm V$:
    \begin{align}
    I_h(\bm U_0;\bm V_0|\bm Z) \leq & I_h(\bm U; \bm V | \bm Z) \label{eq:I:chain:condition}
    \end{align}
    If $I_h(\bm U; \bm V | \bm Z)=0$, then
    $I_h(\bm U_0;\bm V_0|\bm Z) = 0$.
  \end{itemize}
\end{prop}

{\bf Non-Shannon inequalities} The first non-Shannon inequality was
proven by Zhang and Yeung~\cite{zhang1998characterization}.  We review
it here, following the simplified presentation by
Romashchenko~\cite{Romashchenko-talk-2022-11-09} (see also
Csirmaz~\cite{Csirmaz-talk-2022-10-12}).

Zhang and Yeung~\cite{zhang1998characterization} proved the following:

\begin{thm} \label{th:zhang:yeung}
  The following  is a non-Shannon inequality:
  \begin{align}
    I_h(X ; Y) \leq & I_h(X ; Y |A)+I_h(X ; Y |B)+I_h(A; B) \label{eq:zhang:yeung} \\
             + & I_h(X ; Y |A)+I_h(A; Y |X)+I_h(A; X |Y) \nonumber
  \end{align}
\end{thm}

In other words, this inequality is valid for entropic vectors, but it
cannot be proven using the basic Shannon inequalities, hence the term
{\em non-Shannon}.  The proof of this inequality, and that of several
other non-Shannon inequalities proven after 1998, relies on the
following {\em copy lemma}.

\begin{lmm}[Copy Lemma] \label{lmm:copy}
  Let $\bm X, \bm Y$ be two disjoint sets of variables, and let
  $\bm h$ be an entropic vector with variables $\bm X \bm Y$.  Let
  $\bm Y'$ be fresh copies of the variables $\bm Y$.  Thus, each
  variable $Y \in \bm Y$ has a copy $Y' \in \bm Y'$.  Then there
  exists an entropic vector $\bm h'$ over variables $\bm X\bm Y\bm Y'$
  such that the following hold:
  \begin{align*}
    & I_{h'}(\bm Y; \bm Y'|\bm X) = 0 \\
    \forall \bm U \subseteq \bm X\bm Y:\ \ &h'(\bm U) = h(\bm U) \nonumber\\
    \forall \bm U' \subseteq \bm X\bm Y':\ \ &h'(\bm U') =  h(\bm U) \nonumber
  \end{align*}
  We say that $\bm Y'$ is a {\em copy of $\bm Y$ over $\bm X$}.
\end{lmm}

The first inequality asserts $\bm Y \perp \bm Y' | \bm X$.  The second
asserts that $\bm h, \bm h'$ agree on $\bm X\bm Y$.  And the last
equality asserts that $\bm h'$ on $\bm X\bm Y'$ is identical to
$\bm h$ on $\bm X\bm Y$ up to the renaming of variables from $\bm Y'$
to $\bm Y$ (assuming $X'=X$ for $X\in \bm U$).

\begin{proof} (of Lemma~\ref{lmm:copy}) Let $p$ be a probability
  distribution of random variables $\bm X\bm Y$ that realizes the
  entropic vector $\bm h$.  Define the following probability
  distribution $p'$ of random variables $\bm X \bm Y \bm Y'$: the
  domains of the variables $\bm Y'$ is the same as that of $\bm Y$,
  and for all outcomes $(\bm x, \bm y, \bm y')$,
  $p'(\bm X=\bm x, \bm Y = \bm y, \bm Y' = \bm y')\defeq$
  $\frac{p(\bm X = \bm x, \bm Y = \bm y)p(\bm X = \bm x, \bm Y = \bm
    y')}{p(\bm X= \bm x)}$.  The claims in the lemma are easily
  verified.
\end{proof}

In general, the copy lemma does not hold for polymatroids, as we will
see shortly.  We prove now Zhang and Yeung's
inequality~\eqref{eq:zhang:yeung}.  Start from the following Shannon
inequality over 5 variables, $X,Y,A,B,A'$:
\begin{align}
  I_h&(X ; Y) \leq I_h(X ; Y |A)+I_h(X ; Y |B)+I_h(A; B)+ \nonumber \\
  + & I_h(X ; Y |A')+I_h(A'; Y |X)+I_h(A'; X |Y) + \nonumber \\
  + & 3I_h(A';AB|XY)  \label{eq:shannon:complicated}
\end{align}
While this is ``only'' a Shannon inequality, it is surprisingly
difficult to prove; we invite the readers to try it themselves, but,
for completeness, we give the proof in
Appendix~\ref{app:eq:shannon:complicated}.  Consider now an entropic
vector $\bm h$ over four variables, $X,Y,A,B$.  We apply the copy
lemma, and copy $AB$ over $XY$, resulting in an entropic vector $h'$
with variables $X,Y,A,B,A',B'$.  In particular, $h'$
satisfies~\eqref{eq:shannon:complicated} (we don't use $B'$).  Now we
observe that (a) $I(A';AB|XY) = 0$, and (b) every occurrence of $A'$
in the second line can be replaced by $A$; for example
$I(X ; Y |A')=I(X ; Y |A)$, because $I(X ; Y |A')$ is expressed in
terms of $h(A'),h(XA'),h(YA'),h(XYA')$, which are equal to their
copies $h(A),h(XA),h(YA),h(XYA)$.  Thus,
inequality~\eqref{eq:shannon:complicated}
becomes~\eqref{eq:zhang:yeung}, proving that~\eqref{eq:zhang:yeung} is
valid for all entropic functions $\bm h$.

It remains to prove that~\eqref{eq:zhang:yeung} is not a Shannon
inequality, and for that it suffices to describe one polymatroid that
fails the inequality.  To ``see'' this polymatroid, it is best to view
it as being defined over a lattice.  We take this opportunity to
discuss another important concept: polymatroids on lattices.

{\bf Polymatroids on lattices} A {\em polymatroid} on a lattice
$(L,\preceq)$ is a function $h : L \rightarrow \R_+$ satisfying:
\begin{align*}
  h(\hat 0) = & 0 \\
  h(x \vee y) \geq & h(x) && \mbox{monotonicity} \\
  h(x) + h(y) \geq & h(x \vee y) + h(x \wedge y) && \mbox{submodularity}
\end{align*}

Let $\bm X$ be a set of $n$ variables, and $\Sigma$ a set of
functional dependencies for $\bm X$.  Recall from
Sec.~\ref{sec:problem} that $(L_\Sigma,\subseteq)$ is the lattice of
closed sets.  If a (standard) polymatroid $h \in \Gamma_n$ satisfies
the functional dependencies $\Sigma$, then it is not hard to see that
its restriction to $L_\Sigma$ is a polymatroid on $L_\Sigma$.
Conversely, any polymatroid $\bm h$ on the lattice
$(L_\Sigma,\subseteq)$ can be extended to a standard polymatroid
$\bm \bar h:2^{\bm X} \rightarrow \Rp$ by setting
$\bar h(\bm U) \defeq h(\bm U^+)$, and, furthermore, $\bm h$ satisfies
$\Sigma$.  In short, there is a one-to-one correspondence between
polymatroids satisfying a set of functional dependencies, and
polymatroids defined on the associated lattice.


We now complete the proof of Theorem~\ref{th:zhang:yeung}, by showing
that inequality~\eqref{eq:zhang:yeung} does not hold for the
polymatroid $\bm h$ in Fig.~\ref{fig:zhang:yeung:h}.  To read the
figure, recall that $h(\bm U) = h(\bm U^+)$ for any set $\bm U$.  For
example, $ABX^+ = ABXY$, therefore $h(ABX) = h(ABX^+) = h(ABXY) = 4$,
and, also, $h(AB) = h(AB^+) = h(ABXY) = 4$.  We check now that $\bm h$
violates the inequality~\eqref{eq:zhang:yeung}:
\begin{align*}
  & I(X;Y) =  1 \\
  & I(X ; Y |A)+I(X ; Y |B)+I(A; B) = 0 + 0 + 0 \\
  & I(X ; Y |A)+I(A; Y |X)+I(A; X |Y) = 0 + 0 + 0
\end{align*}
The LHS of~\eqref{eq:zhang:yeung} is 1, while the RHS is 0.  This
completes the proof of Theorem~\ref{th:zhang:yeung}.

\begin{figure}
  \centering
    \includegraphics[width=0.45\linewidth]{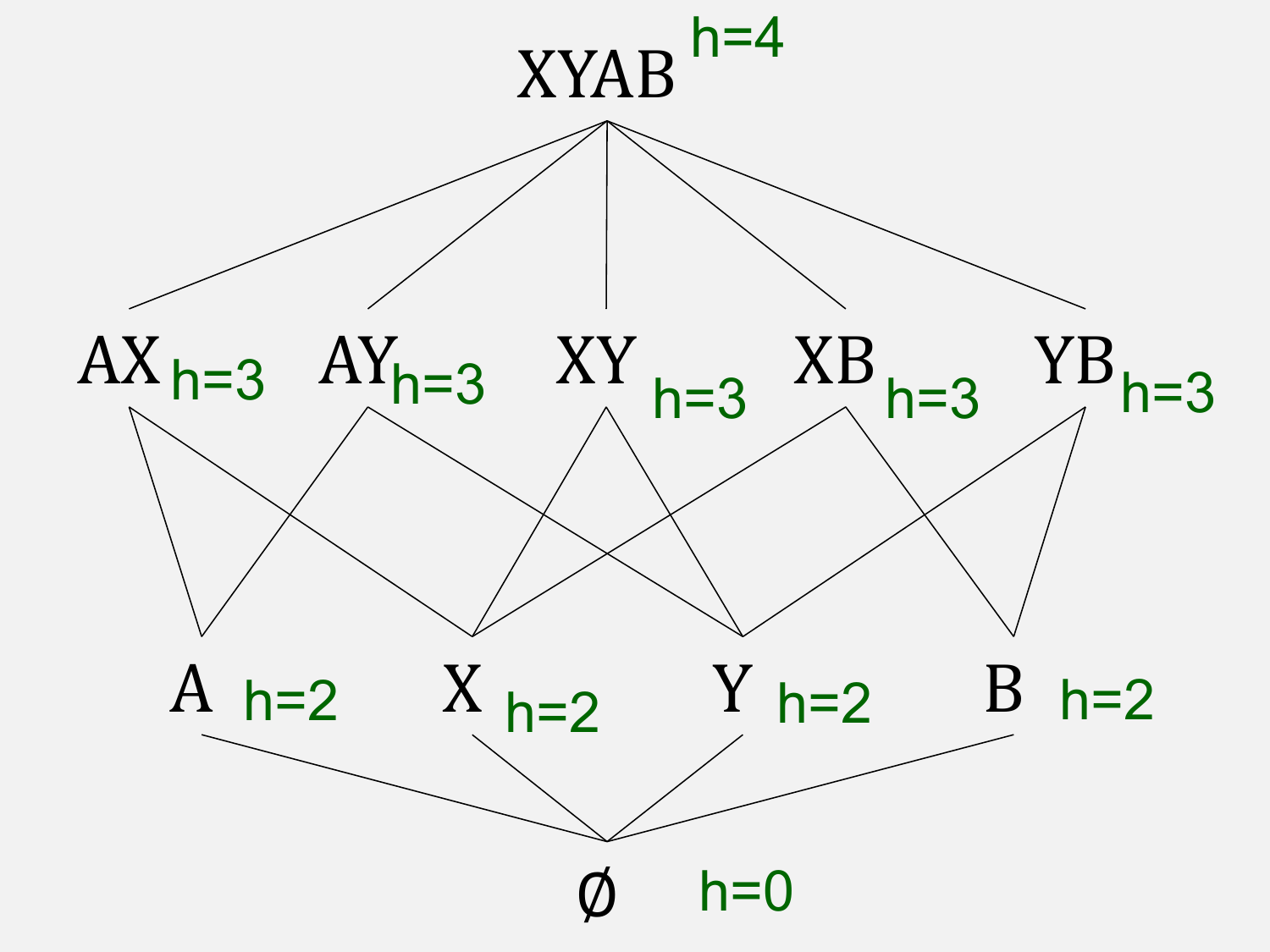}
    \caption{A lattice, and the polymatroid
      from~\cite{zhang1998characterization} defined on the lattice.}
  \label{fig:zhang:yeung:h}
\end{figure}

As a final comment, we note that it is instructive to check directly
that the polymatroid in Fig.~\ref{fig:zhang:yeung:h} fails to satisfy
the copy lemma, without using Zhang and Yeung's inequality; we
provide a direct proof in Appendix~\ref{app:eq:shannon:complicated}.

{\bf Group-theoretic characterization of information inequalities}
Chan and Yeung~\cite{DBLP:journals/tit/ChanY02} described an elegant
characterization of information inequalities in terms of group
inequalities.  Given a finite group $G$ and a subgroup
$G_1 \subseteq G$, a {\em left coset} is a set of the form $a G_1$,
for some $a \in G$.  By Lagrange's theorem, the set of left cosets,
denoted $G/G_1$, forms a partition of $G$, and $|G/G_1| = |G|/|G_1|$.
Fix $n$ subgroups $G_1, \ldots, G_n$, and consider
the relational instance:
\begin{align}
R = & \setof{(aG_1,\ldots, aG_n)}{a \in G} \label{eq:group:realization}
\end{align}
whose set of attributes we identify, as usual, with
$\bm X=\set{X_1, \ldots, X_n}$.  Notice that
$|R| = |G|/|\bigcap_{i=1,n}G_i|$.  The entropic vector $\bm h$
associated to the relation $R$ (Def.~\ref{def:uniform}) is called a
{\em group realizable entropic vector}, and the set of group
realizable entropic vectors is denoted by
$\Upsilon_n \subseteq \Gamma_n^*$, see Fig.~\ref{fig:diagram}.  One
can check that, for any subset of variables $\bm U \subseteq \bm X$,
$h(\bm U) =\log |G|/|\bigcap_{X_i \in \bm U} G_i|$.  The following was
proven in~\cite{DBLP:journals/tit/ChanY02}:

\begin{thm} \label{th:chan:groups}
  For any $\bm h \in \Gamma_n^*$ there exists a sequence
  $\bm h^{(r)} \in \Upsilon_n$, such that
  $\lim_{r \rightarrow \infty} \frac{1}{r} \bm h^{(r)} = \bm h$.
\end{thm}

It follows easily from the original proof that, if $\bm h$ satisfies a
set of functional dependencies, then so do all functions
$\bm h^{(r)}$, for $r \geq 0$; for completeness, we will include the
argument in Appendix~\ref{app:th:chan:groups}.


{\bf Open Problems} Characterizing the valid entropic information
inequalities is a major open problem.
Mat{\'{u}}s~\cite{DBLP:conf/isit/Matus07} proved that, for $n\geq 4$,
there are infinitely many independent non-Shannon inequalities.
Currently, the only techniques known for proving such inequalities
consists of repeated applications of Shannon inequalities and the Copy
Lemma.

A related open problem is the complexity of deciding Shannon
inequalities: what is the complexity of checking
$\Gamma_n \models \bm c \cdot \bm h \geq 0$, as a function of
$||\bm c||_1$?  It is implicit in the proof of
Theorem~\ref{th:primal:dual:bound:polyhedral} that this can be decided
in time exponential in $n$, but the complexity in terms of
$||\bm c||_1$ is open.
More discussion can be found in~\cite{DBLP:conf/icalp/KhamisK0S20}

\subsection{The Entropic Bound Is Asymptotically Tight}

\label{sec:th:tight:non-tight}

We prove here Theorem~\ref{th:tight:non-tight} item (1).  The plan is
the following.  We need to find a database $\bm D$ such that
$\log |Q(D)|$ comes close to
$\text{Log-L-Bound}_{\Gamma_n^*}(Q,\Sigma,\bm b)$. By definition,
there exists $\bm h \in \Gamma_n^*$ s.t. $h(\bm X)$ is close to
$\text{Log-L-Bound}_{\Gamma_n^*}(Q,\Sigma,\bm b)$.  We can't construct
a database $\bm D$ out of $\bm h$, because the probability
distribution realizing $\bm h$ may be non-uniform, instead we use Chan
and Yeung's theorem to approximate $r \bm h$ by a group realizable
vector $\bm h^{(r)}$, which is by definition associated to a relation
instance.  Hence, the need to amplify by the factor $r$.  However, if
we amplify, we don't know how
$\text{Log-L-Bound}_{\Gamma_n^*}(Q,\Sigma,r\bm b)$ grows.  Here we use
Theorem~\ref{th:primal:dual:bound}, showing that
$\text{Log-L-Bound}_{\Gamma_n^*}$ and
$\text{Log-U-Bound}_{\Gamma_n^*}$ are asymptotically equal, then use
the fact that $\text{Log-U-Bound}_{\Gamma_n^*}$ is linear, see
Eq.~\eqref{eq:linearity:log:u:bound}.  We give the details next.

By Corollary~\ref{cor:entropic:bound}, for all $k \in \N$:
\begin{align*}
  \frac{\sup_{\bm D: \bm D \models \bm B^k}\log |Q(\bm D)|}{\text{Log-U-Bound}_{\Gamma_n^*}(Q,\Sigma,k\bm b)}\leq &1 
\end{align*}
Together with Theorem~\ref{th:primal:dual:bound}
(Eq.~\eqref{eq:th:primal:dual:bound:entropic}) this implies:
\begin{align*}
 \sup_k \frac{\sup_{\bm D: \bm D \models \bm B^k}\log |Q(\bm D)|}{\text{Log-L-Bound}_{\Gamma_n^*}(Q,\Sigma,k\bm b)}\leq &1 
\end{align*}
To prove equality, it suffices to show that, $\forall\varepsilon > 0$,
$\exists k \in \N$ such that:
\begin{align}
  \log |Q(\bm D)| \geq & (1-\varepsilon)^4 \text{Log-L-Bound}_{\Gamma_n^*}(Q,\Sigma,k\bm b) \label{eq:suffices:to:show:U}
\end{align}

Let $U \defeq \text{Log-U-Bound}_{\Gamma_n^*}(Q,\Sigma,\bm b)$.  We
will assume that $\text{Log-U-Bound}_{\Gamma_n^*}(Q,\Sigma,\bm b)$ is
finite; otherwise, we let $U$ be an arbitrarily large number and the
proof below requires only minor adjustments, which we omit.  We will
assume w.l.o.g. that $U > 0$.  Recall that
$\text{Log-U-Bound}_{\Gamma_n^*}$ is
linear~\eqref{eq:linearity:log:u:bound}.  We prove:

\begin{claim} \label{claim:asymptotic:tight:1} For all
  $\varepsilon > 0$, there exists $k\in \N$, and a database $\bm D$
  such that $\bm D \models (\Sigma, \bm B^k)$ and
  $\log |Q(\bm D)| \geq (1-\varepsilon)^4 kU$
\end{claim}

Eq.~\eqref{eq:suffices:to:show:U} follows from
$kU = \text{Log-U-Bound}_{\Gamma_n^*}(Q,\Sigma,k\bm b) \geq
\text{Log-L-Bound}_{\Gamma_n^*}(Q,\Sigma,k\bm b)$.
It remains to prove Claim~\ref{claim:asymptotic:tight:1}.  Since
$\text{Log-L-Bound}_{\Gamma_n^*}$ and
$\text{Log-U-Bound}_{\Gamma_n^*}$ are asymptotically
equal~\eqref{eq:th:primal:dual:bound:entropic}, there exists
$k_0 \in \N$ such that
\begin{align}
  & \text{Log-L-Bound}_{\Gamma_n^*}(Q,\Sigma,k_0\bm b) \geq\nonumber \\
  & \geq  (1-\varepsilon)\text{Log-U-Bound}_{\Gamma_n^*}(Q,\Sigma,k_0\bm b)=  (1-\varepsilon)k_0 U \label{eq:k:k0:l:u}
\end{align}
By the definition of
$\text{Log-L-Bound}_{\Gamma_n^*}(Q,\Sigma,k_0\bm b)$
in~\eqref{eq:k:bound}, there exists $\bm h \in \Gamma_n^*$ such that:
\begin{align*}
  h(\bm X) \geq &(1-\varepsilon)\text{Log-L-Bound}_{\Gamma_n^*}(Q,\Sigma,k_0\bm b)  \geq (1-\varepsilon)^2 k_0U\\
  h(\sigma) \leq & k_0 b_\sigma, \ \ \ \forall \sigma \in \Sigma
\end{align*}

At this point we need the following Slack Lemma:
%
%
\begin{lmm}[Slack Lemma] \label{lemma:create:slack} For every
  $\bm h \in \Gamma_n^*$ and every $\varepsilon \in [0,1]$, there
  exists $k \in \N$ and $\bm h' \in \Gamma_n^*$ such that:
  \begin{align*}
    \bm h' \geq & (1-\varepsilon)k \bm h \\
    \forall \bm U, \bm V \subseteq \bm X:\ \ h'(\bm V|\bm U) \leq & (1-\varepsilon/2)kh(\bm V|\bm U)
  \end{align*}
\end{lmm}

\begin{proof}
  Assume w.l.o.g. that $\varepsilon > 0$, and set
  $k \defeq \ceil{\frac{1}{\varepsilon}}$ and
  $\bm h' \defeq (k-1)\bm h$.
  Then
  $\frac{1}{\varepsilon}\leq k \leq 1+\frac{1}{\varepsilon}\leq
  \frac{2}{\varepsilon}$ which implies
  $\varepsilon/2 \leq \frac{1}{k}\leq \varepsilon$. We have:
  \begin{align*}
    \bm h' = &  \left(1-\frac{1}{k}\right)k\bm h\geq (1-\varepsilon)k\bm h\\
    \bm h'(\bm V|\bm U) = & \left(1-\frac{1}{k}\right)k\bm h(\bm V|\bm U) \leq (1-\varepsilon/2)kh'(\bm V|\bm U)
  \end{align*}
\end{proof}

%

We apply the Slack Lemma to $\bm h$ and obtain a number $k_1$ and an
entropic vector $\bm h'$ such that:
\begin{align}
  h'(\bm X) \geq & (1-\varepsilon)k_1h(\bm X) \geq (1-\varepsilon)^3k_0k_1U \label{eq:local:mkk0u}\\
  h'(\sigma) \leq & (1-\varepsilon/2)k_1h(\sigma) \leq (1-\varepsilon/2)k_0k_1b_\sigma,\ \forall \sigma \in \Sigma \nonumber
\end{align}
Let
$g \defeq \min_{\bm U, \bm V \subseteq \bm X: h'(\bm V|\bm U) >
  0}h'(\bm V|\bm U)$ be the smallest non-zero value of
$h'(\bm V|\bm U)$.  By Chan and Yeung's theorem~\ref{th:chan:groups},
there exists a group realizable entropic vector $\bm h^{(r)}$ that
satisfies all the FDs satisfied by $\bm h'$, and
$||\bm h' - \frac{1}{r}\bm h^{(r)}||_\infty \leq \varepsilon g/ 4$.
Since $U>0$ we have $h'(\bm X)>0$ hence $h'(\bm X)\geq g$ and we
derive from~\eqref{eq:local:mkk0u}:
\begin{align*}
  \frac{1}{r} h^{(r)}(\bm X) \geq & h'(\bm X) - \varepsilon g/4 \geq  (1-\varepsilon/4)h'(\bm X) \nonumber\\
  \geq & (1-\varepsilon)^4k_0k_1U
\end{align*}
On the other hand,
$\frac{1}{r}h^{(r)}(\bm V|\bm U) \leq h(\bm V|\bm U) + \varepsilon
g/2$, for all sets $\bm U, \bm V$. We use
$\bm h' \models (1-\varepsilon/2)k_0k_1\bm b$ to prove
$\bm h^{(r)} \models rk_0k_1\bm b$.  Consider a statistics
$\sigma \in \Sigma$.  If $h'(\sigma) = 0$, then $\bm h'$ satisfies the
FD $\sigma$, and therefore $\bm h^{(r)}$ also satisfies this FD, thus
$h^{(r)}(\sigma) = 0\leq k_0k_1 b_\sigma$.  If $h'(\sigma) > 0$ then
$h'(\sigma) \geq g$ and the claim follows from:
\begin{align*}
  \frac{1}{r} h^{(r)}(\sigma) \leq & h'(\sigma) + \varepsilon g/2\nonumber\\
  \leq & h'(\sigma) + (\varepsilon/2) h'(\sigma) = (1+\varepsilon/2) h'(\sigma)\nonumber \\
  \leq & (1+\varepsilon/2)(1-\varepsilon/2)k_0k_1 b_{\sigma}\leq k_0k_1 b_{\sigma} 
\end{align*}
So far, we have:
\begin{align}
  h^{(r)}(\bm X) \geq & (1-\varepsilon)^4rk_0k_1U &   \bm h^{(r)}\models & rk_0k_1\bm b\label{eq:h:r:k:x}
\end{align}

To complete the proof of Claim~\ref{claim:asymptotic:tight:1}, we
construct the database $\bm D$ as follows.  Let the relation $R$ be
the group realization of $\bm h^{(r)}$
(Eq.~\eqref{eq:group:realization}).  For each relation $R_j(\bm Y_j)$,
define $R_j^{\bm D} \defeq \Pi_{\bm Y_j}(R)$.  By construction,
$Q(\bm D) = R$, and
$\log |Q(\bm D)| = h^{(r)}(\bm X) \geq (1-\varepsilon)^4 rk_0k_1U$
by~\eqref{eq:h:r:k:x}.  Furthermore, since $\bm h^{(r)}$ is
group-realized, for every statistics $\sigma \in\Sigma$, with guard
$R_\sigma$, we have
$\log \degree_{R_\sigma}(\sigma) = h^{(r)}(\sigma)\leq
rk_0k_1b_{\sigma}$; thus, $\bm D \models (\Sigma, \bm B^{rk_0k_1})$.
This implies:
\begin{align*}
  \sup_{\bm D: \bm D \models (\Sigma,\bm B^{rk_0k_1})}\log |Q(\bm D)| \geq & (1-\varepsilon)^4 rk_0k_1U
\end{align*}
proving Claim~\ref{claim:asymptotic:tight:1} for $k=rk_0k_1$.


\subsection{The Polymatroid Bound Is Not Asymptotically Tight}

We prove now Theorem~\ref{th:tight:non-tight} item (2).

\begin{prop}
  The following is a non-Shannon inequality:
  \begin{align}
    11h&(ABXYC) \leq \label{eq:non-shannon}\\
       & 3h(XY) + 3h(AX) + 3h(AY) \nonumber\\
    + & h(BX) + h(BY) + 5h(C)\nonumber\\
    + & (h(XYC|AB) + 4h(BC|AXY) + h(AC|BXY)) \nonumber \\
    + & (h(BXY|AC) + 2h(ABY|XC) + 2h(ABX|YC)) \nonumber
  \end{align}
\end{prop}

\begin{proof}
  Consider the following five inequalities:
  \begin{align*}
    0 \leq  & 3h(AX) + 3h(AY) - 4h(AXY) - h(A) \\
    + & h(BX) + h(BY) - h(BXY) \\
    - & h(AB) +3 h(XY) - 2h(X) - 2h(Y) \\
    0 \leq & h(A)+h(C) - h(AC) \\
    0 \leq & 2(h(X)+h(C) - h(XC)) \\
    0 \leq & 2(h(Y)+h(C) - h(YC)) \\
    11h&(ABXYC) =  11h(ABXYC)
  \end{align*}
  The first inequality holds because it is inequality
  Eq.~\eqref{eq:zhang:yeung}, expanded and re-arranged.  The next
  three inequalities are basic Shannon inequalities.  The last line is
  an identity.  A tedious but straightforward calculation shows that
  if we add the five (in)equalities above, then we
  obtain~\eqref{eq:non-shannon}, proving the claim.
\end{proof}

Consider the following query, derived from
inequality~\eqref{eq:non-shannon}:
\begin{align*}
  Q(A,B,X,Y,C) = & R_1(X,Y)\wedge R_2(A,X) \wedge R_3(A,Y) \\
  \wedge & R_4(B,X) \wedge R_5(B,Y) \wedge R_6(C) \\
  \wedge & R_7(A,B,X,Y,C)
\end{align*}
and the following statistics:
\begin{align*}
  \Sigma = & \{(XY), (AX), (AY), (BX), (BY),(C), \\
           & \ (XYC|AB), (BC|AXY), h(AC|BXY), \\
           & \ (BXY|AC), (ABY|XC), (ABX|YC)\}\\
  \bm b = & \{b_{XY} = b_{AX} = b_{AY} = b_{BX} = b_{BY} = 3, b_C = 2,\\
           & \ b_{XYC|AB} = b_{BC|AXY} = b_{AC|BXY} = 0, \\
           & \ b_{(BXY|AC)} = b_{ABY|XC} =  b_{ABX|YC} = 0\}
\end{align*}
In other words, we are given the cardinalities of $R_1, \ldots, R_6$,
but are not given the cardinality of $R_7$, instead we are told that
it satisfies the 6 FD's corresponding to the 6 conditional terms in
inequality~\eqref{eq:non-shannon}.  Consider any scale factor $k > 0$,
and the scaled log-statistics $k \bm b$.
Inequality~\eqref{eq:non-shannon} and the
definition~\eqref{eq:def:log-u-bound} imply:
\begin{align*}
  & \text{Log-U-Bound}_{\Gamma_n^*}(Q,\Sigma,k\bm b) \leq \\
  & \ \ k \frac{3b_{XY} + 3 b_{AX} + 3 b_{AY} +  b_{BX} + b_{BY} + 5 b_C}{11} = \frac{43k}{11}
\end{align*}
By Corollary~\ref{cor:entropic:bound}, for any database $\bm D$, if
$\bm D \models (\Sigma, \bm B^k)$ then:
\begin{align*}
  \log |Q(\bm D)|\leq & \frac{43k}{11}
\end{align*}

On the other hand, consider the polymatroid $k\bm h$, where $\bm h$ is
the polymatroid in Fig.~\ref{fig:h-abxyc}.  Since $h(ABXYC)=4$ and
$\bm h \models (\Sigma, \bm b)$, it follows that $kh(ABXYC) = 4k$, and
$k \bm h \models k \bm b$, therefore:
\begin{align*}
  \frac{\sup_{\bm D: \bm D \models \bm B^k}\log |Q(\bm D)|}{\text{Log-L-Bound}_{\Gamma_n}(Q,\Sigma,k\bm b)}\leq&\frac{43}{44}
\end{align*}
This implies Theorem~\ref{th:tight:non-tight} item (2).

\begin{figure}
  \centering
    \includegraphics[width=0.45\linewidth]{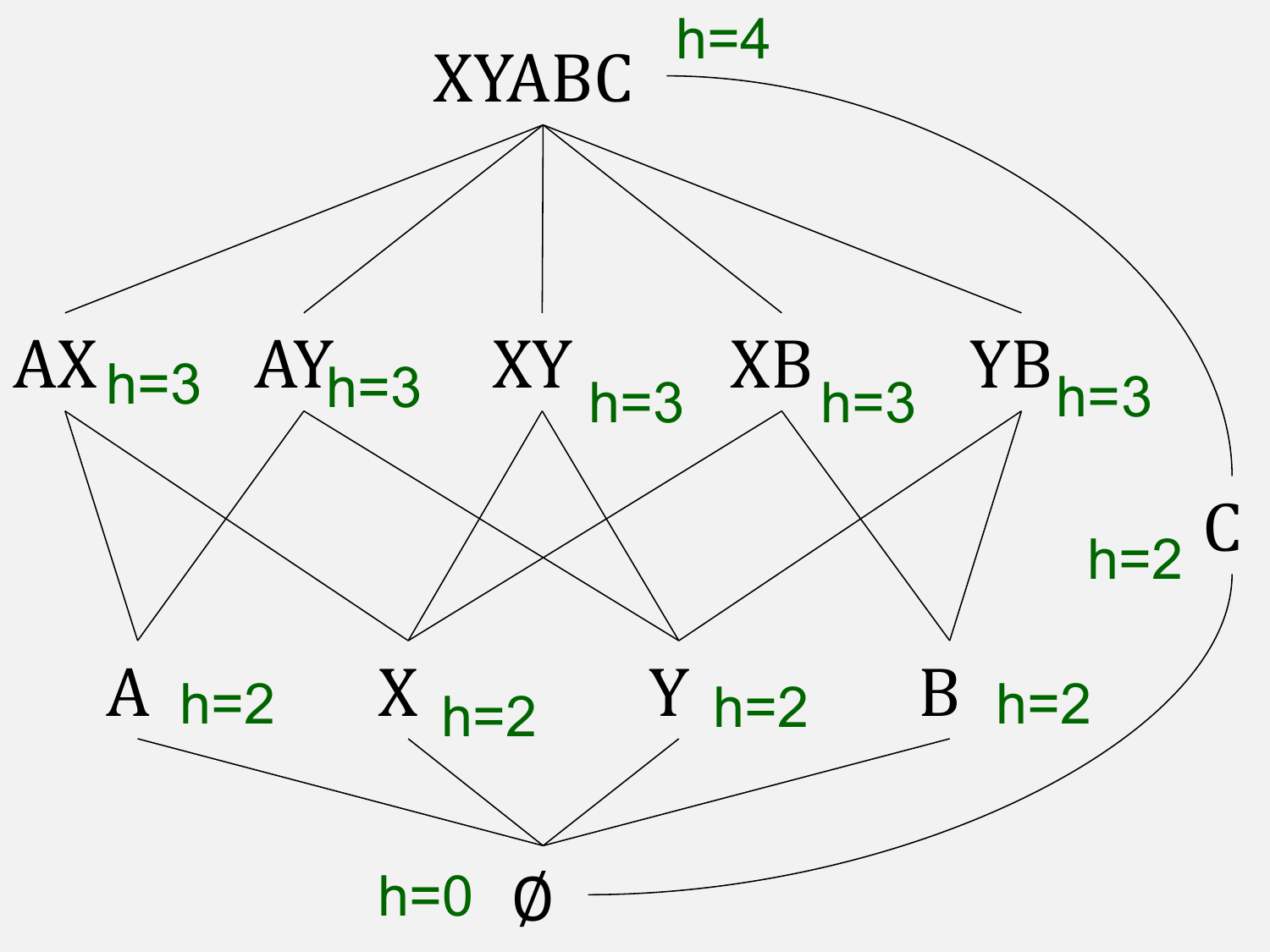}
  \caption{A polymatroid proving that the polymatroid bound is not tight.}
  \label{fig:h-abxyc}
\end{figure}

%

\section{Simple Inequalities}

\label{sec:special:cases}

We have a dilemma: the entropic bound is asymptotically tight, but it
is open whether it is computable, while the polymatroid bound is
computable, but is provably not tight in general.  We show in this
section that, under a reasonable syntactic restriction on the
statistics $\Sigma$, these two bounds are equal.  We do this by
describing a similar syntactic restriction for information
inequalities, which we call {\em simple inequalities}.  In that case
validity over entropic functions coincides with validity over
polymatroids, and we recover the stronger notion of tightness that we
had for the AGM bound.

%

\subsection{Background: Subclasses of Polymatroids}

\label{subsec:modular:normal}

A polymatroid $\bm h$ is called {\em modular} if the submodularity
inequality~\eqref{eq:submodularity} is an equality.  Equivalently,
$\bm h$ is modular if $\bm h \geq 0$ and for every subset
$\alpha \subseteq[n]$, $h(\bm X_\alpha) = \sum_{i \in\alpha}h(X_i)$.
We will denote by $M_n$ the set of modular polymatroids, see
Fig.~\ref{fig:diagram}.
For each $i=1,n$, we call the function $\bm h^{X_i}$ in
Fig.~\ref{fig:modular} a {\em basic modular function}; recall that
$h^{X_i}(\bm U) = 1$ when $X_i \in \bm U$ and $=0$ otherwise.  The
following holds:

\begin{prop}
  (1) A function $\bm h$ is modular iff it is a positive linear
  combination of basic modular functions,
  $\bm h = \sum_i a_i \bm h^{X_i}$, where $a_i \geq 0$ for all $i$.
  (2) Every modular function is entropic.
\end{prop}

\begin{proof}
  Item (1) is straightforward, but item (2) requires some thought.  It
  suffices to prove that $a \bm h^{X_i}$ is entropic for all real
  numbers $a \geq 0$.  For that purpose we need to describe one random
  variable $X_i$, whose entropy is $h(X_i) = a$.  Let $N$ be a natural
  number such that $\log N \geq a$, and consider the uniform
  probability space where $X_i$ has $N$ outcomes with the same
  probabilities, $p_i = 1/N$, $i=1,N$.  Replace $p_1$ by
  $p_1 + \theta$, and replace each $p_j$ with $j>1$ by
  $p_j - \theta/(N-1)$, for $\theta \in [0, 1-\frac{1}{N}]$.  When
  $\theta=0$ then the distribution is uniform and $h(X_j) = \log N$;
  when $\theta = 1-\frac{1}{N}$ then the distribution is
  deterministic, $p_1=1, p_2=\cdots=p_N = 0$, and $h(X_j) = 0$. By
  continuity, there exists some $\theta$ where $h(X_j)=a$.
\end{proof}

Fix a set of variables $\bm W \subseteq \bm X$.  The {\em step
  function at $\bm W$} is:
\begin{align}
  h_{\bm W}(\bm U) \defeq
  &
    \begin{cases}
      0 & \mbox{if $\bm U \subseteq \bm W$}\\
      1 & \mbox{otherwise}
    \end{cases} \label{eq:step:function}
\end{align}
There are $2^n-1$ non-zero step functions (since $\bm h_{\bm X}=0$).
$\bm h_{\bm W}$ is the entropy of the (uniform distribution of the)
following relation with 2 tuples:
\begin{align}
  R_{\bm W} \defeq R^{\bm X-\bm W} \defeq  & \ \ \ \ \ 
    \begin{array}{|cc|c} \cline{1-2}
      \bm W& \bm X - \bm W& p \\ \cline{1-2}
      0\cdots 0&0\cdots 0& 1/2 \\
      0\cdots 0&1\cdots 1& 1/2 \\ \cline{1-2}
    \end{array} \label{eq:r:w}
\end{align}
Sometimes it is convenient to use an alternative notation.  For a set
of variables $\bm V \subseteq \bm X$, define:
\begin{align}
  h^{\bm V}(\bm U) \defeq
  &
    \begin{cases}
      0 & \mbox{if $\bm U \cap \bm V = \emptyset$}\\
      1 & \mbox{otherwise}
    \end{cases} \label{eq:step:function:alt}
\end{align}
Then $\bm h^{\bm V} = \bm h_{\bm X-\bm V}$.  A basic modular function
$\bm h^{X_i}$ is the same as the step function $\bm h^{\set{X_i}}$; if
$|\bm V| \geq 2$ then $\bm h^{\bm V}$ is not modular.

\begin{defn} \label{def:normal:polymatroid}
  A {\em normal polymatroid} is a positive linear combination of step
  functions,
  \begin{align}
    \bm h = &\sum_{\bm V \subseteq \bm X, \bm V\neq \emptyset} a_{\bm V} \bm h^{\bm V}\label{eq:normal:polymatroid}
  \end{align}
  where $a_{\bm V}\geq 0$ for all $\bm V$.
\end{defn}

We denote by $N_n$ the set of normal polymatroids, see
Fig.~\ref{fig:diagram}.  Normal polymatroids are the same as {\em
  polymatroids with a non-negative I-measure} described
in~\cite{Yeung:2008:ITN:1457455,DBLP:journals/tods/KhamisKNS21}.

\begin{prop}
  The $2^n-1$ non-zero step functions $\bm h^{\bm V}$,
  $\bm V\neq \emptyset$ form a basis of the vector space
  $\setof{\bm h \in \R^{2^{[n]}}}{h(\emptyset)=0}$.  More precisely,
  every such vector $\bm h$ satisfies
  $\bm h = \sum_{\bm V} a_{\bm V}\bm h^{\bm V}$, where:
  \begin{align}
    a_{\bm U}  \defeq & - \sum_{\bm V \subseteq \bm  U}(-1)^{|\bm V|}h(\bm V|\bm X-\bm U) \label{eq:normal:coefficients}
  \end{align}
\end{prop}

The proof follows by solving the following system of linear equations
with unknowns $a_{\bm V}$:
\begin{align}
\forall \bm U\neq \emptyset:\ \ \  h(\bm U) = & \sum_{\bm V: \bm V \cap \bm U \neq \emptyset} a_{\bm V} \label{eq:coefficients:normal}
\end{align}
The solution is obtained by using M\"obius' inversion formula (we
prove this in Appendix~\ref{app:eq:normal:coefficients}) and consists
of the expression~\eqref{eq:normal:coefficients}.
Expression~\eqref{eq:normal:coefficients} is called {\em conditional
  interaction information}, and denoted by
$I(X_{i_1};X_{i_2};\cdots|\bm X-\bm U)$, where
$\bm U=\set{X_{i_1},X_{i_2},\ldots}$.  The following holds (the proof
is immediate and omitted):

\begin{prop}
  (1) A function $\bm h$ is a normal polymatroid iff, for every set
  $\bm U \subseteq \bm X$, $\bm U \neq 0$, the conditional interaction
  information~\eqref{eq:normal:coefficients} is $\geq 0$.  (2) Every
  normal polymatroid is entropic.
\end{prop}

\begin{ex}
  The parity function $\bm h$ Fig.~\ref{fig:parity} is the simplest
  example of a polymatroid that is not normal.  The 
  coefficients can be derived using~\eqref{eq:normal:coefficients},
%
  or, we can check directly that:
  \begin{align*}
    \bm h = & \bm h^{X,Y} + \bm h^{X,Z} + \bm h^{Y,Z} - \bm h^{X,Y,Z}
  \end{align*}
  The coefficient of $\bm h^{X,Y,Z}$ is negative, hence $\bm h$ is not
  normal.
\end{ex}

\subsection{Special Inequalities}

We describe here a class of information inequalities, called {\em
  simple inequalities}, were $\Gamma_n^*$-validity coincides with
$\Gamma_n$-validity.  The modular and normal polymatroids turn out to
be the key tools to study these inequalities.

The following is sometimes referred in the literature as the {\em
  modularization lemma}:

\begin{lmm} \label{lemma:modularizaation} For any polymatroid $\bm h$
  there exists a modular polymatroid $\bm h'$ such that (a)
  $\bm h' \leq \bm h$ and (b) $\bm h'(\bm X) = \bm h(\bm X)$.
\end{lmm}

\begin{proof}
  Order the variables arbitrarily $X_1, \ldots, X_n$ and define
  $h'(X_i) \defeq h(X_i|\bm X_{[1:i-1]})$, where
  $\bm X_{[1:i-1]}\defeq \set{X_1,\ldots,X_{i-1}}$.  We check
  condition (a): for $\alpha \subseteq [n]$,
  \begin{align*}
    h'(\bm X_\alpha) = & \sum_{i \in \alpha} h(X_i|\bm X_{[1:i-1]}) \\
    \leq & \sum_{i \in \alpha} h(X_i|\bm X_{[1:i-1]\cap \alpha}) = h(\bm X_\alpha)
  \end{align*}
  We check (b):
  $h'(\bm X) = \sum_{i=1,n} h(X_i|\bm X_{[1:i-1]}) = h(\bm X)$.
\end{proof}

The modularization lemma gives us an alternative, and more general
proof of Theorem~\ref{th:shearer}:

\begin{cor} \label{cor:unconditioned:inequalities} Consider an
  inequality of the form $\sum_i w_i h(\bm V_i) \geq h(\bm X)$, where
  $w_i \geq 0$ and $\bm V_i$ are subsets of $\bm X$.  The following
  conditions are equivalent:
  \begin{enumerate}[(1)]
  \item \label{item:unconditioned:1} The inequality is valid for
    polymatroids.
  \item \label{item:unconditioned:2} The inequality is valid for
    entropic functions.
  \item \label{item:unconditioned:3} The inequality is valid for
    modular functions.
  \end{enumerate}
\end{cor}

\begin{proof}
  The implications \ref{item:unconditioned:1} $\Rightarrow$
  \ref{item:unconditioned:2} $\Rightarrow$ \ref{item:unconditioned:3}
  are immediate.  We prove \ref{item:unconditioned:3} $\Rightarrow$
  \ref{item:unconditioned:1}, by contradiction: if the inequality
  fails on some polymatroid $\bm h$,
  $\sum_i w_i h(\bm V_i) < h(\bm X)$, and $\bm h'$ is the modular
  function in Lemma~\ref{lemma:modularizaation}, then,
  $\sum_i w_i h'(\bm V_i) \leq \sum_i w_i h(\bm V_i) < h(\bm X) =
  h'(\bm X)$ contradicting \ref{item:unconditioned:3}.
\end{proof}

We prove in Appendix~\ref{app:lemma:normalization} the following
extension of the Modularization Lemma:

\begin{lmm} \label{lemma:normalization} For any polymatroid $\bm h$
  there exists a normal polymatroid $\bm h'$ such that (a)
  $\bm h' \leq \bm h$, (b) $h'(\bm X) = h(\bm X)$, and (c)
  $h'(X_i) = h(X_i)$ for every variable $X_i \in \bm X$.
\end{lmm}

\begin{defn} \label{def:simple:inequalities} We call a set of
  statistics $\Sigma$ {\em simple} if, for all
  $(\bm V|\bm U) \in \Sigma$, $|\bm U| \leq 1$.  A {\em simple
    information inequality} is a $\Sigma$-inequality where $\Sigma$ is
  simple:
  \begin{align}
    \sum_{\sigma \in \Sigma} w_\sigma h(\Sigma) \geq & h(\bm X) \label{eq:simple:conditional}
  \end{align}
\end{defn}

We immediately derive:

\begin{cor} \label{cor:simple:inequalities} Given a simple
  inequality~\eqref{eq:simple:conditional}, the following are
  equivalent:
  \begin{enumerate}[(1)]
  \item The inequality is valid for polymatroids.
  \item The inequality is valid for entropic functions.
  \item The inequality is valid for normal polymatroids
  \end{enumerate}
\end{cor}

The proof is identical to that of
Corollary~\ref{cor:unconditioned:inequalities} and omitted.

\subsection{Special Databases}

\label{sec:normal:relation}

When the statistics $\Sigma$ are simple, then we show here that the
polymatroid and the entropic bound coincide.  We also show that the
bound is tight, using a similar notion of tightness as in the AGM
bound, where the ratio between the lower and upper bound depends only
on the query; also, there is no need to amplify the statistics values.
Moreover, like in the AGM bound, the worst-case database instance has
a special structure, which we call a {\em normal database}.
We start by showing:

\begin{thm} \label{th:simple:no:gap}
  If $\Sigma$ is simple, then:
  \begin{align*}
  & \text{Log-U-Bound}_{N_n}(Q,\Sigma,\bm b) =  \\
  & \ \ \ =\text{Log-U-Bound}_{\Gamma_n^*}(Q,\Sigma,\bm b) = \text{Log-U-Bound}_{\Gamma_n}(Q,\Sigma,\bm b)
  \end{align*}
\end{thm}

\begin{proof}
  Since $N_n \subseteq \Gamma_n^* \subseteq \Gamma_n$ we have
  inequalities above: $\cdots \leq \cdots \leq \cdots$
  Corollary~\ref{cor:simple:inequalities} implies
  $\text{Log-U-Bound}_{N_n}(Q,\Sigma,\bm
  b)=\text{Log-U-Bound}_{\Gamma_n}(Q,\Sigma,\bm b)$, hence all three
  quantities are equal.
\end{proof}

We describe now normal relational instances, and normal databases.
Start with a single relation $R(\bm X)$ with $n$ attributes $\bm X$.
Recall that an instance $R$ is a {\em product relation} if
$R = S_1 \times \cdots \times S_n$, for $n$ sets $S_1, \ldots, S_n$:
the worst-case instance of the AGM bound consisted of product
relations.  We generalize this concept:

\begin{defn}
  A relation instance $T$ with $n$ attributes is a {\em normal
    relation} if there exists $m$ finite sets $S_1, \ldots, S_m$ and a
  function $\psi:[n] \rightarrow 2^{[m]}$ such that
  \begin{align*}
    T = & \setof{(\bm s_{\psi(1)}, \bm s_{\psi(2)}, \ldots, \bm s_{\psi(n)})}{\bm s \in S_1 \times \cdots \times S_m}
  \end{align*}
\end{defn}
In a normal relation the values of an attribute can be tuples
themselves.  Every product relation is a normal relation, but not vice
versa.  A database instance is {\em normal} if each of its relations
is normal.  A {\em basic normal relation of size $N$} is the following:
\begin{align}
  T^{\bm V}_N \defeq & \setof{(k\cdot \one_{X_1\in\bm V},\cdots,k\cdot \one_{X_n\in\bm V})}{k=0,N-1}\label{eq:t:w}
\end{align}
Here $\one_{X_i \in \bm V}$ is an indicator variable that is 1 when
$X_i \in \bm V$ and 0 otherwise; thus, if an attribute $X_i$ is in
$\bm V$ then it takes the values $0,1,\ldots,N-1$ in the relation
$T^{\bm V}_N$, otherwise it has constant values 0.  The entropic
vector of $T^{\bm V}_N$ is $(\log N) \bm h^{\bm V}$.  In particular,
the relation $R^{\bm V}$ in~\eqref{eq:r:w} is $T^{\bm V}_2$.

\begin{ex}
  We give three examples of normal relations with $n=3$ attributes:
  \begin{align*}
    A = & \setof{(i,j,k)}{i,j,k \in [0:N-1]} && \mbox{product relation} \\
    B = & \setof{(i,i,i)}{i \in [0:N-1]},&&\mbox{normal relation} \\
    C = & \setof{(i,(i,j),j)}{i,j\in [0:N-1]} &&\mbox{normal relation}
  \end{align*}
  Their cardinalities are $|A|=N^3$, $|B| = N$, $|C|=N^2$.
  We also notice:
  \begin{align*}
    A = & T^X_N\otimes T^Y_N \otimes T^Z_N & B = & T^{XYZ}_N &  C = & T^{X,Y}_N \otimes T^{Y,Z}_N
  \end{align*}
\end{ex}

We prove that the lower bound for simple statistics is tight.

\begin{thm}\label{th:simple:lower:bound}  Let $\Sigma$ be a set of
  simple statistics for a query $Q$ and let $\bm B$ be statistics
  values.  Then there exists a worst-case instance $\bm D$ such that
  $|Q(\bm D)| \geq \frac{1}{2^{2^n-1}}\text{U-Bound}_{\Gamma_n^*}(Q,\Sigma,\bm B)$.
\end{thm}

\begin{proof} We use the following, whose proof is immediate:
  \begin{prop}
    Let $R(\bm X)$, $R'(\bm X)$ be relations over the same attributes
    $\bm X$.
    \begin{itemize}
    \item If $R, R'$ are normal relations, then $R \otimes R'$ is
      normal.
    \item
      $\degree_{R \otimes R'}(\sigma) = \degree_R(\sigma) \cdot
      \degree_{R'}(\sigma)$, for all $\sigma = (\bm V|\bm U)$.
    \end{itemize}
  \end{prop}

  Denote by $U \defeq \text{U-Bound}_{\Gamma_n^*}(Q,\Sigma,\bm B)$,
  $\bm b \defeq \log \bm B$, then
  $\log U = \text{Log-U-Bound}_{N_n}(Q,\Sigma,\bm
  b)=\text{Log-L-Bound}_{N_n}(Q,\Sigma,\bm b)$, by
  Theorem~\ref{th:simple:no:gap} and
  Theorem~\ref{th:primal:dual:bound:polyhedral} respectively.  Let
  $\bm h^* \in N_n$ be the optimal solution to the linear program
  defining $\text{Log-L-Bound}_{N_n}(Q,\Sigma,\bm b)$ (see
  Theorem~\ref{th:primal:dual:bound:polyhedral}), then
  $\bm h^* \models (\Sigma,\bm b)$ and $h^*(\bm X) = \log U$.  Since
  $\bm h^*$ is normal, it can be written as:
  \begin{align*}
    \bm h^* = & \sum_{\bm V\neq \emptyset}a_{\bm V} \bm h^{\bm V}, \ \
                \ \ a_{\bm V}\geq 0
  \end{align*}
  Then, $\log U = h^*(\bm X)= \sum_{\bm V\neq \emptyset}a_{\bm V}$,
  and $U = \prod_{\bm V} 2^{a_{\bm V}}$.

  For each set $\bm V$, $\emptyset \neq \bm V \subseteq \bm X$ we define:
  \begin{align*}
    b_{\bm V} \defeq  \floor{2^{a_{\bm V}}}, \ \ \ 
    P^{\bm V} \defeq & T^{\bm V}_{b_{\bm V}} && \mbox{basic normal relation~\eqref{eq:t:w}} \\
    R \defeq & \bigotimes_{\bm V} P^{\bm V} && \mbox{normal relation}
  \end{align*}
  Define the worst-case instance as $\bm D = (R_1^D, \ldots, R_m^D)$,
  where $R_j^D = \Pi_{\bm Y_j}(R)$.
  We first check that $\bm D$ satisfies the constraints, and for that
  let $\sigma \in \Sigma$ have witness $R_\sigma$, then:
  \begin{align*}
    \log & \degree_{R_\sigma}(\sigma) = \log \degree_R(\sigma) = \sum_{\bm V}  \log \degree_{P^{\bm V}}(\sigma)\\
    = & \sum_{\bm V} (\log b_{\bm V}) h^{\bm V}(\sigma) \leq \sum_{\bm V} a_{\bm V} h^{\bm V}(\sigma) =  h^*(\sigma) \leq b_\sigma
  \end{align*}
  Finally, we check the query's output size:
  \begin{align*}
    \log & |Q(\bm D)| =  \log |R| = \log \prod_{\bm V} |P^{\bm V}| = \sum_{\bm V}(\log b_{\bm V}) h^{\bm  V}(\bm X)
  \end{align*}
  Since $h^{\bm V}(\bm X) = 1$, this implies
  $|Q(\bm D)| = \prod_{\bm V} b_{\bm V} = \prod_{\bm
    V}\floor{2^{a_{\bm V}}} \geq \frac{1}{2^{2^n-1}}U$, because
  $\floor{2^{a_{\bm V}}} \geq \frac{1}{2}2^{a_{\bm V}}$.
\end{proof}

The reader may want to check the analogy with the worst-case instance
of the AGM bound: the optimal solution $\bm v^*$ there became here
$\bm h^*$, and the domain $V_i=[\floor{2^{v_i^*}}]$ defined for the
variable $X_i$ became here the normal relation $P^{\bm V}$.  As
before, we constructed the worst-case instance $\bm D$ without
amplifying the statistics, and $|Q(\bm D)|$ is within a constant,
which depends only on the query, of
$\text{U-Bound}_{\Gamma_n^*}(Q,\Sigma,\bm B)$.

{\bf Discussion} The restriction to simple statistics occurs naturally
in many applications.  Databases are often designed with simple keys
(consisting of a single attribute), and applications that use degrees
often consider only simple degrees.  The restriction to simple
statistics is often acceptable.

It remains open where one can extend this definition to richer classes
of statistics, or inequalities, while still preserving the property
that validity for entropic vectors is the same as validity for
polymatroids.  The set of statistics in Example~\ref{ex:pods2016} is
not ``simple'', yet the entropic bound coincides with the polymatroid
bound.  This (and other examples) suggests that other non-trivial
syntactic classes may exist where these two bounds agree.

\section{Query Evaluation}

\label{sec:query:evaluation}

The query evaluation problem is: given a conjunctive query $Q$,
evaluate it on a (usually large) database $\bm D$.  In this paper we
consider only the {\em data complexity}, where the query is fixed, and
the runtime is given as a function of the statistics of $\bm D$.
Database systems compute queries using a sequence of {\em binary
  joins}, of the form $C(X,Y,Z) = A(X,Y) \wedge B(Y,Z)$, which are
written as $C = A \bowtie B$.  Assuming all relations are pre-sorted,
the time complexity of the join is $\tilde O(|A|+|B|+|A \bowtie B|)$.
A {\em semi-join}, denoted $C = A \ltimes B$, is a join followed by
the projection on the attributes of the first relation, meaning
$C(X,Y) = \exists Z(A(X,Y) \wedge B(Y,Z))$.  A semijoin can be
computed in time $\tilde O(|A|)$.


A {\em Worst Case Optimal Join} (WCOJ) is an algorithm that evaluates
$Q$ in time no larger than its theoretical upper bound.  A sequence of
binary joins is usually not a WCOJ, because intermediate results may
be larger than the theoretical upper bound of the query.  For example
the upper bound for the triangle query in Example~\ref{ex:triangle} is
$N^{3/2}$, but if we evaluate it as $(R\bowtie S) \bowtie T$, the join
$R \bowtie S$ can have size $N^2$.

Any WCOJ algorithm represents an indirect proof of the query's upper
bound, since the size of the output cannot exceed the time complexity
of the algorithm.  For example, if we are given an algorithm for the
triangle query, together with a proof that its runtime is
$\tilde O(N^{3/2})$, then we have a proof that the size of the output
is also $\tilde O(N^{3/2})$.  This means that proving an upper bound
on the query's output is {\em inevitable} for designing a WCOJ.  We
show in this section that one can proceed in reverse: given a proof of
the upper bound, convert it into a WCOJ. We call this paradigm {\em
  From Proofs to Algorithms}.  Thus, the question to ask in designing
a WCOJ algorithm is: how do we {\em prove} an upper bound on the
query's output?  And how do we convert it into an algorithm?

\subsection{Generic Join}

Consider the setting of the AGM bound: we are given only cardinality
statistics on the base relations.  In that case, a proof of the upper
bound is a proof of $\sum_j w_j h(\bm Y_j) \geq h(\bm X)$, since it
implies $|Q| \leq \prod_j |R_j|^{w_j}$.  We gave a proof of this
inequality in Theorem~\ref{th:shearer}; the proof consists of
conditioning on the last variable $X_n$, then applying induction on
the remaining variables.  We convert that proof into an algorithm:
iterate $X_n$ over its domain, and compute recursively the residual
query.  This algorithm is called {\em Generic Join}, or GJ, and was
introduced by Ngo, R{\'{e}}, and Rudra
\cite{DBLP:journals/sigmod/NgoRR13}.  We describe it in detail next.

%
%

Fix a full conjunctive query with variables $\bm X$, which we write as
$Q = \bigjoin_{j=1,m} R_j$.  As usual, $\bm Y_j$ are the variables of
$R_j$.  Generic Join computes $Q$ as follows:

\begin{itemize}
\item Let $X_n$ be an arbitrary variable.
\item Partition the set of indices $j$ into $J_0$ and $J_1$:
  \\
  \null\hfill $J_0 \defeq \setof{j}{X_n \not\in \bm Y_j}$, \hfill  $J_1 \defeq \setof{j}{X_n \in \bm Y_j}$.\hfill\null
\item Compute the set $D = \bigcap_{j \in J_1}\Pi_{X_n}(R_j)$.  
\item For each value $x \in D$, do:
  \begin{itemize}
  \item Compute $R_j[x] := \Pi_{\bm Y_j-X_n}(\sigma_{X_n=x}(R_j))$, for
    $j \in J_1$.
  \item Denote $R_j[x] := R_j$ for $j \in J_0$.
  \item Compute the residual query $\bigjoin_{j=1,m} R_j[x]$.
  \end{itemize}
\end{itemize}

We invite the reader to check how the algorithm can be ``read off''
the proof of Theorem~\ref{th:shearer}.  To compute the runtime of the
algorithm, assume that the relations are given in listing
representation, sorted lexicographically using the attribute order
$X_n, X_{n-1}, \ldots, X_1$.  Then, the runtime, $T_n$, is:
\begin{align*}
  T_n(R_1,\ldots,R_m) = & T_{\text{intersection}} + \sum_x T_{n-1}(R_1[x],\ldots,R_m[x])
\end{align*}
By induction hypothesis:
\begin{align*}
  T_{n-1}(R_1[x],\ldots,R_m[x]) = & \tilde O\left(\prod_j |R_j[x]|^{w_j}\right)
\end{align*}
which leads to:
\begin{align*}
  T_n& = T_{\text{intersection}} +  \tilde O\left(\prod_{j \in J_0} |R_j|^{w_j} \sum_x \prod_{j \in J_1}|R_j[x]|^{w_j}\right)\\
  \leq  & T_{\text{intersection}} +  \tilde O\left(\prod_{j \in J_0} |R_j|^{w_j} \prod_{j \in J_1}\left(\sum_x |R_j[x]|\right)^{w_j}\right)\\
  = & T_{\text{intersection}} + \tilde O\left(\prod_j |R_j|^{w_j}\right)
\end{align*}
We used H\"older's inequality in Fig.~\ref{fig:friedgut} (since
$\sum_{j \in J_1} w_j \geq 1$, because $X_n$ is covered), and the fact
that $\sum_x |R_j[x]| = |R_j|$ for $j \in J_1$.  The crux of the
algorithm is the intersection: its runtime should not exceed
$\prod_j |R_j|^{w_j}$, and for that it suffices to iterate over the
smallest set $\Pi_{X_n}(R_j)$, and probe in the others: the runtime is
$\tilde O(\min_{j \in J_1} |R_j|) \leq \tilde O(\prod_{j\in
  J_1}|R_j|^{w_j})$, since $\sum_{j \in J_1} w_j \geq 1$.

\begin{ex}
  Using the variable order $X,Y,Z$, GJ computes the triangle query
  $R(X,Y)\wedge S(Y,Z) \wedge T(Z,X)$ as follows:
  \begin{align*}
    & \texttt{For } x\in \Pi_X(R) \cap \Pi_X(T)\texttt{ do:}\\
    & \ \ \ \texttt{For } y \in \Pi_Y(R[X=x]) \cap \Pi_Y(S) \texttt{ do:}\\
    & \ \ \ \ \ \texttt{For } z \in \Pi_Z(S[Y=y]) \cap \Pi_Z(T[X=x])  \texttt{ do:}\\
    & \ \ \ \ \ \ \ \texttt{output}(x,y,z)
  \end{align*}
\end{ex}

The choice of algorithm for computing the intersection is critical for
GJ.  To see this, consider the simplest query,
$Q(X) = R(X) \wedge S(X)$, that {\em is} an intersection.  The AGM
bound is $\min(|R|,|S|)$, corresponding to the edge covers $(1,0)$ and
$(0,1)$, and GJ must compute the query in time
$\tilde O(\min(|R|,|S|))$.  By assumption, $R, S$ are already sorted,
but we cannot run a standard merge algorithm, since its runtime is
$O(|R|+|S|)$; instead, we iterate over the smaller relation and do a
binary search in the larger.

Because of its simplicity and ease of implementation, GJ is the poster
child of WCOJ algorithms.  One remarkable property of GJ is that its
runtime is always bounded by the AGM bound, no matter what variable
order we choose.  Before GJ,
Veldhuizen~\cite{DBLP:conf/icdt/Veldhuizen14} described an algorithm
called Leapfrog Triejoin (LFTJ), which uses a similar logic as GJ, but
also specifies in the details of the required trie data structure.
Several implementations of GJ/LFTJ exists
today~\cite{DBLP:conf/sigmod/SchleichOC16,DBLP:conf/sigmod/ArefCGKOPVW15,DBLP:journals/pvldb/FreitagBSKN20,DBLP:journals/tods/MhedhbiKS21,DBLP:journals/corr/abs-2301-10841}.

\subsection{The Heavy/Light Algorithm}

\label{sec:heavy:light}

Balister and
Bollob{\'{a}}s~\cite{DBLP:journals/combinatorica/BalisterB12} provided
the following alternative proof of an inequality of the
form~\eqref{eq:shearer:ii}, which we write in an equivalent form using
integer coefficients:
\begin{align}
  E \defeq \sum_{j=1,m} k_j h(\bm Y_j) \geq & k_0 h(\bm X) \label{eq:unconditional:integer}
\end{align}
where $k_i \in \N$, for $i=0,m$.
View the expression $E$ as a bag of terms $h(\bm Y_j)$ where each term
$h(\bm Y_j)$ occurs $k_j$ times.  A {\em compression step} consists of
the following:
\begin{itemize}
\item Choose two terms $h(\bm U), h(\bm V) \in E$ such that
  $\bm U \not\subseteq \bm V$ and $\bm V \not\subseteq \bm U$.
\item Replace $h(\bm U) + h(\bm V)$ with
  $h(\bm U\cup \bm V)+h(\bm U \cap \bm V)$.
\end{itemize}

\begin{thm}~\cite{DBLP:journals/combinatorica/BalisterB12} Any
  sequence of compression steps eventually leads to:
  \begin{align}
    E =  & \ell_0 h(\bm Z_0) + \ell_1 h(\bm Z_1) + \cdots,&  \mbox{ where } & \bm Z_0 \supset \bm Z_1 \supset \cdots\label{eq:bb}
  \end{align}
  Furthermore, if each variable $X_i$ is covered at least $k_0\geq 1$
  times\footnote{Meaning: $\sum_{j: X_i \in \bm Y_j} k_j \geq k_0$.}
  by the original expression $E$ in~\eqref{eq:unconditional:integer},
  then $\bm Z_0=\bm X$ and $\ell_0 \geq k_0$; in particular, the
  inequality~\eqref{eq:unconditional:integer} is valid.
\end{thm}

\begin{proof} Each compression step strictly increases the quantity
  $\sum_{h(\bm Z) \in E} |\bm Z|^2$.  To see this, write
  $\bm U = \bm A \cup \bm C$, $\bm V = \bm B \cup \bm C$ where
  $\bm A,\bm B,\bm C$ are disjoint sets, then
  $|\bm U|^2 + |\bm V|^2 = (|\bm A|+|\bm C|)^2 + (|\bm B|+|\bm C|)^2$,
  while
  $|\bm U \cup \bm V|^2 + |\bm U \cap \bm V|^2 = (|\bm A|+|\bm B|+|\bm
  C|)^2 + |\bm C|^2$, and the latter is strictly larger when
  $|\bm A|\cdot |\bm B| > 0$.  This quantity cannot exceed
  $(\sum_j k_j)n^2$, therefore compression needs to terminate, and
  this happens when for any two sets in $E$ one contains the
  other. Then, $E$ must have the form~\eqref{eq:bb}.  Finally, we
  observe that compression preserves the number of times each variable
  $X_i$ is covered by $E$, because the number of sets in
  $\set{\bm U, \bm V}$ containing $X_i$ is the same as the number of
  sets in $\set{\bm U \cup \bm V, \bm U \cap \bm V}$ containing $X_i$.
  Therefore, if each variable is covered by $E$ at least $k_0$ times,
  then $\bm Z_0=\bm X$ and $\ell_0 \geq k_0$.
\end{proof}

Call a sequence of compression steps that converts an expression $E$
in~\eqref{eq:unconditional:integer} to~\eqref{eq:bb} a {\em BB-proof
  sequence}.  To derive an algorithm, we need to impose an additional
restriction.  Call a BB-proof sequence {\em divergent} if, after each
compression step
$h(\bm U) + h(\bm V) \rightarrow h(\bm U\cup \bm V)+h(\bm U \cap \bm
V)$, we can split $E$ into $E'+E''$, such that $E'$ contains
$h(\bm U\cup \bm V)$ and covers every variable at least $k_0'$ times,
$E''$ contains $h(\bm U \cap \bm V)$ and covers every variable at
least $k_0''$ times, and $k_0 = k_0'+k_0''$.

We convert a divergent BB-proof sequence into an algorithm called the
{\em Heavy/Light Algorithm}.  Let $\bm B$ be the statistics values,
$\bm b \defeq \log \bm B$, and set $B\defeq \max_j B_j$.  Assume
w.l.o.g. that the inequality
$E \defeq \sum_j k_jh(\bm Y_j)\geq k_0h(\bm X)$ is optimal, meaning
that
$\sum_j (k_j/k_0) b_j =
\text{Log-U-Bound}_{\Gamma_n}(Q,\Sigma,\log(\bm B))$ (see
Eq.~\eqref{eq:def:log-u-bound}).  Denote by $\bm h^*$ an optimal solution
to the dual, meaning $\bm h^* \models (\Sigma, \bm b)$ and
$h^*(\bm X) = \text{Log-L-Bound}_{\Gamma_n}(Q,\Sigma,\bm b)$ (see
Eq.~\eqref{eq:k:bound}).  These two quantities are the same by
Thm.~\ref{th:primal:dual:bound:polyhedral}:
$\sum_j (k_j/k_0) b_j=h^*(\bm X)$.  The algorithm uses a working
memory which stores, for each term $h(\bm Z)$ in $E$, a temporary
relation $S(\bm Z)$, called the {\em guard} of $h(\bm Z)$, and
maintains the invariant: $\log |S(\bm Z)| \leq h^*(\bm Z)$.
Initially, the working memory is $\setof{R_j}{k_j > 0}$: by
complementary slackness, if $k_j>0$ then the dual constraint
constraint is tight, $h^*(\bm Y_j) = b_j$, and the invariant holds
because $\log |R_j|\leq \log B_j = b_j = h^*(\bm Y_j)$.

The algorithm repeatedly processes a compression step
$h(\bm U)+h(\bm V)\rightarrow h(\bm U\cup \bm V)+ h(\bm U\cap \bm V)$
of the BB-sequence, as follows.  If the two guards are $S(\bm U)$ and
$S'(\bm V)$, let $\bm C \defeq \bm U \cap \bm V$, write
$\bm U = \bm A \bm C$, $\bm V = \bm B \bm C$, and define
$M \defeq 2^{h^*(\bm A|\bm C)}$.  Partition the guard $S(\bm A \bm C)$
into two subsets:
\begin{align*}
  S_{\text{light}}(\bm A, \bm C) = & \setof{(\bm a, \bm c) \in S}{\degree_S(\bm A|\bm C= \bm c) \leq M} \\
  S_{\text{heavy}}(\bm C) = & \setof{\bm c \in \Pi_{\bm C}(S)}{\degree_S(\bm A|\bm C= \bm c) > M}
\end{align*}
Compute new guards using a join and a semijoin:
\begin{align*}
S''(\bm A,\bm B,\bm C) := & S_{\text{light}}(\bm A,\bm C) \bowtie S'(\bm B,  \bm C)\\
S'''(\bm C) := & S_{\text{heavy}}(\bm C) \ltimes S'(\bm B, \bm C)
\end{align*}
The invariant holds because $|S''| \leq M \cdot |S|$ implies
$\log |S''| \leq h^*(\bm A|\bm C) + h^*(\bm B \bm C) \leq h^*(\bm A\bm
B\bm C)$, and because $|S_{\text{heavy}}|\leq |S|/M$ (since every
$\bm c \in S_{\text{heavy}}$ occurs $\geq M$ times in $S$) implies
$\log |S'''| \leq \log|S|-\log M = h^*(\bm A\bm C)-h^*(\bm A|\bm C) =
h^*(\bm C)$.  The runtime of the join and semijoin is
$\leq \tilde O(2^{h^*(\bm X)}) = \tilde
O(\text{U-Bound}_{\Gamma_n}(Q,\Sigma,\bm B))$.  Next, the algorithm
proceeds recursively, by processing independently $E'$ and $E''$,
semi-joins the result of $E'$ with the relations missing from $E'$,
similarly semi-joins the result of $E''$ with the relations missing
from $E''$, then returns the union of these two results.  Correctness
is easily checked.

\begin{figure}
  \centering
    \includegraphics[width=0.45\linewidth]{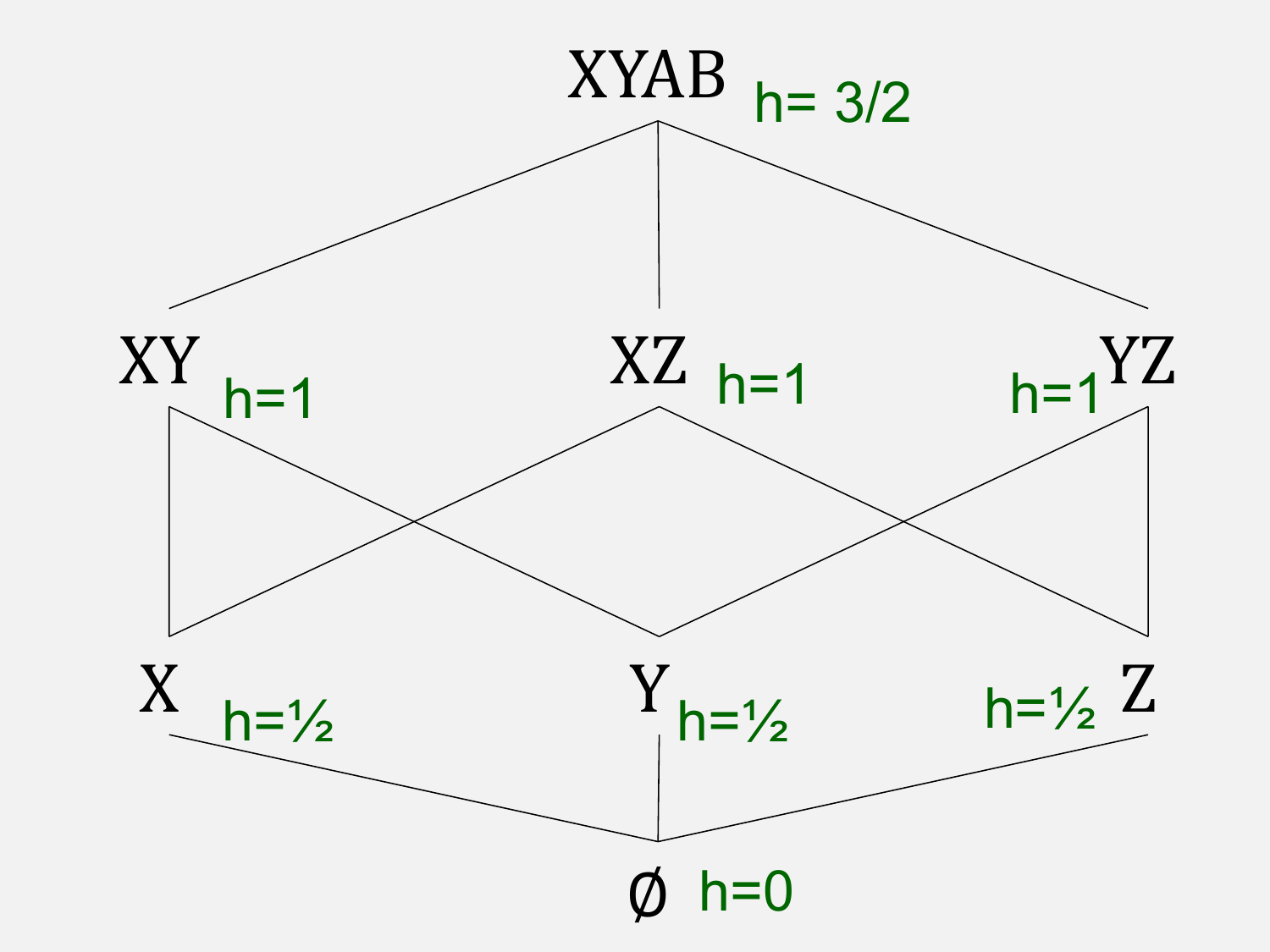}
  \caption{A simple polymatroid used in Example~\ref{ex:heavy:light}.}
  \label{fig:heavylight}
\end{figure}

\begin{ex} \label{ex:heavy:light} Consider the triangle query
  $R(X,Y)\wedge S(Y,Z) \wedge T(Z,X)$, and the following divergent
  proof:
  $\underline{h(XY)+h(YZ)}+h(ZX) \rightarrow h(XYZ) +
  \underline{h(Y)+h(ZX)} \rightarrow h(XYZ)+h(XYZ)$.  Assume for
  simplicity that the three relations have the same cardinalities
  $|R|=|S|=|T|=B$.  The optimal polymatroid is
  $\bm h^* \defeq \log B \cdot \bm h$, for $\bm h$ in
  Fig.~\ref{fig:heavylight}.
  
  The Heavy/Light Algorithm proceeds as follows. For the first
  compression step it partitioning $R$ into:
  \begin{align*}
    R_{\text{light}}:= & \setof{(x,y)}{\degree_R(X|Y=y) \leq 2^{h^*(Y|X)}= B^{1/2}}\\
    R_{\text{heavy}}:= & \setof{y}{\degree_R(X|Y=y) > 2^{h^*(Y|X)}= B^{1/2}}
  \end{align*}
  then it computes:
  \begin{align*}
    \texttt{Temp}_1(X,Y,Z)  := & R_{\text{light}}(X,Y) \bowtie S(Y,Z)\\
    \texttt{Temp}_2(Y) := & R_{\text{heavy}}(Y) \ltimes S(Y,Z)
  \end{align*}
  At this point the BB-proof diverged into two branches,
  $h(XYZ)\geq h(XYZ)$, and $h(Y) + h(ZX) \geq h(XYZ)$, and we perform
  a recursive call for each branch.  The first branch immediately
  returns $\texttt{Temp}_1$, which we semi-join with $T$:
  $Q'(X,Y,Z) = \texttt{Temp}_1(X,Y,Z) \ltimes T(Z,X)$.  The other
  branch applies the second compression step which corresponds to the
  following join operation:
  \begin{align*}
    \texttt{Temp}_3(X,Y,Z) := & \texttt{Temp}_2(Y) \times T(Z,X)
  \end{align*}
  which we semi-join with $R$ and $S$:
  $Q''(X,Y,Z) := \texttt{Temp}_3(X,Y,Z) \ltimes R(X,Y)\ltimes S(Y,Z)$.
  Finally, we return the union $Q' \cup Q''$.  The reader may verify
  that the runtime is $\tilde O(B^{3/2})$.
\end{ex}

An advantage of the Heavy/Light Algorithm over GJ is that it reuses
existing join operators, which already have very efficient
implementations in database systems.  However, the algorithm only
works for divergent BB-proofs.  This raises the question: does every
inequality~\eqref{eq:unconditional:integer} have a divergent proof?
The answer is negative, as provided by the following example due to
Yilei Wang~\cite{yilei-wang-2021}.

\begin{ex}
  The following has no divergent BB-proof:
  \begin{align*}
    E = h(XYZ)+h(ZUV)+h(VWX)+h(YUW) \geq & 2h(XYZUVW)
  \end{align*}
Assume w.l.o.g. that we start by compressing
$h(XYZ)+h(ZUV)\rightarrow h(XYZUV) + h(Z)$ (by symmetry, all other
choices are equivalent).  Then we need to partition $E$ into $E'+E''$.
Suppose $E'$ contains $h(Z)$; since $E'$ covers every variable, it
must contain both remaining terms $h(VWX)+h(YUW)$, which means that
$E''$ can only contain $h(XYZUV)$ alone, and it does not cover $W$.
\end{ex}

\subsection{PANDA}

Both Generic Join and the Heavy/Light Algorithm are restricted to
cardinality statistics, in other words they only work in the framework
of the AGM bound.  PANDA, introduced
in~\cite{DBLP:conf/pods/Khamis0S17}, is a WCOJ algorithm that works
for general statistics.  While it runs in time given by the
theoretical query upper bound, it also includes a polylogarithm factor
in the size of the database, with a rather large exponent.  We
describe PANDA here at a high level, and refer the reader
to~\cite{DBLP:conf/pods/Khamis0S17} for details.

Let $\Sigma$ be a set of statistics, and consider a
$\Sigma$-inequality with integer coefficients:
\begin{align}
  E \defeq  \sum_{\sigma \in \Sigma} k_\sigma h(\sigma) \geq & k_0 h(\bm X) \label{eq:conditional:inequality:integer}
\end{align}
A {\em CD-proof sequence} for the
inequality~\eqref{eq:conditional:inequality:integer} is a sequence of
steps that convert the LHS to the RHS, where each step is one of the
following:
\begin{itemize}
\item Composition: $h(\bm U) + h(\bm V|\bm U) \rightarrow h(\bm U\bm
  V)$.
\item Decomposition:
  $h(\bm U\bm V) \rightarrow h(\bm U) + h(\bm V|\bm U)$.
\item Submodularity:
  $h(\bm V|\bm U) \rightarrow h(\bm V|\bm U \bm W)$.
\item No-Op: $h(\bm U) \rightarrow 0$.
\end{itemize}
We say that the CD-proof sequence {\em proves} the
inequality~\eqref{eq:conditional:inequality:integer} if its starts
from the LHS and ends with the RHS.  The following was proven
in~\cite{DBLP:conf/pods/Khamis0S17}:

\begin{lmm}
  Inequality~\eqref{eq:conditional:inequality:integer} is valid for
  polymatroids iff it admits a CD-proof sequence.
\end{lmm}

PANDA converts a CD-proof sequence into an algorithm, similarly to the
way we converted a BB-sequence to an algorithm.  Given statistics
$\Sigma, \bm B$, guarded by the query $Q$
(Sec.~\ref{sec:entropic:polymatroid:bound}), assume that
inequality~\eqref{eq:conditional:inequality:integer} is the optimal
solution to $\text{Log-U-Bound}_{\Gamma_n}(Q,\Sigma,\log(\bm B))$;
otherwise, choose a better inequality.  Denote by $h^*$ an optimal
solution to $\text{Log-L-Bound}_{\Gamma_n}(Q,\Sigma,\log(\bm B))$.
%
%
The algorithm has a working memory consisting of a guard, call it
$S_{\bm V|\bm U}(\bm Z)$, for every term $h(\bm V|\bm U)$ in $E$,
satisfying the following invariant: $\bm V \subseteq \bm Z$ and there
exists a subset $\bm U_0 \subseteq \bm U \cap \bm Z$, such that:
\begin{align*}
  \log \degree_{S_{\bm V|\bm U}}(\bm V|\bm U_0) \leq & h^*(\bm V|\bm U)
\end{align*}
The guard need not have all variables $\bm U$, but only a subset
$\bm U_0$ that is sufficient to prove the bound on the max-degree.

Initially, the working memory consists of all guards $R_{\sigma}$ of
the statistics $\sigma = (\bm V|\bm U) \in \Sigma$, where
$k_\sigma > 0$.  By complementary slackness, if $k_\sigma > 0$, then
the corresponding constraint on $\bm h^*$ is tight,
$h^*(\sigma) = b_\sigma (\defeq \log B_{\sigma})$, therefore
$\log \degree_{R_\sigma}(\sigma)\leq b_\sigma = h^*(\sigma)$ because
the input database satisfies the statistics.  PANDA performs the
following action for each step of the CD-proof sequence:

\begin{description}
\item[Composition]
  $h(\bm U) + h(\bm V|\bm U) \rightarrow h(\bm U\bm V)$.  Compute the
  new guard as:
  \begin{align*}
    S_{\bm U\bm V} := & \Pi_{\bm U}(S_{\bm U}) \bowtie \Pi_{\bm U_0 \bm V}(S_{\bm V|\bm U})
  \end{align*}
  Since
  $|S_{\bm U\bm V}| \leq |S_{\bm U}|\cdot \degree_{S_{\bm V|\bm
      U}}(\bm V|\bm U_0)$, we have:
  \begin{align*}
    \log|S_{\bm U\bm V}| \leq & \log\degree_{S_{\bm U|\emptyset}}(\bm U|\emptyset) +  \log \degree_{S_{\bm V|\bm U}}(\bm V|\bm U_0)\\
    \leq & h^*(\bm U) + h^*(\bm V|\bm U) = h^*(\bm V)
  \end{align*}
  Thus, the invariant holds, and the runtime does not exceed the
  polymatroid bound, whose log is $h^*(\bm X)$.
\item[Submodularity] $h(\bm V|\bm U) \rightarrow h(\bm V|\bm U\bm W)$.
  Here PANDA only records that the new term $h(\bm V|\bm U\bm W)$ has
  the same guard as the old term $h(\bm V|\bm U)$.
\item[Decomposition]
  $h(\bm U\bm V) \rightarrow h(\bm U)+h(\bm V|\bm U)$.  Here PANDA
  first projects out the extra variables in the guard of
  $h(\bm U\bm V)$ and obtains a relation
  $S := \Pi_{\bm U\bm V}(S_{\bm U\bm V})$ whose size $N \defeq |S|$
  satisfies $\log N \leq h^*(\bm U\bm V)$.  Next, it performs {\em
    regularization}: partition $S$ into $\log N$ fragments
  $S = \bigcup_{i=1,\log N} S_i$, where:
  \begin{align*}
  S_i(\bm U,\bm V) \defeq & \setof{(\bm u, \bm v)\in S}{\degree_S(\bm V|\bm U=\bm u) \in [2^{i-1},2^i]}
  \end{align*}
  PANDA then continues with $\log B$ recursive calls.  The $i$'th
  recursive call replaces $S$ with $S_i$ in the query, adds two new
  statistics $(\bm V|\bm U)$ and $(\bm U)$ to $\Sigma$, and two
  log-statistics values, $b_{\bm V|\bm U} \defeq i$ and
  $b_{\bm U} \defeq N/2^{i-1}$, both with guard $S_i(\bm U \bm V)$.
  Then, PANDA computes a {\em new} optimal primal/dual solutions to
  the polymatroid bound, resulting in a new
  inequality~\eqref{eq:conditional:inequality:integer} and a new
  polymatroid $\bm h^*$.  It uses these to compute the residual query
  where $S$ is replaced by $S_i$.  Finally, it returns the union of
  all $\log N$ results from all recursive calls.
\end{description}

We leave out several details of PANDA, including the proof of
termination, and refer the reader to~\cite{DBLP:conf/pods/Khamis0S17}.
We also note that PANDA was extended from computing full conjunctive
queries, to computing Boolean conjunctive queries, with a runtime
given by the {\em submodular width} of the query, a notion introduced
by Marx~\cite{DBLP:journals/jacm/Marx13}.

\section{The Domination Problem}

\label{sec:domination}

We now move beyond the query upper bound problem, and consider a
related question, called the {\em domination problem}: given two
queries $Q, Q'$, check if, for any database $\bm D$,
$|Q(\bm D)| \leq |Q'(\bm D)|$.  The queries $Q$ and $Q'$ need not have
the same number of variables.  In this section we consider full
conjunctive queries that may have self-joins, i.e. the same relation
name may occur several times in the query; for example in
$R(X,Y)\wedge R(Y,Z)$ the same relation $R$ occurs twice.

\begin{defn}
  Given two conjunctive queries $Q(\bm X)$, $Q'(\bm Y)$ we say that
  $Q'$ {\em dominates} $Q$, and write $Q \preceq Q'$, if for every
  database instance $\bm D$, $|Q(\bm D)| \leq |Q'(\bm D)|$.
\end{defn}

The original motivation for the domination problems comes from the
query containment problem under bag semantics.  Given a (not
necessarily full) conjunctive query $Q(\bm Y_0)$, as in~\eqref{eq:cq},
its value under {\em bag semantics} is a bag of tuples, where each
tuple $\bm y_0$ occurs as many times as the number of homomorphisms
from $Q$ to $\bm D$ that map $\bm Y_0$ to $\bm y_0$.  SQL uses bag
semantics.  Chaudhuri and Vardi~\cite{DBLP:conf/pods/ChaudhuriV93}
were the first to study the {\em query containment problem under bag
  semantics}: given $Q, Q'$, check whether
$Q(\bm D) \subseteq Q'(\bm D)$ for every $\bm D$, where both
$Q(\bm D)$, $Q'(\bm D)$ are bags of tuples.  This problem has been
intensively studied in the last thirty years.  It has been shown that
the containment problem under bag semantics is undecidable for {\em
  unions of conjunctive
  queries}~\cite{DBLP:journals/tods/IoannidisR95}, and for {\em
  conjunctive queries with
  inequalities}~\cite{DBLP:conf/pods/JayramKV06}; both used reduction
from Hilbert's 10th Problem.  It should be noted that, under set
semantics, the containment problem for these two classes of queries is
decidable.

When $Q() = \ldots$ is a Boolean query, then under standard
set-semantics it returns either $\set{}$ or $\set{()}$, representing
FALSE and TRUE.  Under bag semantics it may return a bag
$\set{(),(),\ldots,()}$, representing a number, and this number is
equal to the size of the output of the full query,
$Q(\bm X) = \ldots$.  Based on this discussion, the domination problem
$Q \preceq Q'$ for full conjunctive queries is the same as the query
containment problem under bag semantics for Boolean queries.

Kopparty and Rossman~\cite{DBLP:journals/ejc/KoppartyR11} were the
first to establish the connection between the domination problem and
information theory.  We describe this connection, following their
example.

\begin{ex} \label{ex:vee} This example
  from~\cite{DBLP:journals/ejc/KoppartyR11} is attributed to Eric Vee.
  Consider the following queries:
  \begin{align*}
    Q(X,Y,Z) = & R(X,Y) \wedge R(Y,Z) \wedge R(Z,X) \\
    Q'(U,V,W) = & R(U,V) \wedge R(U,W)
  \end{align*}
  We will show that $Q \preceq Q'$.  Chaudhuri and
  Vardi~\cite{DBLP:conf/pods/ChaudhuriV93} already noted that, if
  there exists a surjective homomorphism $Q'\rightarrow Q$, then
  $Q \preceq Q'$.  In our example we have three homomorpisms
  $\varphi_1, \varphi_2, \varphi_3 : Q \rightarrow Q'$, but none of
  them is surjective.

  Consider the following linear expression in entropic terms, defined
  over the variables $U,V,W$ in $Q'$:
  \begin{align*}
    E \defeq & h(UV)+h(UW) - h(U) = h(UV)+h(W|U)
  \end{align*}
  (The expression is derived from the tree decomposition of $Q'$, as
  we explain below.)  For each of the three homomorphism $\varphi_i$,
  denote by $E\circ \varphi_i$ the result of substituting the
  variables $U,V,W$ in $E$ with
  $\varphi_i(U),\varphi_i(V),\varphi_i(W)$.

\begin{claim}
  The following inequality holds for all polymatroids:
  \begin{align}
    h(XYZ) \leq & \max(E \circ \varphi_1,E \circ \varphi_2,E \circ \varphi_3) \label{eq:kopparty:rossman}
  \end{align}
\end{claim}

\begin{proof} We expand:
    \begin{align*}
      \max&(E \circ \varphi_1,E \circ \varphi_2,E \circ \varphi_3) = \\
      = & \max(h(XY)+h(Y|X),h(YZ)+h(Z|Y),h(XZ)+h(X|Z))\\
      \geq & \frac{1}{3}\left(h(XY)+h(YZ)+h(ZX)+h(Y|X)+h(Z|Y)+h(X|Z)\right)\\
      = & \frac{1}{3}\left((h(XY)+h(Z|Y))+(h(YZ)+h(X|Z))+(h(ZX)+h(Y|X))\right)\\
      \geq & h(XYZ)
    \end{align*}
  where the last inequality follows from
  $h(XY)+ h(Z|Y) \geq h(XY)+h(Z|XY) = h(XYZ)$ and similarly for the
  other two terms.
\end{proof}

To prove $Q\preceq Q'$, consider a database instance $\bm D$, let
$N \defeq |Q(\bm D)|$, and consider the uniform probability
distribution $(Q(\bm D),p)$.  Its entropy $\bm h$ satisfies
inequality~\eqref{eq:kopparty:rossman}: assume w.l.o.g. that
$h(XYZ) \leq E\circ \varphi_1 = h(XY)+h(Y|X)$ (the other two cases are
similar).  We use $\varphi_1$ to define a probability space
$(Q'(\bm D), p')$: for every three constants $u,v,w$ in the instance
$\bm D$ s.t. $p(X=u)\neq 0$, define
\begin{align*}
    p'(U=u,V=v,W=w) \defeq & \frac{p(X=u,Y=v)p(X=u,Y=w)}{p(X=u)}
\end{align*}
Thus, $V\perp W | U$, the distribution of $UV$ is the same as that of
$XY$, and the distribution of $UW$ is also the same as that of $XY$.
(This is similar to the Copy Lemma~\ref{lmm:copy}.) Denoting by
$\bm h'$ the entropic vector associated to $p'$, we derive:
  \begin{align*}
    \log |Q'(\bm D)| \geq & h'(UVW) = h'(VW|U)+h(U) \\
    = & h'(V|U)+h'(W|U)+h'(U) \mbox{\ \ \ because $V\perp W | U$}\\
    = & h(Y|X) + h(Y|X) + h(X)=  h(XY) + h(Y|X)\\
    = & E\circ \varphi_1  \geq  h(XYZ) = \log|Q(\bm D)|
  \end{align*}
\end{ex}

We generalize Example~\ref{ex:vee}.  A {\em tree decomposition} of a
query $Q(\bm X) = \bigwedge_j R_j(\bm Y_j)$ is a pair $(T, \chi)$,
where $T$ is a tree and $\chi: \text{Nodes}(T) \rightarrow 2^{\bm X}$
such that every atom $R_j(\bm Y_j)$ is {\em covered}, meaning
$\exists n$, $\bm Y_j \subseteq \chi(n)$, and for any variable
$X \in \bm X$, the set of nodes
$\setof{n\in \text{Nodes}(T)}{X \in \chi(n)}$ induces a connected
subgraph of $T$.  Each set $\chi(n)$ is called a {\em bag}.  $Q$ is
{\em chordal} if it admits a tree decomposition where every bag
$\chi(n)$ induces a clique in the Gaifman graph of $Q$; equivalently,
for any two variables $X,Y \in \chi(n)$ the query has a predicate that
contains both $X,Y$.  A chordal query has a canonical tree
decomposition where the bags are the maximal cliques.  $Q$ is called
{\em acyclic}\footnote{More precisely, it is called
  $\alpha$-acyclic~\cite{DBLP:journals/jacm/Fagin83}.} if there exists
a tree decomposition where each bag is precisely one atom of the
query, $\chi(n) = \bm Y_j$ for some $j$.  An acyclic query is, in
particular, chordal.

Fix a query $Q(\bm U)$ with variables $\bm U$ and a tree decomposition
$T$.  We define the following expression of entropic terms:
\begin{align}
  E_T \defeq & \sum_{n \in \text{Nodes}(T)} h(\chi(n)) - \sum_{(n,n') \in \text{Edges}(T)} h(\chi(n)\cap \chi(n')) \label{eq:et:difference}
\end{align}
Equivalently, choose a root node for $T$ and orient all
edges to point away from the root.  Then:
\begin{align*}
  E_T = & \sum_{n \in \text{Nodes}(T)}h(\chi(n)|\chi(n)\cap \chi(\text{Parent}(n)))
\end{align*}
where we set $\chi(\text{Parent}(\text{Root})) \defeq \emptyset$.  The
following holds:

\begin{thm} Let $Q(\bm X), Q'(\bm U)$ be two full conjunctive queries,
  over variables $\bm X$ and $\bm U$ respectively.
  \begin{itemize}
  \item \cite{DBLP:journals/ejc/KoppartyR11,DBLP:journals/tods/KhamisKNS21} Let $T$ be a tree
    decomposition for $Q'$.  If the following inequality holds for all
    entropic vectors $\bm h$:
    \begin{align}
      h(\bm X) \leq & \max_{\varphi: \varphi \in \hom(Q',Q)}E_T \circ \varphi \label{eq:bag:containment:h}
    \end{align}
    then $Q'$ dominates $Q$, $Q \preceq Q'$.
  \item \cite{DBLP:journals/tods/KhamisKNS21} If $Q'$ is chordal and
    $Q\preceq Q'$, then inequality~\eqref{eq:bag:containment:h} holds
    for all entropic vectors, where $T$ is the canonical tree
    decomposition of $Q'$ consisting of its maximal cliques.  In other
    words,~\eqref{eq:bag:containment:h} is a necessary and sufficient
    condition for dominance.
  \end{itemize}
\end{thm}

Let's call a tree decomposition $T$ {\em simple} if for every edge
$(n,n') \in \text{Edges}(T)$, $|\chi(n) \cap \chi(n')| \leq 1$.  If
$Q'$ admits a simple tree decomposition, then
condition~\eqref{eq:bag:containment:h} is decidable; the proof follows
immediately from Lemma~\ref{lemma:normalization}.  This implies:

\begin{cor}  Assume that $Q'$ is chordal and admits a simple tree
  decomposition.  Then it is decidable whether $Q \preceq Q'$.
  Moreover, if $Q \not\preceq Q'$, then there exists a normal database
  instance (Sec.~\ref{sec:normal:relation}) such that $|Q(\bm
  D)|>|Q'(\bm D)|$.
\end{cor}

Finally, we remark that the connection between the query domination
problem and information inequalities is very tight.  The following was
proven in~\cite{DBLP:journals/tods/KhamisKNS21}:

\begin{thm}
  The following problems are computationally equivalent.  (1) Check if
  an inequality of the form
  $\max_{j=1,p}(\bm c^{(j)} \cdot \bm h) \geq 0$ is valid for all
  entropic vectors $\bm h$, where
  $\bm c^{(1)}, \ldots \bm c^{(p)} \in \R^{2^{[n]}}$ are $p$ vectors.
  (2) Given two queries $Q, Q'$ where $Q'$ is acyclic, check whether
  $Q \preceq Q'$.
\end{thm}

It is currently open whether these problems are decidable.

\section{Conditional Inequalities and Approximate Implication}

\label{sec:conditional:inequalities}

Our last application of information inequalities is for the
approximate implication problem, which can be described informally as
follows.  Let $\sigma_1, \ldots, \sigma_p$ be some constraints on the
database (we will define shortly what constraints we consider), and
suppose we have a proof of the implication
$\bigwedge_i \sigma_i \Rightarrow \sigma_0$.  The question is, if the
database $\bm D$ satisfies the constraints $\sigma_i$ only
approximatively, is it the case that that $\sigma_0$ also holds
approximatively?  We will show here that this question is related to
{\em conditional information inequalities}, whose study requires us to
do another deep dive into the space of polymatroids and entropic
functions.  We start by defining a conditional inequality:

\begin{defn} \label{def:cond:ii}
  A {\em conditional information inequality} is an assertion of the
  following form:
  \begin{align}
    \bm c_1 \cdot \bm h \geq 0 \wedge \cdots \bm c_p \cdot \bm h \geq  0 & \Rightarrow \bm c_0 \cdot \bm h \geq 0 \label{eq:cond:ii}
  \end{align}
  where $\bm c_i \in \R^{2^{[n]}}$ for $i=0,p$ are vectors.
\end{defn}
Sometimes it will be more convenient to replace $\bm c_i$ by
$-\bm c_i$, and write the implication as
$\bigwedge_i \bm c_i \cdot \bm h \leq 0 \Rightarrow \bm c_0 \cdot \bm
h \leq 0$.  As before, the validity of a conditional inequality
depends on the domain of $\bm h$, e.g. it can be valid for
polymatroids, or entropic functions, etc.

The first non-Shannon inequality discovered was a conditional
inequality~\cite{DBLP:journals/tit/ZhangY97}, predating the first
unconditioned Shannon inequality~\cite{zhang1998characterization}.
Kaced and Romashchenko~\cite{DBLP:journals/tit/KacedR13} showed the
first examples of essentially conditional inequalities (explained
below).  We start by describing the connection between conditional
inequalities and the constraint implication problem in databases, then
study the {\em relaxation problem}, a technique for transferring exact
inferences to approximate judgments.  We end with the proof of
Theorem~\ref{th:primal:dual:bound}, which we have postponed until we
developed sufficient technical machinery.

\subsection{The Constraint Implication Problem}

An {\em integrity constraint}, $\sigma$, is an assertion about a
relation $R(\bm X)$ that is required to hold strictly.  The
constraints considered here are Functional Dependencies (FD), already
reviewed in Sec.~\ref{sec:problem}, and Multivalued Dependencies
(MVD).  An {\em MVD} is a statement
$\sigma = (\bm U \mvd \bm V|\bm W)$ where $\bm U \cup \bm V \cup W$
form a partition of $\bm X$.  A relation instance $R$ {\em satisfies}
the MVD, $R \models \sigma$, if
$R = \Pi_{\bm U\bm V}(R) \bowtie \Pi_{\bm U\bm W}(R)$.

The {\em implication problem} asks whether a set of FDs and/or MVDs
$\sigma_i$, $i=1,p$, implies another FD or MVD $\sigma_0$:
\begin{align}
  \sigma_1 \wedge \cdots \wedge \sigma_p \Rightarrow & \sigma_0 \label{eq:fd:mvd:exact:implication}
\end{align}
Armstrong's axioms~\cite{DBLP:journals/tods/ArmstrongD80} are complete
for the implication problem for FDs, while Beeri et
al.~\cite{DBLP:conf/sigmod/BeeriFH77} gave a complete axiomatization
for both FDs and MVDs, and showed that the implication problem is
decidable.  In contrast,
Herrmann~\cite{DBLP:journals/iandc/Herrmann06} showed that the
implication problem of {\em Embedded MVDs} is undecidable; we do not
discuss EMVDs here.

Lee~\cite{DBLP:journals/tse/Lee87} showed the following connection
between information theory and constraints.  Fix a relational instance
$R$, and let $\bm h$ be its associated (uniform) entropic vector.
Then $R \models \bm U \fd \bm V$ iff $h(\bm V|\bm U)=0$, and
$R\models \bm U \mvd \bm V | \bm W$ iff $I_h(\bm V;\bm W|\bm U)=0$.
Therefore, every implication problem for FDs and MVDs can be stated as
a conditional information inequality.  For example, the {\em
  augmentation axiom}~\cite{DBLP:conf/sigmod/BeeriFH77} states
\begin{align*}
(A\mvd B|CD) \Rightarrow & (AC\mvd B|D)
\end{align*}
and is equivalent to the following conditional
inequality:\footnote{This has the form in Def.~\ref{def:cond:ii} once
  we write it as $-I(B;CD|A) \geq 0$ implies $-I(B;D|AC) \geq 0$.}
\begin{align}
  I(B;CD|A) = 0 \Rightarrow & I(B;D|AC) = 0 \label{eq:ex:implication}
\end{align}
This can be proven immediately by observing that the identity
$I(B;CD|A) = I(B;C|A) + I(B;D|AC)$ implies:
\begin{align*}
  I(B;D|AC) \leq & I(B;CD|A)
\end{align*}
Since both terms are $\geq 0$, the
implication~\eqref{eq:ex:implication} follows.

Beyond database applications, Conditional Independencies (CI) are
commonly used in AI, Knowledge Representation, and Machine Learning.
A CI is an assertions of the form $X \perp Y \mid Z$, where $X, Y, Z$
are three random variables, stating that $X$ is independent of $Y$
conditioned on $Z$.  The AI community has extensively studied the
implication problem for CIs.  It was shown that the implication
problem is decidable and finitely axiomatizable for {\em saturated}
CIs~\cite{GeigerPearl1993} (where $XYZ=$ all variables), but not
finitely axiomatizable in
general~\cite{StudenyCINoCharacterization1990}.

\subsection{The Relaxation Problem}

How can we prove a conditional inequality~\eqref{eq:cond:ii}?  One
approach is as follows.  Find $p$ non-negative real numbers
$\lambda_1, \ldots, \lambda_p$ for which the following inequality is
valid:
\begin{align}
  \bm c \cdot \bm h \geq & (\sum_{i=1,p} \lambda_i \bm c_i) \cdot \bm h \label{eq:cond:ii:uncond}
\end{align}
Then, observe that~\eqref{eq:cond:ii:uncond}
implies~\eqref{eq:cond:ii}.  A natural question is whether every
conditional information inequality can be derived in this way, from an
unconditional inequality: when that is the case, then we say that the
conditional inequality~\eqref{eq:cond:ii} {\em relaxes}
to~\eqref{eq:cond:ii:uncond}, or that it is {\em essentially
  unconditional}.  Otherwise we say that it is {\em essentially
  conditional}.  For example, the augmentation axiom above can be
relaxed, hence it is essentially unconditioned.

Besides offering an important proof technique, relaxation is important
in modern database applications, because often the integrity
constraints don't hold {\em exactly}, but only {\em approximatively},
especially when they are mined from a given
dataset~\cite{DBLP:journals/is/GiannellaR04,DBLP:conf/vldb/SismanisBHR06,DBLP:conf/icde/ChuIPY14,DBLP:conf/cikm/BleifussBFRW0PN16,DBLP:journals/pvldb/0001N18,DBLP:conf/sigmod/SalimiGS18,DBLP:conf/uai/PoonD11,pmlr-v48-friesen16}.
The relaxation problem allows us to transfer proofs over exact
constraints to approximate constraints.

Every implication problem for FDs and MVDs
relaxes~\cite{DBLP:journals/lmcs/KenigS22}:

\begin{thm} \label{thm:fd:mvd:relax}
  Consider the statement~\eqref{eq:fd:mvd:exact:implication} asserting
  the implication between a set of FDs/MVDs.  Consider the associated
  conditional information inequality:
  \begin{align}
    \left(\bigwedge_i h(\sigma_i) = 0\right) \Rightarrow& (h(\sigma_0) = 0)\label{eq:fd:mvd:implication}
  \end{align}
  where each expression $h(\sigma_i)$ represents either
  $I_h(\bm V;\bm W|\bm U)$ or $h(\bm V|\bm U)$.  Then the following
  are equivalent:
  \begin{itemize}
  \item The implication~\eqref{eq:fd:mvd:exact:implication} holds for FDs/MVDs.
  \item The implication~\eqref{eq:fd:mvd:implication} holds for all polymatroids.
  \item The implication~\eqref{eq:fd:mvd:implication} holds for all entropic functions.
  \item The implication~\eqref{eq:fd:mvd:implication} holds for all normal polymatroids.
  \item The inequality $n^2/4 (\sum_i h(\sigma_i)) \geq h(\sigma_0)$
    holds for all polymatroids. Furthermore, if $\sigma_0$ is an FD
    (rather than MVD), then $n^2/4$ can be replaced by 1.
  \end{itemize}
\end{thm}

Thus, the implication problem for FDs and MVDs relaxes.  A consequence
of the theorem is that, if~\eqref{eq:fd:mvd:implication} fails, then
there exists a relation with only two tuples falsifying the
implication, namely one of the relations $R_{\bm W}$ in~\eqref{eq:r:w}
associated to a step function.

Next, we examine whether all conditional inequalities relax.  To do
this, we need another (last) deep dive into the structure of entropic
functions, and introduce {\em almost-entropic functions}.

\subsection{Background:  Almost Entropic Functions}
\label{subsec:almost:entropic}

We start with a brief review of cones,
following~\cite{schrijver-book,boyd_vandenberghe_2004}.  A set
$K \subseteq \R^n$ is called a \emph{cone}, if $\bm x \in K$ and
$\theta \geq 0$ implies $\theta \bm x \in K$.  The cone is
\emph{convex} if $\bm x_1,\bm x_2 \in K$ and $\theta \in [0,1]$
implies $\theta \bm x_1+ (1-\theta)\bm x_2 \in K$. For any set
$K \subseteq \R^n$, we denote by $\bar K$ its {\em topological
  closure}, and by
$K^* \defeq \setof{\bm y}{\forall \bm x \in K, \bm y^T \cdot \bm x
  \geq 0}$, it's {\em dual}.  The dual is always a closed, convex
cone, $K \subseteq K^{**}$ and, when $K$ is a closed, convex cone,
then $K = K^{**}$.

A cone is \emph{polyhedral} if it has the form
$K = \setof{\bm x}{\bm M \cdot \bm x \geq 0}$, for some matrix
$\bm M \in \R^{m \times n}$.  Any polyhedral cone is closed and
convex, and its dual is also polyhedral.
%

The set of polymatroids $\Gamma_n$ and of entropic vectors
$\Gamma_n^*$ are subsets of $\R^{2^n}$.  The superscript $*$ in
$\Gamma_n^*$ is an unfortunate notation, since it does not represent a
dual, but this notation is already widely used.  Valid inequalities
(Def.~\ref{def:ii}) are the dual cones, $(\Gamma_n^*)^*$, and
$(\Gamma_n)^*$ respectively.  Clearly, $\Gamma_n$ is polyhedral, and
therefore $\Gamma_n = \Gamma_n^{**}$.  What about $\Gamma_n^*$?

It turns out that, when $n \geq 3$, then $\Gamma_n^*$ is neither a
cone nor convex.  This may come as a surprise, so we take the
opportunity to briefly review here the elegant proof by Zhang and
Yeung~\cite{DBLP:journals/tit/ZhangY97}.  Consider the parity function
$h$, shown in Fig.~\ref{fig:parity}, and observe that $h(X)=1$,
$h(Z|XY)=0$, and $I_h(X;Y)=0$.  If $c$ is a natural number, the vector
$c\cdot \bm h$ is also entropic, by Prop.~\ref{prop:sum}; but in
general, $c\cdot \bm h$ is entropic only if there exists a natural
number $N$ such that $c = \log N$, which implies $\Gamma_n^*$ is
neither a cone, nor convex.  To prove this, assume that
$c \cdot \bm h$ is entropic, and realized by a probability
distribution $p(X,Y,Z)$.  Choose any two values $x$, $y$ such that
$p(X=x)>0$ and $p(Y=y)>0$.  Since
$I_{c \bm h}(X;Y) = c\cdot I_h(X;Y)= 0$, it holds that $X \perp Y$,
hence $p(X=x)p(Y=y) = p(X=x,Y=y) > 0$.  Furthermore,
$c\cdot h(Z|XY) = 0$, therefore $p$ satisfies the functional
dependency $XY \rightarrow Z$, and there exists a unique value $z$
s.t. $p(X=x,Y=y,Z=z)>0$.  We have obtained
$p(X=x)p(Y=y)=p(X=x,Y=y,Z=z)$ and, by symmetry, it also holds that
$p(X=x)p(Z=z)=p(X=x,Y=y,Z=z)$.  This implies $p(Y=y)=p(Z=z)$.  Since
$y$ was arbitrary, it follows that $p(Y=y)=p(Y=y')$ for all $y,y'$ in
the support of $Y$.  Therefore, the marginal distribution of $Y$ is
uniform, and its entropy is $c\cdot h(Y) = \log N$, where $N$ is the
size of the support, proving $c = \log N$, since $h(Y)=1$.  Recall
that in Sec.~\ref{sec:entropic:polymatroid:bound} we stated that the
conditional entropy $h(-|\bm U)$ is not always an entropic vector: we
invite the reader to give such an example.

While $\Gamma_n^*$ is neither a cone nor convex,
Yeung~\cite{Yeung:2008:ITN:1457455} proved:

\begin{thm}
  The topological closure $\bar \Gamma^*_n$ of $\Gamma_n^*$ is a
  closed, convex cone. A vector $\bm h \in \bar \Gamma^*$ is called
  {\em almost-entropic}.
\end{thm}

The complete picture of all sets of polymatroids discussed in this
paper is shown in Fig.~\ref{fig:diagram}.

If an inequality is valid for $\Gamma_n^*$, then it is also valid for
$\bar \Gamma_n^*$, by continuity. However, this no longer holds for
{\em conditional} information inequalities: Kaced and
Romashchenko~\cite{DBLP:journals/tit/KacedR13} gave an example of a
conditional inequality that is valid for $\Gamma_n^*$, but not for
$\bar \Gamma_n^*$.  Since we are interested in the relaxation problem,
we will consider only validity for $\bar \Gamma_n^*$.

When $n=3$ then one can show that $\bar \Gamma_3^* = \Gamma_3$, hence
it is polyhedral.  However, Mat{\'u}{\v s}~\cite{DBLP:conf/isit/Matus07} showed
that, for $n \geq 4$, $\bar \Gamma_n^*$ is not polyhedral.  This
explains the difficulties in understanding the non-Shannon
inequalities, and also in reasoning about conditional inequalities.

\subsection{A Conditional Inequality that Does Not Relax}

Kaced and Romashchenko~\cite{DBLP:journals/tit/KacedR13} gave four
examples of essentially conditional inequalities.  This is a very
surprising result: it means that a proof of the
implication~\eqref{eq:fd:mvd:exact:implication} becomes useless if one
of the assumptions has even a tiny violation in the data.  We describe
here one of their examples, following the adaptation
in~\cite{DBLP:journals/lmcs/KenigS22}.

\begin{thm}~\cite{DBLP:journals/tit/KacedR13} \label{thm:kaced:romashchenko}
  The following conditional inequality is valid for $\bar \Gamma_n^*$,
  and is essentially conditional:
  \begin{align*}
    I(X; Y|A)=I(X; Y|B)&= I(A; B)= I(A; X|Y)=0    \Rightarrow I(X;Y)=0
  \end{align*}
\end{thm}

Notice that none of the constraints is an FD or an MVD, so this does
not contradict Theorem~\ref{thm:fd:mvd:relax}.  We prove the theorem,
and start by proving the conditional inequality. For that we use the
following non-Shannon inequality, by
Mat{\'{u}}s~\cite{DBLP:conf/isit/Matus07}.  For every $k\geq 1$,
\begin{align}
  I(X;Y) \leq & \frac{k+3}{2}I(X;Y|A)+I(X;Y|B)+ I(A;B)\nonumber\\
     + &  \frac{k+1}{2}I(A;X|Y)+ \frac{1}{k}I(A;Y|X)\label{eq:matus}
\end{align}
When $k=1$, this is Zhang and Yeung's
inequality~\eqref{eq:zhang:yeung}.  Mat{\'{u}}s
proved~\eqref{eq:matus} by induction on $k$, by applying the Copy
Lemma~\ref{lmm:copy} at each induction step; we omit the proof.

Since~\eqref{eq:matus} holds for $\Gamma_n^*$, it also holds for
$\bar \Gamma_n^*$.  To check the conditional inequality in
Thm.~\eqref{thm:kaced:romashchenko}, let $\bm h \in \bar \Gamma_n^*$
such that $I(X; Y|A)=I(X; Y|B)= I(A; B)= I(A; X|Y)=0$: then
inequality~\eqref{eq:matus} becomes $I(X;Y) \leq \frac{1}{k}I(A;Y|X)$
and, since $k$ is arbitrary, it follows that $I(X;Y) = 0$.

Finally, we show that the conditional inequality does not relax, by
describing, for each $\lambda > 0$, an entropic vector $\bm h$ s.t.
\begin{align}
  I(X;Y) \geq & \lambda(I(X; Y|A)+I(X; Y|B)+ I(A; B)+ I(A; X|Y)) \label{eq:not:relaxed}
\end{align}
Let  $\bm h$ be  the entropy of the following distribution:

\null\hfill
\begin{tabular}{|l|l|l|l|l} \cline{1-4}
  $A$ & $B$ & $X$ & $Y$ & $p$ \\ \cline{1-4}
  0 & 0 & 0 & 0 & $1/2 -\varepsilon$ \\
  1 & 0 & 0 & 1 & $1/2 -\varepsilon$ \\
  0 & 1 & 1 & 0 & $\varepsilon$ \\
  1 & 1 & 0 & 0 & $\varepsilon$ \\ \cline{1-4}
\end{tabular}
\hfill\null

\noindent If $\varepsilon$ is small enough then one can
check:\footnote{Complete calculations are included
  in~\cite{DBLP:journals/lmcs/KenigS22}; note that here we have
  swapped the roles of $A$ and $B$, in order to better draw the
  connection to Zhang and Yeung's inequality~\eqref{eq:zhang:yeung}.}
\begin{align*}
  I(X;Y)=&\varepsilon + O(\varepsilon^2), \ I(A;X|Y)=O(\varepsilon^2), \\
  I(X;Y|A)=&I(X;Y|B)=I(A;B)=0
\end{align*}
which proves~\eqref{eq:not:relaxed} for $\varepsilon$ is small enough.

\subsection{Conditional Inequalities Relax with Error Terms}

However, in another twist, it turns out that every conditional
inequality relaxes, if we admit a small error term.  The following was
proven in~\cite{DBLP:journals/lmcs/KenigS22}:

\begin{thm} \label{th:relaxation} Suppose that the following holds:
  \begin{align}
\forall \bm h \in \bar \Gamma_n^*,\ \ \left(\bigwedge_{i=1,p} \bm c_i \cdot \bm h \leq 0\right) \Rightarrow & \bm c_0\cdot \bm h\leq 0 \label{eq:cond:ii:non-relaxed}
  \end{align}
  Then, for every $\varepsilon > 0$ there exists
  $\lambda_1, \ldots, \lambda_p \geq 0$ such that:
  \begin{align}
\forall \bm h \in \bar \Gamma_n^*,\ \ \bm c_0\cdot \bm h\leq & \left(\sum_{i=1,p} \lambda_i \bm c_i \cdot \bm h\right)+\varepsilon h(\bm X) \label{eq:cond:ii:relaxed}
  \end{align}
  where $\bm X$ is the set of all $n$ variables. 
\end{thm}

Even with the error term, condition~\eqref{eq:cond:ii:relaxed} still
implies~\eqref{eq:cond:ii:non-relaxed}, because, if
$\bm c_i \cdot \bm h \leq 0$ for all $i$,
then~\eqref{eq:cond:ii:relaxed} implies
$\bm c_0 \cdot \bm h \leq \varepsilon h(\bm X)$ and, since
$\varepsilon$ is arbitrary, we obtain $\bm c_0 \cdot \bm h \leq 0$.
In fact, Mat{\'{u}}s' inequality~\eqref{eq:matus} can be seen as a
relaxation, with an error term, of the conditional inequality in
Theorem~\ref{thm:kaced:romashchenko}.  Theorem~\ref{th:relaxation}
shows that this was not accidental: {\em every} conditional inequality
follows from an unconditional with an error term that tends to 0.  In
the next section we will show that Theorem~\ref{th:relaxation} has a
surprising application, to the proof of
Theorem~\ref{th:primal:dual:bound}.  Before that, we prove
Theorem~\ref{th:relaxation} by showing:

\begin{lmm}~\cite{DBLP:journals/lmcs/KenigS22} \label{lemma:relaxation}
  Let $K \subseteq \R^n$ be a closed, convex cone, and
  $\bm c_i \in \R^n$, $i=0,p$ be vectors such that the following
  holds:
  \begin{align}
    \forall \bm x \in K:\ \ \left(\bigwedge_{i=1,p} \bm c_i\cdot \bm x \leq 0\right) \Rightarrow  & \bm c_0 \cdot \bm x \leq 0 \label{eq:k:ei}
  \end{align}
  Then, for every $\varepsilon > 0$, there exists
  $\lambda_1,\ldots,\lambda_p \geq 0$ such that:
  \begin{align}
    \forall \bm x \in K:\ \ \bm c_0 \cdot \bm x \leq & \left(\sum_{i=1,p}\lambda_i \bm c_i\right) \cdot \bm x+ \varepsilon ||\bm x||_\infty \label{eq:k:ai}
  \end{align}
\end{lmm}

The lemma implies the theorem, because $\bar \Gamma_n^*$ is a closed,
convex cone, and $||\bm h||_\infty = h(\bm X)$.

\begin{proof} (of Lemma~\ref{lemma:relaxation}) Let
  $L\defeq \setof{-\bm c_i}{i=1,p}$.  Then condition~\eqref{eq:k:ei}
  says that $\bm x \in K\cap L^*$ implies
  $\bm c_0 \cdot \bm x \leq 0$, or, equivalently,
  $-\bm c_0 \in (K \cap L^*)^*$.  The following holds for any closed,
  convex cones $K_1, K_2$ (see\cite[Sec.5.3]{DBLP:journals/lmcs/KenigS22}):
  \begin{align*}
    (K_1 \cap K_2)^* = & \overline{\conehull(K_1^* \cup K_2^*)}
  \end{align*}
  where $\conehull(A)$ is the conic hull of a set $A$, i.e. the set of
  positive, linear combinations of vectors in $A$.  Also, if $L$ is
  finite, then $L^{**} = \conehull(L)$.  Therefore, $-\bm c_0$ belongs
  to the following set:
  \begin{align*}
(K \cap L^*)^*= & \overline{\conehull(K^* \cup L^{**})} \\
   = &  \overline{\conehull(K^* \cup \conehull(L))}\\
   = &  \overline{\conehull(K^* \cup L)}
  \end{align*}
  For any $\varepsilon > 0$, there exists $\bm e \in \R^n$ such that
  $||\bm e||_1 < \varepsilon$ and:
  \begin{align*}
    -\bm c_0 + \bm e \in & \conehull(K^* \cap L)
  \end{align*}
  By the definition of the conic hull, and the fact that $K^*$ is a
  convex cone, we obtain that there exists $\bm d \in K^*$ and
  $\lambda_i \geq 0$, for $i=1,p$ such that:
  \begin{align*}
    -\bm c_0 + \bm e = & \bm d - \sum_{i=1,p} \lambda_i \bm c_i
  \end{align*}
  We prove~\eqref{eq:k:ai}.  Let $\bm x \in K$, and observe that
  $\bm d \cdot \bm x \geq 0$, then:
  \begin{align*}
    \left(\sum_{i=1,p} \lambda_i \bm c_i\right)\cdot \bm x - \bm c_0
    \cdot \bm x + \bm e \cdot \bm x = & \bm d \cdot \bm x \geq 0
  \end{align*}
  and~\eqref{eq:k:ai} follows from
  $\bm e \cdot \bm x \leq ||\bm e||_1 \cdot ||\bm x||_\infty \leq
  \varepsilon ||\bm x||_\infty$.
\end{proof}

We end this section with a brief discussion of why
Theorem~\ref{th:relaxation} only holds for the set of almost entropic
functions $\bar \Gamma_n^*$, and fails for $\Gamma_n^*$.  Kaced and
Romashchenko~\cite{DBLP:journals/tit/KacedR13} gave an example of a
conditional inequality that is valid for $\Gamma_n^*$, but not for
$\bar \Gamma_n^*$.  If Theorem~\ref{th:relaxation} were to hold for
that inequality, then the relaxed
inequality~\eqref{eq:cond:ii:relaxed} holds for $\bar \Gamma_n^*$, and
then we could prove that the conditional inequality also holds for
$\bar \Gamma_n^*$.

\subsection{Proof of Theorem~\ref{th:primal:dual:bound}}

\label{sec:proof:of:duality}

We have now the machinery needed to prove
Theorem~\ref{th:primal:dual:bound}, which states that
$\text{Log-L-Bound}_{\Gamma_n^*}$ and
$\text{Log-U-Bound}_{\Gamma_n^*}$ are asymptotically equal.  What
makes the proof a little difficult is the ill-behaved nature of the
entropic functions $\Gamma_n^*$.  To cope with that, use some help
from the almost entropic functions $\bar \Gamma_n^*$, and prove that,
here, the two bounds are equal (not just asymptotically).  We state
and prove the theorem for an arbitrary closed, convex cone $K$:

\begin{thm}
  \label{th:primal:dual:bound:k} Fix $Q, \Sigma, \bm b$, and let $K$
  be a closed, convex cone s.t. $N_n \subseteq K \subseteq
  \Gamma_n$. (Recall that $N_n$ is the set of normal polymatroids,
  Def.~\ref{def:normal:polymatroid}.)  Then:
  \begin{align*}
  \text{Log-L-Bound}_K(Q,\Sigma,\bm b)= & \text{Log-U-Bound}_K(Q,\Sigma,\bm b)
  \end{align*}
\end{thm}

In particular, the lower and upper bounds are equal for
$K = \bar \Gamma_n^*$, while they are not necessarily equal for
$\Gamma_n^*$.  It shows that the set $\bar \Gamma_n^*$ is better
behaved that $\Gamma_n^*$, and that explains why it was used in prior
work~\cite{DBLP:conf/icdt/GogaczT17,DBLP:conf/pods/Khamis0S17} to
study lower bounds.

Before we prove the theorem, we show how to use it to prove
Theorem~\ref{th:primal:dual:bound}.  Its proof follows from
Theorem~\ref{th:primal:dual:bound:k} and two
identities,~\eqref{eq:u:gamman:u:gammans}
and~\eqref{eq:the:same:u:bound}, which we prove below.  The first is
very simple:
\begin{align}
  \text{Log-U-Bound}_{\bar \Gamma_n^*}(Q,\Sigma,\bm b)=\text{Log-U-Bound}_{\Gamma_n^*}(Q,\Sigma,\bm b)\label{eq:u:gamman:u:gammans}
\end{align}
and follows directly from the fact that $\bar \Gamma_n^*$ and
$\Gamma_n^*$ have the same dual cone,
$(\bar \Gamma_n^*)^* = (\Gamma_n^*)^*$; in other words, they define
the same set of valid inequalities $\bm c\cdot \bm h \geq 0$.

The second equality requires a proof, and we state it as a lemma:
\begin{lmm} \label{lmm:the:same:u:bound} The following holds:
  \begin{align}
    \sup_k \frac{\text{Log-L-Bound}_{\Gamma_n^*}(Q,\Sigma,k\bm b)}{\text{Log-L-Bound}_{\bar \Gamma_n^*}(Q,\Sigma,k\bm b)}=&1 \label{eq:the:same:u:bound}
  \end{align}
\end{lmm}

\begin{proof} The proof is similar to that of
  Theorem~\ref{th:tight:non-tight} item (1) in
  Sec.~\ref{sec:th:tight:non-tight}.  The LHS is obviously $\leq 1$.
  To prove that it is $\geq 1$, denote by\footnote{If
    $\text{Log-L-Bound}_{\bar \Gamma_n^*}(Q,\Sigma,\bm b)=\infty$ then
    we choose $L$ an arbitrarily large number, and make minor
    adjustments to the proof; we omit the details.}
  $L \defeq \text{Log-L-Bound}_{\bar \Gamma_n^*}(Q,\Sigma,\bm b)$, and
  observe that $\text{Log-L-Bound}_{\bar \Gamma_n^*}$ is linear in
  $\bm b$,
  $\text{Log-L-Bound}_{\bar \Gamma_n^*}(Q,\Sigma,k\bm b) = k L$,
  because $\bar \Gamma_n^*$ is a convex cone.  It suffices to show
  that, for all $\varepsilon > 0$, there exists $\bm h\in \Gamma_n^*$
  such that $\bm h \models (\Sigma, k\bm b)$ and
  $h(\bm X) \geq (1-\varepsilon)^2 kL$.

  Start with some  $\bm h \in \bar \Gamma_n^*$ satisfying:
  \begin{align*}
    \forall \sigma \in \Sigma:\ h(\sigma) \leq & b_\sigma, & h(\bm X)= & L
  \end{align*}
  which exists by the definition of
  $\text{Log-L-Bound}_{\bar \Gamma_n^*}$
  (Def.~\ref{def:ep:bound:dual}) and the fact that the set
  $\setof{\bm h \in \R^{2^{[n]}}}{||\bm h||_\infty \leq L}$ is
  compact.  Chan and Yeung's Theorem~\ref{th:chan:groups} proves that
  $\Upsilon_n$ is dense in $\Gamma_n^*$, and therefore it is also
  dense in $\bar \Gamma_n^*$.  Assume w.l.o.g. that $h(\bm X) > 0$,
  then
  $g \defeq \min_{\bm U, \bm V: h(\bm V|\bm U>0)} h(\bm V|\bm U) > 0$
  (the smallest non-zero value of any expression $h(\bm V|\bm U)$).
  Let $\delta \defeq \varepsilon g /4$.  Since $\Upsilon_n$ is dense,
  there exists $r \in \N$ and $\bm h^{(r)} \in \Upsilon_n$ such that
  $||\frac{1}{r}\bm h^{(r)}-\bm h||_\infty \leq \delta$.  We have
  $h(\bm X) >0$ hence $h(\bm X) \geq g$, therefore:
  \begin{align*}
    \frac{1}{r}h^{(r)}(\bm X) \geq & h(\bm X)- \delta = h(\bm X)-\varepsilon g/4\\
    \geq & h(\bm X) - (\varepsilon/4)h(\bm X) \geq (1-\varepsilon)h(\bm X)\\
    = & (1-\varepsilon) L
  \end{align*}
  We claim that
  $\frac{1}{r}h^{(r)}(\bm V|\bm U)\leq (1+\varepsilon/2) h(\bm V|\bm
  U)$, for all $\bm U, \bm V$. If $h(\bm V|\bm U) = 0$, then $\bm h$
  satisfies the FD $\bm U\rightarrow \bm V$, and therefore
  $\bm h^{(r)}$ also satisfies this FD (see the note after
  Theorem~\ref{th:chan:groups}), implying $h^{(r)}(\bm V|\bm U) = 0$.
  Otherwise, $h(\bm V|\bm U) \geq g$ and,
  \begin{align*}
    \frac{1}{r}h^{(r)}(\bm V|\bm U) \leq & h(\bm V|\bm U) +2\delta = h(\bm V|\bm U) + \varepsilon g /2 \\
    \leq & (1+\varepsilon/2) h(\bm V|\bm U)
  \end{align*}
  Therefore, we have:
  \begin{align*}
    \forall \sigma \in \Sigma:\ h^{(r)}(\sigma)\leq & (1+\varepsilon/2)r b_\sigma& h^{(r)}(\bm X) \geq & (1-\varepsilon) r L
  \end{align*}
  Finally, by the Slack lemma,
  $\exists k\in \N, \bm h' \in \Gamma_n^*$ such that:
  \begin{align*}
    \forall \sigma \in \Sigma:\ h'(\sigma)\leq & (1-\varepsilon/2)k h^{(r)}(\sigma) \leq (1-(\varepsilon/2)^2) k r b_\sigma\leq k r b_\sigma\\
    h'(\bm X) \geq & (1-\varepsilon)k h^{(r)}(\bm X) \geq  (1-\varepsilon)^2k r L
  \end{align*}
  This completes the proof.
\end{proof}


In the remainder of this section we prove Theorem~\ref{th:primal:dual:bound:k}.

\begin{proof} (of Theorem~\ref{th:primal:dual:bound:k}) We will use
  the following definition from~\cite[Example
  5.12]{boyd_vandenberghe_2004}:
  
  \begin{defn} \label{def:primal:dual:cone:program}
    Let $K$ be a proper cone (meaning: closed, convex, with a
    non-empty interior, and pointed i.e. $\bm x, - \bm x \in K$
    implies $\bm x =0$).  A primal/dual cone program in standard
    form\footnote{We changed to the original
      formulation~\cite{boyd_vandenberghe_2004} by replacing $\bm c$
      with $-\bm c$, replacing $\bm y$ with $-\bm y$.} is the
    following:
  \begin{align*}
    &
      \begin{array}{ll|ll}
        Primal && Dual& \\ \hline
        \mbox{Maximize} & \bm c^T \cdot \bm x & \mbox{Minimize} & \bm y^T\cdot b\\
        \mbox{where} & \bm A \cdot \bm x = \bm b &  \mbox{where} & (\bm y^T\cdot \bm A - \bm c^T)^T \in K^*\\
             & \bm x \in K & & 
      \end{array}
  \end{align*}
\end{defn}

Denote by $P^*, D^*$ the optimal value of the primal and dual
respectively.  {\em Weak duality} states that $P^* \leq D^*$, and is
easy prove.  When {\em Slater's condition} holds, which says that
there exists $\bm x$ in the interior of $K$ such that
$\bm A \bm x = \bm b$, then {\em strong duality} holds too:
$P^* = D^*$.

$\text{Log-L-Bound}_K(Q,\Sigma,\bm b)$ and
$\text{Log-U-Bound}_K(Q,\Sigma,\bm b)$ can be expressed as a cone
program, by letting $\bm A$ and $\bm c$ be the matrix and vector
defined in the proof of Theorem~\ref{th:primal:dual:bound:polyhedral}
(thus, $\bm A \cdot \bm h = (h(\sigma))_{\sigma \in \Sigma}$ and
$c_{\bm X} = 1$, $c_{\bm U}=0$ for $\bm U \neq \bm X$):
%
%
\begin{align}
  & \begin{array}{ll|ll}
      \multicolumn{2}{l|}{\text{Log-L-Bound}_K}& \multicolumn{2}{|l}{\text{Log-U-Bound}_K} \\ \hline
      \multicolumn{2}{l|}{\mbox{Maximize } \bm c^T \cdot h} & \multicolumn{2}{|l}{\mbox{Minimize }  \bm w^T\cdot\bm b}\\
      \mbox{where} & \bm A \cdot \bm h + \bm \beta = \bm b &  \mbox{where} & (\bm w^T\cdot \bm A - \bm c^T)^T \in K^*\\
                                               & (\bm h,\bm \beta) \in K\times \R_+^s & & \bm w\geq 0 
    \end{array} \label{eq:primal:dual:cone:h}
\end{align}
Here $\bm \beta$ are slack variables that convert an inequality
$h(\sigma) \leq b_{\sigma}$ into an equality
$h(\sigma)+ \beta_{\sigma} = b_{\sigma}$.  We leave it to the reader
to check that these two programs are indeed primal/dual as in
Def.~\ref{def:primal:dual:cone:program}.

However, in general, Slater's condition need not hold
for~\eqref{eq:primal:dual:cone:h}.  For example, $K=\bar \Gamma_n^*$
lies in the hyperplane $h(\emptyset)=0$, and thus has an empty
interior. This could be addressed by removing the $\emptyset$
dimension, but we have a bigger problem. Some of the constraints may
be tight: when $b_\sigma = 0$, then $h(\sigma) = 0$, meaning that no
feasible solution $\bm h$ exists in the interior of $K$.
Instead, we will define a different cone, $K_0$, by using polymatroids
on a lattice, as in Sec.~\ref{sec:background:non-shannon}.

Partition $\Sigma$ into $\Sigma_0 = \setof{\sigma}{b_\sigma =0}$ and
$\Sigma_1 = \setof{\sigma}{b_\sigma > 0}$, and denote by $\bm b_1$ the
restriction of $\bm b$ to $\Sigma_1$.  In other words, $\Sigma_0$
defines a set of functional dependencies, while $(\Sigma_1,\bm b_1)$
defines non-tight statistics.  Let $(L_{\Sigma_0},\subseteq)$ be the
lattice of the closed sets of $\Sigma_0$ (defined in
Sec.~\ref{sec:problem}); we will drop the subscript and write simply
$(L,\subseteq)$ to reduce clutter.  Define $F \subseteq \R^{2^{[n]}}$
the following cone\footnote{In fact $F$ is even a vector space.}:
\begin{align*}
  F \defeq & \setof{\bm h \in \R^{2^{[n]}}}{\forall \sigma \in \Sigma_0: h(\sigma)=0}
\end{align*}
Let $L_0 \defeq L-\set{\hat 0}$, $N \defeq |L_0|$.  Our cone
$K_0 \subseteq \Rp^N$ is:
\begin{align}
  K_0 \defeq & \Pi_{L_0}(K \cap F) \label{eq:the:cone:k}
\end{align}
The function $\Pi_{L_0}$ projects a $2^{[n]}$-dimensional vector
$(h_{\bm U})_{\bm U \in 2^{\bm X}}$ to the $N$-dimensional vector
$(h_{\bm U})_{\bm U \in L_0}$.  Thus, $K_0$ not only removes the
$\emptyset$ dimension, but also removes all dimensions subject to a
tight constraint.  We prove the following:
\begin{enumerate}[(1)]
\item \label{item:k:1}  $K_0$ is proper
\item \label{item:k:2} $\text{Log-L-Bound}_K = \text{Log-L-Bound}_{K_0}$ and\\
  $\text{Log-U-Bound}_K = \text{Log-U-Bound}_{K_0}$
\item \label{item:k:3} $\text{Log-L-Bound}_{K_0}=\text{Log-U-Bound}_{K_0}$
\end{enumerate}
Theorem~\ref{th:primal:dual:bound:k} follows from these three claims.

We start with item~\ref{item:k:1}, and observe that $K_0$ is a closed,
convex cone, because $K \cap F$ is a closed, convex cone, and
$\Pi_{L_0}$ is a linear isomorphism $K \cap F \rightarrow K_0$:
indeed, $\Pi_{L_0}$ is surjective by the definition of $K_0$, and it
is injective because, if $\bm h, \bm h' \in K\cap F$, then
$\Pi_{L_0}(\bm h)=\Pi_{L_0}(\bm h')$ implies that, for every set
$\bm U$, $h(\bm U) = h(\bm U^+) = h'(\bm U^+) = h'(\bm U)$.  It is
immediate to check that $K_0$ is pointed, and we will show below that
$K_0$ has a non-empty interior: this implies that it is proper.

Next, we prove item~\ref{item:k:2}, and for that we will write $\bm h$
for a vector in $K \cap F$ and write $\bm h^{(0)}$ for a vector in
$K_0$.  For any statistics $\sigma = (\bm V|\bm U) \in \Sigma_1$,
denote by $\sigma^+ = ((\bm U\bm V)^+|\bm U^+)$.  We say that a vector
$\bm h^{(0)}$ {\em satisfies} the statistics $(\Sigma_1, \bm b_1)$, in
notation $\bm h^{(0)} \models (\Sigma_1,\bm b_1)$, if
$\bm h^{(0)}(\sigma^+) \leq \bm b_{\sigma}$ for all
$\sigma \in \Sigma_1$.  By definition, a vector $\bm h^{(0)} \in K_0$
satisfies the FDs $\Sigma_0$.

\begin{lmm}  The following holds:
  \begin{align*}
    \text{Log-L-Bound}_K(Q,\Sigma,\bm b)=\text{Log-L-Bound}_{K_0}(Q,\Sigma_1,\bm b_1)
  \end{align*}
\end{lmm}
\noindent Equivalently, the lemma states:
\begin{align*}
  \sup_{\bm h} & \bigsetof{h(\bm X)}{\bm h \in K, \bm h \models  (\Sigma, \bm b)}\\
  = & \sup_{\bm h^{(0)}} \bigsetof{h^{(0)}(\hat 1)}{\bm h^{(0)} \in K_0, \bm h^{(0)} \models  (\Sigma_1, \bm b_1)}
\end{align*}
and the proof is immediate, because the projection of a vector
$\bm h \in K$ satisfying $(\Sigma, \bm b)$ is a vector
$\bm h^{(0)} \in K_0$ satisfying $(\Sigma_1, \bm b_1)$, and,
conversely, every such $\bm h^{(0)}$ is the projection of a vector
$\bm h$.

Recall that we have assumed $N_n \subseteq K \subseteq \Gamma_n$:

\begin{lmm}  The following holds:
  \begin{align*}
    \text{Log-U-Bound}_K(Q,\Sigma,\bm b)=\text{Log-U-Bound}_{K_0}(Q,\Sigma_1,\bm b_1)
  \end{align*}
\end{lmm}

\begin{proof} We need to prove:
  \begin{align}
    \inf_{\bm w} & \bigsetof{\sum_{\sigma \in \Sigma} w_\sigma b_\sigma}{K \models \sum_{\sigma \in \Sigma}w_\sigma h(\sigma)\geq h(\bm X)}\label{eq:relaxation:application}\\
 = &  \inf_{\bm w^{(0)}}  \bigsetof{\sum_{\sigma \in \Sigma_1} w_{\sigma^+}^{(0)} b_{\sigma}}{K_0\models \sum_{\sigma \in \Sigma_1}w^{(0)}_{\sigma^+} h^{(0)}(\sigma^+)\geq h^{(0)}(\bm X)}\nonumber
  \end{align}
  A vector $\bm w$ on the LHS defines an unconstrained inequality,
  while a vector $\bm w^{(0)}$ on the RHS defines a constrained
  inequality, because $\bm w^{(0)}$ satisfies
  \begin{align*}
    \forall \bm h^{(0)} \in K_0, \sum_{\sigma \in \Sigma_1}w^{(0)}_{\sigma^+} h^{(0)}(\sigma^+)\geq h^{(0)}(\bm X)
  \end{align*}
  iff it satisfies
  \begin{align}
    \forall \bm h\in K, \bigwedge_{\sigma\in\Sigma_0} h(\sigma)=0 \Rightarrow \sum_{\sigma \in \Sigma_1}w^{(0)}_{\sigma^+} h(\sigma^+)\geq h(\bm X) \label{eq:w:k:k0:constrained}
  \end{align}
  which is a constrained inequality.

  We start by showing that LHS$\geq$RHS in
  Eq.~\eqref{eq:relaxation:application}.  For that it suffices observe
  that, if $\bm w = (\bm w_\sigma)_{\sigma \in \Sigma}$ defines a
  valid inequality
  $\sum_{\sigma \in \Sigma}w_\sigma h(\sigma)\geq h(\bm X)$, then its
  projection $\bm w^{(0)} = (w_{\sigma^+})_{\sigma \in \Sigma_1}$
  defines a valid constrained inequality
  $\sum_{\sigma \in \Sigma_1}w^{(0)}_{\sigma^+} h^{(0)}(\sigma^+)\geq
  h^{(0)}(\bm X)$, because all the missing terms $h(\sigma)$ for
  $\sigma \in \Sigma_0$ are $=0$.  Therefore
  $\inf_{\bm w}(\cdots) \geq \inf_{\bm w^{(0)}}(\cdots)$

  We prove now that LHS$\leq$RHS.
%
  Let $\bm w^{(0)}$ be a vector defining a valid constrained
  inequality~\eqref{eq:w:k:k0:constrained}.  The objective value of
  the RHS of~\eqref{eq:relaxation:application} is
  $\sum_{\sigma \in \Sigma_1} w^{(0)}_{\sigma^+}b_\sigma$.  By the relaxation
  theorem~\ref{th:relaxation}, for every $\varepsilon > 0$ there
  exists $\lambda_\sigma \geq 0$, for $\sigma \in \Sigma_0$, such that
  the following is a valid, unconstrained inequality:
  \begin{align*}
    K \models \sum_{\sigma\in \Sigma_0}\lambda_\sigma h(\sigma)
    + \sum_{\sigma \in \Sigma_1} w^{(0)}_{\sigma^+} h(\sigma^+) + \varepsilon h(\bm X) \geq & h(\bm X)
  \end{align*}
  This inequality is not yet of the form on the LHS
  of~\eqref{eq:relaxation:application}, because we have terms
  $h(\sigma^+)$ instead of $h(\sigma)$.  For each such term,
  $h(\sigma^+) = h((\bm U\bm V)^+|\bm U^+)$, where
  $\sigma = (\bm V|\bm U)$, we use the following Shannon inequality:
  \begin{align*}
    h(\sigma^+) = & h((\bm U\bm V)^+|\bm U^+) \leq h((\bm U\bm V)^+|\bm U) \\
    = & h(\bm U\bm V|\bm U) + h((\bm U\bm V)^+|\bm U\bm V) \\
    = & h(\sigma) + h((\bm U\bm V)^+|\bm U\bm V) \leq h(\sigma)+\sum_{\sigma \in
    \Sigma_0}h(\sigma)
  \end{align*}
  The last inequality,
  $h((\bm U\bm V)^+|\bm U\bm V) \leq \sum_{\sigma \in
    \Sigma_0}h(\sigma)$, can be checked by induction on the number of
  steps needed to compute the closure $(\bm U\bm V)^+$ using the FDs
  in $\Sigma_0$. (It also follows from
  Theorem~\ref{thm:fd:mvd:relax}.)  This implies that there exists
  coefficients $\lambda'_{\sigma}\geq 0$ such that the following is a
  valid, unconstrained inequality:
  \begin{align*}
    K \models \sum_{\sigma\in \Sigma_0}\lambda'_\sigma h(\sigma)
+ \sum_{\sigma \in \Sigma_1} w^{(0)}_{\sigma^+} h(\sigma) + \varepsilon h(\bm X) \geq & h(\bm X)
  \end{align*}
  or, equivalently,
  \begin{align*}
    K \models \sum_{\sigma\in \Sigma_0}\frac{\lambda'_\sigma}{1-\varepsilon} h(\sigma)
+ \sum_{\sigma \in \Sigma_1} \frac{w^{(0)}_{\sigma^+}}{1-\varepsilon} h(\sigma) \geq & h(\bm X)
  \end{align*}
  This is a valid inequality for the LHS
  of~\eqref{eq:relaxation:application}, and its objective value is
  $(\sum_{\sigma \in \Sigma_1}
  w^{(0)}_{\sigma^+}b_\sigma)/(1-\varepsilon)$, because $b_\sigma=0$
  for all $\sigma\in \Sigma_0$.  Since $\varepsilon$ can be chosen
  arbitrarily small, it follows that LHS$\leq$RHS
  in~\eqref{eq:relaxation:application}.
%
\end{proof}

This completes the proof of item~\ref{item:k:2}.  It remains to prove
item~\ref{item:k:3}.  For that we represent both
$\text{Log-L-Bound}_{K_0}(Q,\Sigma_1,\bm b_1)$ and
$\text{Log-U-Bound}_{K_0}(Q,\Sigma_1,\bm b_1)$ as the solutions to the
primal/dual cone program~\eqref{eq:primal:dual:cone:h} over the cone
$K_0$.  Notice that the vector $\bm b$ is restricted to $\bm b_1$ and,
therefore, $b_\sigma > 0$ for all $\sigma \in \Sigma_1$: we no longer
have tight constraints.  Similarly, the matrix $\bm A$ will be
restricted to a matrix $\bm A_1$ whose rows correspond to the closed
sets $\bm U \subseteq \bm X$.  It remains to check Slater's condition,
and, in particular, prove that $K_0$ has a non-empty interior.  For
that purpose we extended the definition of step functions from
Sec.~\ref{subsec:modular:normal} to our lattice $L$.  For each closed
set $\bm W \in L$, s.t.  $\bm W \neq \hat 1$, we define the step
function at $\bm W$ as follows.
\begin{align*}
  \forall \bm U \in L: \ \
  h_{\bm W}^{(0)}(\bm U) \defeq & \begin{cases}
    0 & \mbox{if $\bm U \subseteq \bm W$}\\
    1 & \mbox{otherwise}
    \end{cases}
\end{align*}
Let $\bm h_{\bm W} \in K$ be the standard step function in
Eq.~\eqref{eq:step:function} (we assumed $N_n \subseteq K$).
$\bm h_{\bm W}$ satisfies all FDs $\Sigma_0$, because the only FDs
that it does not satisfy are of the form $\bm U \rightarrow \bm V$
where $\bm U \subseteq \bm W$, $\bm V \not\subseteq \bm W$, and none
of the FDs in $\Sigma_0$ have this form because $\bm W^+ = \bm W$.  It
follows $\bm h_{\bm W} \in K \cap F$, and this proves
$\bm h^{(0)}_{\bm W} = \Pi_{L_0}(\bm h_{\bm W}) \in K_0$.  There are
$N$ step functions $\bm h^{(0)}_{\bm W} \in K_0$, and it is
straightforward to check that they are independent vectors in $\R^N$.
Let $\varepsilon > 0$ be small enough such that
$2 \varepsilon N < \min_{\sigma \in \Sigma_1} b_\sigma$.  Define
$\bm h \defeq \sum_{\bm W \in L_0} \varepsilon \bm h_{\bm W}$, and
$\bm h^{(0)} \defeq \Pi_{L_0}(\bm h)$.  Since $K \cap F$ is a convex
cone, $\bm h \in K\cap F$ and therefore $\bm h^{(0)} \in K_0$.  We
claim that there exists slack variables $\bm \beta$ such that
$(\bm h^{(0)},\beta)$ is a feasible solution to the cone program
in~\eqref{eq:primal:dual:cone:h} and, furthermore,
$(\bm h^{(0)}, \bm \beta)$ belongs to the interior of
$K_0 \times \Rp^s$.  Indeed, for all $\sigma\in \Sigma_1$,
$h^{(0)}(\sigma) \leq \varepsilon N < b_\sigma$, hence, if we define
$\beta_\sigma \defeq b_\sigma - h^{(0)}(\sigma)$, the pair
$(\bm h^{(0)}, \bm \beta) \in K_0 \times \Rp^s$ is a feasible solution
to the primal~\eqref{eq:primal:dual:cone:h}.  Next, we prove that
$(\bm h^{(0)}, \bm \beta)$ is in the interior of $K_0 \times \Rp^s$.
Since $\beta_\sigma > 0$ for all $\sigma \in \Sigma_1$, it follows
that $\bm \beta$ is in the interior of $\Rp^s$.  Set
$\bm h' \defeq \sum_{\bm W \in L_0} (\varepsilon + \delta_{\bm W}) \bm
h_{\bm W}$, where $\delta_{\bm W} \in (-\varepsilon,\varepsilon)$ are
$N$ arbitrary numbers, we have $\bm h' \in K \cap F$, hence
$\Pi_{L_0}(\bm h') \in K_0$.  This proves that $\bm h^{(0)}$ is in the
interior of $K_0$, thus, verifying Slater's condition.  It follows
that
$\text{Log-L-Bound}_{K_0}(Q,\Sigma_1,\bm b_1) =
\text{Log-U-Bound}_{K_0}(Q, \Sigma_1, \bm b_1)$, completing the proof
of item~\ref{item:k:3}.
\end{proof}

\section{Conclusions}

Data is ultimately information, and therefore the connection between
databases and information theory is no surprise.  We have discussed
applications of information inequalities to several database theory
problems: query upper bounds, query evaluation, query domination, and
reasoning about approximate constraints.  There are major open
problems in information theory, for example the decidability of
entropic information inequalities, the complexity of deciding Shannon
inequalities, a characterization of the cone $\bar \Gamma_n^*$, and
each such open problem has a corresponding open problem in database
theory.  In some cases the converse holds too, for example the query
domination problem is computationally equivalent to checking validity
of max-information inequalities, hence any proof of (un)-decidability
of one problem carries over to the other.

A broader question is whether information theory can find wider
applications in finite model theory.  For example, functional
dependencies and multivalued dependencies can be specified either
using first order logic sentences, or using entropic terms.  Are there
other properties in finite model theory that can be captured using
information theory?  Such a connection would enable logical
implications to be relaxed to approximate reasoning, with lots of
potential in modern, data-driven applications that rely heavily on
statistical reasoning.

\bibliographystyle{alpha}
\bibliography{bib}

\appendix

\section{Proof of the claim in Example~\ref{ex:pods2016}}

\label{app:ex:pods2016}
We prove that any valid $\Sigma$-information inequality is a positive
linear combination of the four inequalities shown in
Example~\ref{ex:pods2016}.  Such an inequality has the following form:
\begin{align}
  w_Rh(XY) + w_Sh(YZ)+w_Th(ZU) & \nonumber \\
  +w_Bh(U|XZ) + w_Ah(X|YU) \geq & h(XYZU) \label{eq:ex:pods2016:ineq}
\end{align}
where $w_R, \ldots, w_A \geq 0$ are non-negative real numbers.  Since
the inequality holds for all polymatroids, it also holds for every
step function $\bm h^{\bm V}$ (see Eq.~\eqref{eq:step:function:alt}),
for all $\bm V \subseteq \set{X,Y,Z,U}$.  There are $2^4-1 = 15$ step
functions, but we only use 5 of them:
\begin{align*}
  \bm h^{X}: && w_R + w_B \geq & 1\\
  \bm h^{Y}: && w_R + w_S \geq & 1\\
  \bm h^{Z}: && w_S + w_T \geq & 1\\
  \bm h^{U}: && w_T + w_A \geq & 1\\
  \bm h^{XU}: && w_R + w_T \geq & 1
\end{align*}
Consider the three constraints for $\bm h^Y,\bm h^Z,\bm h^{XU}$, which
mention only the variables $w_R, w_S, w_T$.  Any solution to these
three constraints can be immediately extended to a solution to all 5
constraints, by setting $w_B \geq \max(0, 1-w_R)$ and
$w_A \geq \max(0,1-w_T)$.  On the other hand, the three constraints on
$w_R, w_S, w_T$ assert that they form a fractional edge cover of a
triangle.  The fractional edge covering polytope of a triangle has
four extreme vertices,
%
%
$(0,1,1),(1,0,1),(1,1,0),(1/2,1/2,1/2)$.  It follows that the extreme
vertices of our polytope over all 5 variables are:
\begin{align*}
  &
    \begin{array}{|ccccc|} \hline
      w_R& w_S&w_T&w_A&w_B \\ \hline
      0 & 1 & 1 & 0 & 1 \\
      1 & 0 & 1 & 0 & 0 \\
      1 & 1 & 0 & 1 & 0 \\
      1/2 & 1/2 & 1/2 & 1/2 & 1/2 \\ \hline
      \end{array}
\end{align*}
Each of these vectors corresponds precisely to one of the four
inequalities that we listed in Example~\ref{ex:pods2016}, and,
conversely, any $\Sigma$-inequality of the
form~\eqref{eq:ex:pods2016:ineq} is dominated by some convex
combination of one of these four.

\section{Proof of~\eqref{eq:shannon:complicated}, and Failure of  the Copy Lemma}
\label{app:eq:shannon:complicated}

We start by proving Inequality~\eqref{eq:shannon:complicated}.  The
proof follows immediately from the following identity:
\begin{align*}
  - & I(X ; Y) + I(X ; Y |A)+I(X ; Y |B)+I(A:B)+  \\
    & + I(X ; Y |A')+I(A'; Y |X)+I(A'; X |Y) + 3I(A';AB|XY) \\
  = & I(A;B|A') + I(A;A'|Y) + I(A;A'|X) \\
    & + I(A;A'|BXY) + I(B;A'|Y) + I(B;A'|X) \\
    & + I(B;A'|AXY) + I(X;Y|BA') + I(X;Y|AA') \\
    & + I(X;A'|ABY) + I(Y;A'|AB) \\
\end{align*}

Next, we prove that the polymatroid in Fig.~\ref{fig:zhang:yeung:h}
does not satisfy the Copy Lemma.  Assuming otherwise, let $h'$ denote
the polymatroid over variables $X,Y,A,B,A',B'$.  Using only basic
Shannon inequalities, we derive a contradiction.  We will drop the
index $h'$ from $I_{h'}(\cdots)$ and write simply $I(\cdots)$.  By
assumption, the values $h'(\bm U)$ for all sets $\bm U$ that do not
contain both $A$ and $A'$, or both $B$ and $A'$ are known, for example
$h'(A'|XY) = h(A|XY) = 1$.  We also known the value $h'(\bm U)$ when
$\bm U$ contain $XY$, example
$h'(AA'XY) = h'(AA'|XY)+h'(XY) = h'(A|XY)+h'(A'|XY)+h(XY) =1 + 1 +
3=5$.  We proceed by examining the other sets where $A,A'$ or $B,B'$
co-occur, and start by showing $h'(AA') \leq 3$:
\begin{align*}
  I(A;A'|X) = & h'(AX)+h'(A'X) - h'(X) - h'(AA'X) \\
  = & 3+3-2-h'(AA'X)\geq 0
\end{align*}
and we derive $h'(AA'X)\leq 4$.  Similarly (replacing $X$ with $Y$) we
derive $h'(AA'Y) \leq 4$.  Finally, we have:
\begin{align*}
  I(X;&Y|AA') = \\
      = & h'(AA'X)+h'(AA'Y)-h'(AA'XY) - h'(AA') \geq 0
\end{align*}
and we derive that $h'(AA') \leq 3$.  We repeat the argument above by
replacing $A$ with $B$, and derive similarly that $h'(BA') \leq 3$.
Next, we show that $h'(ABA') \geq 5$, which follows from:
\begin{align*}
  I(XY;&A'|AB) = \\
  = & h'(ABXY)+h'(ABA')-h'(AB) - h'(ABA'XY) \\
  = & 4  +h'(ABA')-4-5 \geq 0
\end{align*}
thus $h'(ABA') \geq 5$.  (We also have $h'(ABA')\leq h'(ABA'XY)=5$,
hence $h'(ABA') = 5$, but the inequality suffices for us.)  Finally,
we derive a contradiction:
\begin{align*}
  I(A;B|A') = & h'(AA')+h'(BA')-h'(A')-h'(ABA')\\
  \leq & 3 + 3 - 2 - 5 = -1
\end{align*}

\section{Addendum to Theorem~\ref{th:chan:groups}}

\label{app:th:chan:groups}

We briefly sketch here the proof that, if an entropic function $\bm h$
satisfies a set of functional dependencies, then so do all
$\bm h^{(r)}$, for all $r \geq 0$.  For that we need to review the
main argument of the proof in~\cite{DBLP:journals/tit/ChanY02}.

Let $\bm h$ be an entropic function, realized by a probability
distribution $(R,p)$.  The first step is to ensure that the
probabilities $p(t)$ can be assumed to be rational numbers.  Assume
w.l.o.g. that $R$ is the support of $p$, then, by Lee's
result~\cite{DBLP:journals/tse/Lee87},
$h \models \bm U \rightarrow \bm V$, iff
$R \models \bm U \rightarrow \bm V$.  Consider now any sequence of
probability distributions on $R$, $p^{(k)}: R \rightarrow [0,1]$, of
rational numbers, such that $\lim_k p^{(k)} = p$.  Then $p^{(k)}$, and
its entropic vector $\bm h^{(k)}$, continue to satisfy the same FDs as
$R$ and, thus, the same FDs as $\bm h$.  Since $\bm h^{(k)}$ can be
arbitrarily close to $\bm h$, it suffices to prove that the theorem
holds for an entropic vector $\bm h$ realized by a probability
distribution $(R,p)$ where the probabilities are rational numbers.
Assume they have a common denominator $q > 0$, and let $N = |R|$.

From here on, we follow Chan and Yeung's
proof~\cite{DBLP:journals/tit/ChanY02}.  For each
$r = q, 2q, 3q, \ldots$ define the following $r \times n$ matrix
$\bm M_r = (m_{\rho i})_{\rho=1,r;i=1,n}$. Its rows are copies of the
tuples in $R$, where each tuple $\bm x\in R$ occurs $r\cdot p(\bm x)$
times in the matrix $\bm M_r$.  Intuitively, $\bm M_r$ can be viewed
as a relation with $n$ attributes and $r$ tuples, including
duplicates, whose uniform probability distribution has the same
entropic vector $h$ as $(R,p)$.  Let $G$ be the symmetric group $S_r$,
i.e. the group of permutation on the set $\set{1,2,\ldots,r}$; one
should think of $G$ as the group of permutations on the rows of
$\bm M_r$. For each $i = 1,\ldots, n$, let $G_i$ the subgroup that
leaves the column $i$ invariant, in other words:
\begin{align*}
  G_i = & \setof{\sigma \in G}{m_{\sigma(\rho),i}=m_{\rho,i}, \forall \rho = 1,r}
\end{align*}
Denoting similarly $G_\alpha$ the subgroup of permutations that leave
the set of columns $\alpha \subseteq [n]$ invariant, one can check
that $G_\alpha = \bigcap_{i \in \alpha} G_i$.  Let $\bm h^{(r)}$ be
the entropy of the uniform probability distribution on the relational
instance $\setof{(aG_1,\ldots, aG_n)}{a \in G}$.  Using a
combinatorial argument, Chan and
Yeung~\cite{DBLP:journals/tit/ChanY02} prove that
$\lim_{r \rightarrow \infty} \frac{1}{r} \bm h^{(r)} = \bm h$.  We
will not repeat that argument here, but make the additional
observation that, if $R$ satisfies the FD $\bm U \rightarrow \bm V$,
then $G_{\bm U} \subseteq G_{\bm V}$, which implies that $\bm h^{(r)}$
also satisfies the same FD.

\section{Proof of Equation~\eqref{eq:normal:coefficients}}

\label{app:eq:normal:coefficients}

M\"obius inversion formula states that, if
$f, g : 2^{\bm X} \rightarrow \R$ are two set functions, and one of
the identities below holds, then so does the other:
\begin{align}
  f(\bm U) = & \sum_{\bm V \subseteq \bm U}g(\bm V) & g(\bm U) = & \sum_{\bm V \subseteq \bm U}(-1)^{|\bm U-\bm V|}f(\bm V) 
\label{eq:mobius:inversion}
\end{align}
It is immediate to derive that the following identities are also
equivalent:\footnote{Define $h(\bm U) \defeq f(\bm X-\bm U)$
  then~\eqref{eq:mobius:inversion} becomes:
  \begin{align*}
    h(\bm U) = &\sum_{\bm V \subseteq \bm X-\bm U}g(\bm V) & g(\bm U) = & \sum_{\bm V \subseteq \bm U}(-1)^{|\bm U-\bm V|}h(\bm X - \bm V) 
  \end{align*}
  The claim follows by replacing $\bm V$ with $\bm X-\bm V$ in the
  second equation, then renaming $h$ to $f$.}
\begin{align*}
  f(\bm U) = & \sum_{\bm V \subseteq \bm X-\bm U}g(\bm V) & g(\bm U) = & \sum_{\bm V: \bm X-\bm U \subseteq \bm V}(-1)^{|\bm V\cap\bm U|}f(\bm V)
\end{align*}

To prove Equation~\eqref{eq:normal:coefficients}, we use
Equation~\eqref{eq:coefficients:normal}, and the fact that
$h(\bm X) = \sum_{\bm V: \bm V \subseteq \bm X} a_{\bm V}$:
\begin{align*}
  h(\bm U) = & \sum_{\bm V: \bm V \cap \bm U \neq \emptyset} a_{\bm V} =  h(\bm X) - \sum_{\bm V: \bm V \subseteq \bm X-\bm U} a_{\bm V} \\
  h(\bm X|\bm U) = & \sum_{\bm V: \bm V \subseteq \bm X-\bm U} a_{\bm V}
\end{align*}
M\"obius' inversion formula implies:
\begin{align*}
  a_{\bm U} = & \sum_{\bm V: \bm X-\bm U\subseteq \bm V}(-1)^{|\bm V\cap \bm U|}h(\bm X|\bm V)
\end{align*}
If $\bm X-\bm U\subseteq \bm V$, then we can write $\bm V$ uniquely as
$(\bm X-\bm U) \cup \bm V_0$, where $\bm V_0\subseteq \bm U$.  After
renaming $\bm V_0$ to $\bm V$ we derive:
\begin{align*}
  a_{\bm U} = & \sum_{\bm V \subseteq \bm  U}(-1)^{|\bm V|}h(\bm X|(\bm X-\bm U)\cup\bm V)\\
= & \sum_{\bm V \subseteq \bm  U}(-1)^{|\bm V|}h(\bm X)-\sum_{\bm V \subseteq \bm  U}(-1)^{|\bm V|}h((\bm X-\bm U)\cup\bm V)\\
= & - \sum_{\bm V \subseteq \bm  U}(-1)^{|\bm V|}h((\bm X-\bm U)\cup\bm V)\\
= & - \sum_{\bm V \subseteq \bm  U}(-1)^{|\bm V|}\left(h((\bm X-\bm U)\cup\bm V)-h(\bm X-\bm U)\right)\\
  = & - \sum_{\bm V \subseteq \bm  U}(-1)^{|\bm V|}h(\bm V|\bm X-\bm U)
\end{align*}
We used twice the fact that
$\sum_{\bm V \subseteq \bm U}(-1)^{|\bm V|} = 0$ when
$\bm U \neq \emptyset$.  This completes the proof.

\section{Proof of Lemma~\ref{lemma:normalization}}

\label{app:lemma:normalization}

To prove the lemma we need to establish a simple fact:
\begin{prop} \label{prop:max:normal}
  Let $\alpha_1, \ldots, \alpha_n \geq 0$ be non-negative real
  numbers.  Define the following set function
  $\bm h: 2^{\bm X} \rightarrow \R$:
  \begin{align*}
    h(\bm U) \defeq & \max_{i: X_i \in \bm U} \alpha_i
  \end{align*}
  Then $\bm h$ is a normal polymatroid.
\end{prop}

\begin{proof}
  Set $\alpha_0 \defeq 0$ and assume w.l.o.g. that
  $\alpha_0 \leq \alpha_1 \leq \alpha_2 \leq \cdots \leq \alpha_n$.
  Define $\delta_i \defeq \alpha_i - \alpha_{i-1}$ for $i=1,n$.  We
  prove:
  \begin{align}
    \bm h = & \sum_{i=1,n} \delta_i \bm h^{\bm X_{[i+1:n]}}\label{eq:h:delta}
  \end{align}
  If $k$ is the largest index s.t. $X_k \in \bm U$, then:
  \begin{align*}
    \sum_{i=1,n} \delta_i \bm h^{\bm X_{[i+1:n]}}(\bm U) = & \sum_{i=1,k}\delta_i=\alpha_k = h(\bm U)
  \end{align*}
  which proves~\eqref{eq:h:delta}.  Since
  $\delta_i \geq 0, \forall i$, $\bm h$ is a normal polymatroid.
\end{proof}

\begin{proof} (of Lemma~\ref{lemma:normalization}) We prove the claim
  by induction on the number of variables $n = |\bm X|$.  When $n=1$
  then the claim holds vacuously, so assume $n \geq 2$.  We will use
  the following identity:
  \begin{align*}
\forall \bm U \subseteq \bm X-\set{X_n}:&&& h(\bm U) = h_1(\bm U)+h_2(\bm U)\\
\text{where}:  &&&  h_1(\bm U) \defeq h(\bm U|X_n),\ \ h_2(\bm U) \defeq I(\bm U; X_n)
  \end{align*}
  Since $\bm h_1$, is a polymatroid in $n-1$ variables, by induction
  hypothesis we obtain a normal polymatroid $\bm h_1'$ such that the
  pair $\bm h_1, \bm h_1'$ satisfies the properties (a),(b),(c) of
  Lemma~\ref{lemma:normalization}.  Define:

  \begin{align*}
    \forall \bm U \subseteq \bm X-\set{X_n}:&&   h_2'(\bm U) \defeq & \max_{X_i \in \bm U} I(X_i; X_n)
  \end{align*}
  $\bm h_2'$ is a normal polymatroid, by Prop.~\ref{prop:max:normal}.
  The pair $\bm h_2, \bm h_2'$ satisfies properties (a),(c) of
  Lemma~\ref{lemma:normalization}: condition
  (a) follows from Eq.~\eqref{eq:I:chain:condition} in
  Prop.~\ref{prop:I:properties}:
  \begin{align}
    h_2'(\bm U) = \max_{X_i \in \bm U} I(X_i; X_n) \leq & I(\bm U; X_n) = h_2(\bm U) \label{eq:h2uiuxn}
  \end{align}
  and condition (c), $h_2(X_i)=h_2'(X_i)$, holds by definition.  (They
  do not satify contidion (b), i.e.
  $h_2(\bm X-\set{X_n})\neq h_2'(\bm X-\set{X_n})$).  We define:

  \begin{align*}
    a \defeq & \min_{i=1,n-1} h(X_n|X_i) = \min_{i=1,n-1} (h(X_n) - I(X_i;X_n)) = h(X_n) - h_2'(\bm X-\set{X_n})
  \end{align*}
  and observe that $a \geq 0$.

  Since both $\bm h_1'$ and $\bm h_2'$ are normal polymatroids, we can
  write them as:
  \begin{align*}
    \bm h_1' = & \sum_{\bm V \subseteq \bm X-\set{X_n},\bm V\neq\emptyset}b_{\bm V} \bm h^{\bm V}&
    \bm h_2' = & \sum_{\bm V \subseteq \bm X-\set{X_n},\bm V\neq\emptyset}c_{\bm V} \bm h^{\bm V}
  \end{align*}
  where $b_{\bm V}, c_{\bm V}$ are non-negative coefficients.  Define
  $\bm h'$ as:
  \begin{align*}
    \bm h' \defeq & \sum_{\bm V \subseteq \bm X-\set{X_n},\bm V\neq\emptyset}b_{\bm V} \bm h^{\bm V}+ \sum_{\bm V \subseteq \bm X-\set{X_n},\bm V\neq\emptyset}c_{\bm V} \bm h^{\bm V\cup\set{X_n}}+a\bm h^{\set{X_n}}
  \end{align*}
  We claim that $\bm h'$ satisfies the conditions of the lemma.
  Obviously, $\bm h'$ is a normal polymatroid, it remains to check
  conditions (a),(b),(c).  First, we use directly the defintion of
  $\bm h^{\bm V}$ in Eq.~\eqref{eq:step:function:alt} and derive the
  following identities, for all
  $\bm U, \bm V \subseteq \bm X - \set{X_n}$, $\bm V \neq \emptyset$:
  \begin{align*}
    h^{\bm V}(\bm U \cup \set{X_n}) = & h^{\bm V\cup \set{X_n}}(\bm U)=h^{\bm V}(\bm U)& h^{\bm V\cup \set{X_n}}(\bm U\cup \set{X_n})= & h^{\bm V}(\bm X - \set{X_n}) = 1.\\
h^{\set{X_n}}(\bm U)= & 0 & h^{\set{X_n}}(\bm U\cup \set{X_n})= & 1
  \end{align*}
  These identities imply the following, for $\bm U \subseteq \bm X-\set{X_n}$:
  \begin{align}
    h'(\bm U) = & h_1'(\bm U) + h_2'(\bm U)\nonumber\\
    h'(\bm U \cup \set{X_n}) = & h_1'(\bm U) + h_2'(\bm X-\set{X_n}) + a  =  h_1'(\bm U) + h(X_n)\label{eq:local:xyz}
  \end{align}
  We can check condition (a), $\bm h' \leq \bm h$, by considering two
  cases.  First, when $\bm U \subseteq \bm X - \set{X_n}$, then:
  \begin{align}
    h'(\bm U) = &  h_1'(\bm U) + h_2'(\bm U)  \leq  h_1(\bm U)+h_2(\bm U) = h(\bm U|X_n) + I(\bm U; X_n) = h(\bm U) \label{eq:hph1ph2:1}
  \end{align}
  We used the fact that both $\bm h_1'$ and $\bm h_2'$ satisfy
  condition (a) (see Eq.~\eqref{eq:h2uiuxn}).  Second, when the
  argument includes $X_n$, we have:
  \begin{align}
    h'(\bm U\cup \set{X_n}) = &  h_1'(\bm U) + h(X_n) \leq h(\bm U|X_n) + h(X_n) =  h(\bm U\cup\set{X_n}) \label{eq:hph1ph2:2}
  \end{align}
  Equations~\eqref{eq:hph1ph2:1} and~\eqref{eq:hph1ph2:2} imply
  $\bm h' \leq \bm h$, which proves condition (a).  Next, we check
  condition (b), and here we use~\eqref{eq:local:xyz}, and the fact
  that $\bm h_1'$ satisfies condition (b):
  \begin{align*}
    h'(\bm X) = & h_1'(\bm X-\set{X_n})+h(X_n) =  h_1(\bm X-\set{X_n})+h(X_n)=h(\bm X-\set{X_n}|X_n) + h(X_n) = h(\bm X)
  \end{align*}
  Finally we check condition (c), $h'(X_i)=h(X_i)$.  For $i<n$ we use
  the fact that $\bm h_1'$, $\bm h_2'$ satisfy (c):
  \begin{align*}
    h'(X_i) = & h_1'(X_i) + h_2'(X_i) = h_1(X_i) + h_2(X_i) =  h(X_i|X_n)+I(X_i;X_n)=h(X_i)
  \end{align*}
  When $i=n$ we use~\eqref{eq:local:xyz} and dervie
  $h'(X_n) = h_1'(\emptyset) + h(X_n)$.
\end{proof}

\end{document}